\documentclass[12pt]{article}
\usepackage{mathrsfs}
\usepackage{amsbsy}
\usepackage{amsmath,bm}
\usepackage{amssymb}
\usepackage{multicol}
\usepackage{multirow}
\usepackage{booktabs}
\usepackage{times,epsfig,verbatim,lscape}
\usepackage{amsthm}
\usepackage{graphicx}
\usepackage{float}
\usepackage{subfigure}
\usepackage{color}
\usepackage{appendix}
\usepackage[round]{natbib}
\usepackage{rotating}
\usepackage{threeparttable}
\usepackage[textwidth=14.5cm]{geometry}
\usepackage{blindtext}
\usepackage{titlesec}
       \titleformat{\section}
             {\normalfont\large\bfseries}{\thesection}{11pt}{\large}
       \titleformat{\subsection}
             {\normalfont\bfseries}{\thesubsection}{11pt}
\allowdisplaybreaks
\renewcommand{\baselinestretch} {1.2}
\makeatletter 
\def\bbeta{{\boldsymbol\beta}}
\def\bmu{{\boldsymbol\mu}}
\def\singlespace{\def\baselinestretch{1}\@normalsize}

\def\non{\nonumber}
\def\bse{\begin{eqnarray*}}
\def\ese{\end{eqnarray*}}
\def\be{\begin{eqnarray}}
\def\ee{\end{eqnarray}}
\def\bsq{\begin{equation*}}
\def\esq{\end{equation*}}
\def\bq{\begin{equation}}
\def\eq{\end{equation}}

\def\overd{\stackrel{d}\rightarrow}
\def\hat{\widehat}
\numberwithin{equation}{section}

\renewcommand{\theequation} {\arabic{section}.\arabic{equation}}

\newtheorem{condition}{\underline{\bf Assumption}}

\newtheorem{thm}{\underline{\bf Theorem}}[section]

\newtheorem{lem}{\underline{\bf Lemma}}[section]

\newtheorem{remark}{\underline{\bf Remark}}
\def\tit.arg{\textbf{Jackknife Model Averaging for Composite Quantile Regression}}
\DeclareMathOperator*{\argmin}{argmin}
\def\shorttitle.arg{
JMA for Composite Quantile Regression
}

\def\key.arg{Asymptotic optimality; Composite Quantile Regression; Model averaging.
}

\def\abst.arg{Model averaging considers the model uncertainty and is an alternative to model selection. In this paper, we propose a  frequentist model averaging estimator for composite quantile regressions.
In recent years, research on these topics has been added as a separate method, but no study has investigated them in combination.
We apply a delete-one cross-validation method to estimate the model weights,
and prove that the jackknife model averaging estimator is asymptotically optimal in terms of minimizing out-of-sample composite final prediction error. Simulations are conducted to demonstrate the good finite sample properties of our estimator and compare it with commonly used model selection and averaging methods. The proposed method is applied to the analysis of the stock returns data and the wage data and performs well.}

\def\author.arg{Miaomiao Wang$^{a,b,c,}$\footnote{Email: wangmm@amss.ac.cn}, Guohua Zou$^{d}$
\vskip 0.2cm
{\it \small $^{a}$School of Chinese Pharmacy, Beijing University of Chinese Medicine, Beijing, China. \\
$^{b}$Academy of Mathematics and Systems Science, Chinese Academy of Sciences, Beijing, China. \\
$^{c}$University of the Chinese Academy of Sciences, Beijing, China. \\
$^{d}$School of Mathematical Sciences, Capital Normal University, Beijing, China.
}\\
\medskip
}
\allowdisplaybreaks
\begin{document}
\begin{sloppypar}
\thispagestyle{empty}

\begin{center}
{\Large \tit.arg}

\vskip 3mm

\author.arg
\end{center}

\vskip 3mm \centerline{\small ABSTRACT} \abst.arg
\\


\clearpage\pagebreak\newpage
\pagenumbering{arabic}
\setcounter{page}{1}

\section{Introduction}\label{sec1}
\setcounter{equation}{0}
\baselineskip=24pt
Quantile regression (QR) incepted by \citep{koenker1978quantile} is now an indispensable tool for modeling
heterogeneous data in different areas including econometrics, finance, microarrays, medical and agricultural studies. More details refer to \cite{koenker2005quantile} and \cite{Yu2010Quantile}.

As it does in least squares regression (LSR), variable selection plays an important role in QR studies. Some literatures suggest that the model selection method commonly used in LSR can be applied to QR.
\cite{machado1993robust} suggested a generalization of Schwarz information criterion (SIC) for general M-estimators which are applicable to QR by replacing the squared loss function with the ``check function" introduced by \citep{koenker1978quantile}.
\cite{koenker2005quantile} suggested selecting based on minimizing variant Akaike information criterion (AIC) which may have superior performance for prediction although the AIC overestimates the model dimension.
The asymptotic property of such penalized quantile regression estimators fits well within the convergence framework introduced by \cite{Knight2000Asymptotics}.
Another research focus is on developing QR penalty methods.
\cite{Koenker2004Quantile} employed $L_{1}$ regularization method, which uses the
sum of the absolute values of the coefficients as the penalty, to shrink individual effects towards a common value.
\cite{Li2008L1} evolved an efficient algorithm and an estimator for computing the entire solution path and the effective dimension of the $L_{1}$ penalized QR, respectively.
Moreover, \cite{wu2009quantile} presented penalized quantile regressions with the SCAD and adaptive-LASSO penalties and established the asymptotic properties.
Recently, many researches on this issue have emerged; see for
example, \cite{Alhamzawi2012Bayesian}, \cite{Li2010Bayesian}, \cite{Wu2013Penalized}, \cite{Yuan2010Bayesian}, among others.

Regardless of the data-driven approach, the search for the best model identifies the presence of multiple candidate models, which means that the level of uncertainty associated with model selection is often ignored when reporting accurate estimators. This is a recognized limitation of model selection and has been extensively studied in the literature; see for example, \cite{claeskens2016model} and \cite{leeb2006conditional}.
One method of combining uncertainty in model selection is model averaging, where the estimation of unknown parameters is based on a set of weighted models rather than a single model. There are two branches:
Bayesian model averaging (BMA) \cite{hoeting1999bayesian} and frequentist model averaging (FMA) \cite{hansen2007least}.
When using BMA, the uncertainty of the model is considered by setting the prior probability for each candidate model.
There are some problems with this approach, such as how to set a priori probabilities and how to handle the conflict with each other. In contrast, the FMA method does not require a priori, and the corresponding estimator is completely data-driven. Therefore, the FMA approach has received significant attention over the past decade.
After the pioneering contributions of \cite{hansen2007least} and \cite{hjort2003averaging}, the FMA research results have appeared frequently in just a few years.
Much of the work in this literature is about finding criteria for estimating model weights.
\cite{Buckland1997Model}, \cite{Burnhan2002Model} and \cite{hjort2006focused} constructed model averaging weights based on the values of model selection criterion scores.
\cite{yang2001adaptive} proposed an algorithm called adaptive regression by mixing
to handle multiple candidate error distribution.
\cite{hansen2007least} introduced the Mallows' criterion and \cite{hansen2008joe, wan2010least} also contributed to this.
Other major weight selection methods developed in recent years include
MSE minimization \cite{liang2011optimal, wan2014ijf, zhang2014mixed}, cross-validation (CV) procedures \citep{ando2014model, hansen2012jackknife, lu2015jackknife, zhang2013jackknife}, and minimization of Kullback-Leibler type measures \citep{zhang2016optimal,zhang2015kullback}. A major trend throughout the literature is that for prediction, a combination of competitive models is usually superior to a single model from a sampling theory perspective.

In the existing FMA studies, \cite{lu2015jackknife} and \cite{shanyang2009sinica} are the only two studies that emphasize QR.
\cite{lu2015jackknife} allocated weights to QR models by a leave-one-out CV method, and demonstrated that this method yields an asymptotically optimal FMA estimator in terms of minimizing final prediction error . \cite{shanyang2009sinica} extended the ARM algorithm of \cite{yang2001adaptive} to QR. As we all know, their estimators are based on one quantile.
The composite quantile regression (CQR) method, proposed by \cite{zou2008quantile}, is a useful extension of the conventional quantile regression method. This method can simultaneously consider multiple quantile regression models, and then can derive a more efficient estimator. Recently, many papers considered the applications of the CQR method. \cite{Jiang2012Single} extended the CQR method to a single-index model and \cite{Jiang2012Oracle} studied the model selection for nonlinear models with the weighted CQR method. \cite{Kai2010Local} proposed a local CQR smoothing method for nonparametric regression models. \cite{Zhao2015Empirical} studied empirical likelihood inferences via CQR, and so on. However, as far as known, the model averaging for CQR models have not been considered in the literature.

Taking this issue into account, we propose a model averaging estimator for composite quantile regression. Our results about inference of parameter estimates may be treated as the counterpart to the conclusions in \cite{lu2015jackknife} who focused on linear quantile regression. However, the implementation process are more challenging due to the non-smooth nature of the loss function, coupled with the fact that unlike quantile regression, we consider multiple quantile and the number of quantiles may be infinite.
Because of these, it is very tricky to handle the Lagrange remainder of Taylor expansion. In the proof asymptotic normality of estimators, we apply the Lindeberg-Feller central limit theorem and verifying its conditions is also an important step. Adopting these inference conclusions, we prove that the resultant FMA estimator, using a composite leave-one-out CV method for assigning model weights, possesses an asymptotically optimal property.
Due to the weight selection method combining the strength across multiple quantile,
our proofs of optimality are based on very different technical assumptions.
The results are derived using empirical process theory and depend on the conditions of the covariate dimension and the growth rate of the combined quantile relative to the sample size.

The reminder of the article is organized as follows. Section \ref{sec:MF} describes the model framework and introduces the model averaging estimator for CQR. In Section \ref{sec:mainresults}, we present the weight selection criterion and prove that the FMA estimator based on the proposed weight choice method has an asymptotically optimal property. Section \ref{sec:simu} provides an empirical comparison of the proposed FMA estimator with a host of other estimators including the method in \cite{lu2015jackknife}. In Section \ref{sec:Real}, we apply the proposed method to real example. Some concluding remarks are contained in Section \ref{sec:con}. Proofs of technical results are given in the Appendixes.

\section{Model framework and estimator}\label{sec:MF}
\subsection{Composite quantile regression}
Let $y$ be a scalar dependent variable, $\boldsymbol{x}=(x_{1},x_{2},\ldots,x_{p})^{T}$ be a $p$-dimensional vector of covariates. Consider the following linear model
\begin{eqnarray}
\label{eq2.1}
y=\boldsymbol{x}^{T}\bbeta+\varepsilon,
\end{eqnarray}
where $\bbeta=(\beta_{1},\cdots,\beta_{p})^{T}$ is $p\times1$ unknown regression coefficient vector and $\varepsilon$ is the model error.
Let
 $\mathcal{D}_{n}=\{(\boldsymbol{x}_{i},y_{i}),i=1,\ldots,n\}$ be independent and identically distributed (IID) copies of $(\boldsymbol{x},y)$, where $\boldsymbol{x}_{i}=(x_{i1},x_{i2},\ldots,x_{ip})^{T}$. Suppose that the random sample $(\boldsymbol{x}_{i},y_{i}),i=1,\ldots,n,$ satisfies the following linear regression model:
\begin{eqnarray}
y_{i}=\mu_{i}+\varepsilon_{i},\ i=1,\ldots,n.
\end{eqnarray}
where $\mu_{i}=\boldsymbol{x}_{i}^{T}\bbeta$, $\varepsilon_{i}$ is the model error that satisfies $\mathrm{P}(\varepsilon_{i}\leq b_{k}|\boldsymbol{x}_{i})=\tau_{k}$ and
$b_{k}$ is the $100\tau_{k}$ \% quantile of $\varepsilon$, $k = 1,\cdots ,K$, where $K$ is the number of quantiles. As in \cite{zou2008quantile}, the CQR estimators of $\bbeta$ and $b_{k}$, $k = 1,\cdots ,K$ can be given by solving
\begin{eqnarray}
\label{eq2.2}
(\hat{b}_{1},\cdots,\hat{b}_{K},\hat{\bbeta})=\operatorname*{\argmin\limits}_{b_{1},\cdots,b_{K},\bbeta}
\sum_{k=1}^{K}\left\{\sum_{i=1}^{n}\rho_{\tau_{k}}(y_{i}-b_{k}-\boldsymbol{x}_{i}^{T}\bbeta)\right\},
\end{eqnarray}
where $\rho_{\tau}(\varepsilon)=\varepsilon[\tau-\boldsymbol{1}\{\varepsilon\leq0\}]$ and $\boldsymbol{1}\{\cdot\}$ denotes the usual indicator function.
Note that we consider multiple quantile regression models but the
coeficients are the same across different quantiles.

\subsection{Model averaging estimator}
Write the $m^{th}$ candidate model as
 \begin{eqnarray}
y_{i}=\mu_{i(m)}+\varepsilon_{i(m)}, \ i=1,\ldots,n,
\end{eqnarray}
where $\mu_{i(m)}=\boldsymbol{x}_{i(m)}^{T}\bbeta_{(m)}=\sum_{j=1}^{\tilde{k}_{m}}\theta_{j(m)}x_{ij(m)}$, $\tilde{k}_{m}$ denotes the number of covariates which are the same across different quantile regression models, $\bbeta_{(m)}=(\beta_{1(m)},\ldots,\beta_{\tilde{k}_{m}(m)})^{T}, \boldsymbol{x}_{i(m)}=(x_{i1(m)},\ldots,x_{i\tilde{k}_{m}(m)})^{T}$, $x_{ij(m)}$ is a covariate and $\theta_{j(m)}$ is the corresponding coefficients, $j=1,\ldots,\tilde{k}_{m}$, and
$\varepsilon_{i(m)}=y_{i}-\sum_{j=1}^{\tilde{k}_{m}}\beta_{j(m)}x_{ij(m)}$ satisfies $\mathrm{P}(\varepsilon_{i(m)}\leq b_{k(m)}|\boldsymbol{x}_{i(m)})=\tau_{k}$,
where $b_{k(m)}$ is the $100\tau_{k}$ \% quantile of the error $\varepsilon_{(m)}$ under the $m^{th}$ model, $k = 1,\cdots ,K$.
The CQR estimators of $\bbeta_{(m)}$ and $b_{k(m)}$, $k = 1,\cdots ,K$ in the above model can be given by solving
\begin{eqnarray}
(\hat{b}_{1(m)},\cdots,\hat{b}_{K(m)},\hat{\bbeta}_{(m)})
&=&\operatorname*{\argmin\limits}_{b_{1(m)},\cdots,b_{K(m)},\bbeta_{(m)}}Q_{n(m)}\left\{\left(b_{1},\cdots,b_{K},\bbeta^{T}\right)^{T}\right\}\non\\
&=&\operatorname*{\argmin\limits}_{b_{1(m)},\cdots,b_{K(m)},\bbeta_{(m)}}
\sum_{k=1}^{K}\left\{\sum_{i=1}^{n}\rho_{\tau_{k}}(y_{i}-b_{k(m)}-\boldsymbol{x}_{i(m)}^{T}\bbeta_{(m)})\right\}.
\end{eqnarray}
Let $\boldsymbol{w}\equiv (w_{1},\ldots,w_{M})^{T}$ be a  weight vector in the unit simplex of $\mathbb{R}^{M}$ and $\mathcal{W}\equiv\{\boldsymbol{w}\in[0,1]^{M}: \sum_{m=1}^{M}w_{m}=1\}$. Let $\bmu=\left(\mu_{1},\cdots,\mu_{n}\right)^{T}$. The model averaging estimator of
 $(b_{1},\cdots,b_{K},\bmu^{T})^{T}$ is
\begin{eqnarray}\label{eq:avg2}
&&\left\{\hat{b}_{1}(\boldsymbol{w}),\cdots,\hat{b}_{K}(\boldsymbol{w}),\hat{\bmu}(\boldsymbol{w})^{T}\right\}^{T}\\\non
&&=\left(\sum_{m=1}^{M}w_{m}\hat{b}_{1,m},\cdots,\sum_{m=1}^{M}w_{m}\hat{b}_{K,m},\sum_{m=1}^{M}w_{m}\hat{\bmu}_{i(m)}^{T}\right)^{T},
\end{eqnarray}
where $\hat{\bmu}_{i(m)}=\left(\hat{\mu}_{1(m)},\cdots,\hat{\mu}_{n(m)}\right)^{T}=\left(\boldsymbol{x}_{1(m)}^{T}\hat{\bbeta}_{(m)},\cdots,\boldsymbol{x}_{n(m)}^{T}\hat{\bbeta}_{(m)}\right)^{T}$.
\section{Weight selection and asymptotic property}\label{sec:mainresults}
\subsection{Leave-one-out CV criterion}
We propose selecting $\boldsymbol{w}$ by the jackknife or leave-one-out CV criterion described as follows. For $m=1,\ldots,M$, let $\left(\hat{b}_{1i(m)},\cdots,\hat{b}_{Ki(m)},\hat{\bbeta}_{i(m)}^{T}\right)^{T}$ be the jackknife estimator of $(b_{1},\cdots,b_{K},\bbeta_{(m)}^{T})^{T}$ in model $m$ with the $i^{th}$ observation deleted. Consider the leave-one-out CV criterion
\begin{eqnarray}
\label{cvw}
CV_{n}(\boldsymbol{w})=\frac{1}{nK}\sum_{i=1}^{n}\left\{\sum_{k=1}^{K}\rho_{\tau_{k}}\left(y_{i}-\sum_{m=1}^{M}w_{m}\hat{b}_{ki(m)}-\sum_{m=1}^{M}w_{m}\boldsymbol{x}_{i(m)}^{T}\hat{\bbeta}_{i(m)}\right)\right\}.
\end{eqnarray}
The jackknife weight vector $\widehat{\boldsymbol{w}}=(\widehat{w}_{1},\ldots,\widehat{w}_{M})^{T}$ is obtained by choosing $\boldsymbol{w}\in\mathcal{W}$ such that
\begin{eqnarray}
\widehat{\boldsymbol{w}}=\operatorname*{\argmin\limits}_{\boldsymbol{w}\in\mathcal{W}}CV_{n}(\boldsymbol{w}).
\end{eqnarray}
The calculation of the weight vector $\widehat{\boldsymbol{w}}$ is a straightforward linear programming problem in quantile regression literature.
Substituting $\hat{\boldsymbol{w}}$ for $\boldsymbol{w}$ in (\ref{eq:avg2}) results in the following jackknife model averaging (JMA) estimator of $(b_{1},\cdots,b_{K},\bmu^{T})^{T}$ for CQR:
\begin{eqnarray}
&&\left\{\hat{b}_{1}(\boldsymbol{\hat{w}}),\cdots,\hat{b}_{K}(\boldsymbol{\hat{w}}),\hat{\bmu}(\boldsymbol{\hat{w}})^{T}\right\}^{T}\\\non
&&=\left(\sum_{m=1}^{M}\hat{w}_{m}\hat{b}_{1,m},\cdots,\sum_{m=1}^{M}\hat{w}_{m}\hat{b}_{K,m},\sum_{m=1}^{M}\hat{w}_{m}\hat{\bmu}_{i(m)}^{T}\right)^{T},
\end{eqnarray}
\subsection{Asymptotic property of estimator}
This section is devoted to an investigation of the theoretical properties of the JMA. Specifically, we show that jackknife weight vector $\widehat{\boldsymbol{w}}$ is asymptotically optimal in the sense of minimizing the following out-of-sample composite quantile prediction error (CPE):
\begin{eqnarray}
\mathrm{CPE}_{n}(\boldsymbol{w})
=\mathrm{E}\left\{\frac{1}{K}\sum_{k=1}^{K}\rho_{\tau_{k}}\left(y-\sum_{m=1}^{M}w_{m}\hat{b}_{k,m}-\sum_{m=1}^{M}w_{m}\boldsymbol{x}_{m}^{T}\hat{\bbeta}_{(m)}\right)\Big|\mathcal{D}_{n}\right\},
\end{eqnarray}
where $(y,\boldsymbol{x})$ is an independent copy of $(y_{i},\boldsymbol{x}_{i})$, $\boldsymbol{x}_{m}=\left(x_{1(m)},\ldots,x_{\tilde{k}_{m}(m)}\right)^{T}$ and $\mathcal{D}_{n}=\left\{(\boldsymbol{x}_{i},y_{i}),i=1,\ldots,n\right\},$ $x_{1(m)},\ldots,x_{\tilde{k}_{m}(m)}$ are variables in $\boldsymbol{x}$ that correspond to the $\tilde{k}_{m}$ regressors in the $m^{th}$ model.

Following the notations of \cite{lu2015jackknife}, let $f(\cdot\mid \boldsymbol{x}_{i})$ and $F(\cdot\mid \boldsymbol{x}_{i})$ denote the conditional probability density function (PDF) and cumulative distribution function (CDF) of $\varepsilon_{i}$ given $\boldsymbol{x}_{i}$ respectively, and $f_{y\mid \boldsymbol{x}}(\cdot\mid \boldsymbol{x}_{i})$ the conditional PDF of $y_{i}$ given $\boldsymbol{x}_{i}$. Consider the following pseudo-true parameter
\begin{eqnarray} \label{problem:2}
(b_{1(m)}^{*},\cdots,b_{K(m)}^{*},\bbeta_{(m)}^{*T})^{T}=\operatorname*{\argmin\limits}_{b_{1(m)},\cdots,b_{K(m)},\bbeta_{(m)}}
E\left\{\sum_{k=1}^{K}\rho_{\tau_{k}}\left(y_{i}-b_{k(m)}-\boldsymbol{x}_{i(m)}^{T}\bbeta_{(m)}\right)\right\}.
\end{eqnarray}
Then $$u_{ki(m)}^{b}=\mu_{i}+b_{k}-\left(1,\boldsymbol{x}_{i(m)}^{T}\right)\left(b_{k(m)}^{*},\bbeta_{(m)}^{*T}\right)^{T}$$
is the approximation bias for the $m^{th}$ candidate model when the quantile is $\tau_{k}$.

\begin{thm}\label{th3.1}
Suppose Assumptions \ref{a1}-\ref{a3} hold and $MK^{2}\bar{k}^{2}\log n/n\rightarrow0$. Then $\widehat{\boldsymbol{w}}$ is asymptotically optimal in the sense that
\begin{eqnarray}
\frac{\mathrm{CPE}(\widehat{\boldsymbol{w}})}{\inf\limits_{\boldsymbol{w}\in\mathcal{W}}\mathrm{CPE}(\boldsymbol{w})}=1+o_{p}(1).
\end{eqnarray}

\end{thm}
\begin{proof}
See \ref{sec:th3.1}.
\end{proof}

\begin{remark}
Theorem \ref{th3.1} implies that the out-of-sample composite final prediction error
produced by the FMA estimator using jackknife weight vector is asymptotically equal to the prediction error of the FMA estimator using the infeasible optimal weight vector. Similar conclusion as in \cite{lu2015jackknife}, who considered one quantile.
\end{remark}
\section{Simulation studies}\label{sec:simu}
In this section, we conduct some simulations to study the finite sample performance of the proposed method (labeled MCV$_{c}$).
\subsection{Methods for comparison}
We consider jackknife model averaging (JMA) method proposed in \cite{lu2015jackknife} (labeled MCV$_{0}$) for comparison. For each $\tau_{k}$, $k=1,\cdots,K$, they provided the jackknife model averaging estimator. The estimator uses leave-one-out cross-validation to choose the weight $\boldsymbol{w}_{k}=(w_{1,k},\cdots,w_{M,k})^{T}$, then applying the estimator to predict. In detail, let
\be
\left(\tilde{b}_{k(m)},\hat{\bbeta}_{k(m)}^{T}\right)^{T}
&=&\operatorname*{\argmin\limits}_{b_{k(m)},\bbeta_{(m)}}
\left\{\sum_{i=1}^{n}\rho_{\tau_{k}}\left(y_{i}-b_{k(m)}-\boldsymbol{x}_{i(m)}^{T}\bbeta_{(m)}\right)\right\},
\ee
and the $\tau_{k}$ QR estimator of $u_{i(m)}+b_{k(m)}$ is
$\tilde{b}_{k(m)}+\boldsymbol{x}_{i(m)}^{T}\hat{\bbeta}_{k(m)}$.
The model averaging estimator for $u_{i}+b_{k}$ is
\be\label{eq:JMA}
\sum_{m=1}^{M}w_{m,k}\left(\tilde{b}_{k(m)}+\boldsymbol{x}_{i(m)}^{T}\hat{\bbeta}_{k(m)}\right)
\ee

For $\tau_{k}$, the leave-one-out cross-validation criterion function is
$$CV_{n,k}(\boldsymbol{w}_{k})=\frac{1}{n}\sum_{i=1}^{n}\rho_{\tau_{k}}\left\{y_{i}-\sum_{m=1}^{M}w_{m,k}\left(\tilde{b}_{k(m)}+\boldsymbol{x}_{i(m)}^{T}\hat{\bbeta}_{k(m)}\right)\right\}.$$
$\hat{\boldsymbol{w}}_{k}=(\hat{w}_{1,k},\ldots,\hat{w}_{M,k})^{T}$ is obtained by choosing $\boldsymbol{w}_{k}\in\mathcal{W}$ such that
\begin{eqnarray}
\hat{\boldsymbol{w}}_{k}=\operatorname*{\argmin\limits}_{\boldsymbol{w}_{k}\in\mathcal{W}}CV_{n,k}(\boldsymbol{w}_{k}).
\end{eqnarray}
Substituting $\hat{\boldsymbol{w}}_{k}$ for $\boldsymbol{w}$ in (\ref{eq:JMA}) results in the following $\tau_{k}$ JMA estimator of $u_{i}+b_{k}$:
$$
\sum_{m=1}^{M}\hat{w}_{m,k}\left(\tilde{b}_{k(m)}+\boldsymbol{x}_{i(m)}^{T}\hat{\bbeta}_{k(m)}\right).
$$

We also consider the estimator based on the model which has the minimum
$$CV_{c,m}=\frac{1}{nK}\sum_{i=1}^{nK}\left\{\sum_{k=1}^{K}\rho_{\tau_{k}}\left(y_{i}-\hat{b}_{ki(m)}-\sum_{m=1}^{M}w_{m}\boldsymbol{x}_{i(m)}^{T}\hat{\bbeta}_{i(m)}\right)\right\}.
$$
We label it CV$_{c}$.

Comparisons are also drawn with the model selection and averaging estimators based on the Akaike information criterion (AIC) and the Schwarz information criterion (SIC). From \cite{Liu1994Quantile} and \cite{Zou2008Regularized}, they are defined as
 $$\mathrm{AIC}_{c,m}=2nK\log\left\{\frac{1}{nK}\sum_{i=1}^{n}\sum_{k=1}^{K}\rho_{\tau_{k}}\left(y_{i}-\hat{b}_{ki(m)}-\boldsymbol{x}_{i(m)}^{T}\hat{\bbeta}_{i(m)}\right)\right\}
+2(\tilde{k}_{m}+K),\ and$$
$$\mathrm{SIC}_{c,m}=2nK\log\left\{\frac{1}{nK}\sum_{i=1}^{n}\sum_{k=1}^{K}\rho_{\tau_{k}}\left(y_{i}-\hat{b}_{ki(m)}-\boldsymbol{x}_{i(m)}^{T}\hat{\bbeta}_{i(m)}\right)\right\}
+(\tilde{k}_{m}+K)\log(nK).$$
These criteria are same as the conventional criteria except the error variance is estimated based on the check loss function.
Based on these criteria, we define the following composite smoothed AIC and the composite smoothed SIC weights for model $m$:
$$\hat{w}_{m}^{\mathrm{AIC}_{c}}=\frac{\exp\left(-\frac{1}{2}\mathrm{AIC}_{c,m}\right)}{\sum_{j=1}^{M}\exp\left(-\frac{1}{2}\mathrm{AIC}_{c,j}\right)}
\ \text{and}\ \hat{w}_{m}^{\mathrm{SIC}_{c}}=\frac{\exp\left(-\frac{1}{2}\mathrm{SIC}_{c,m}\right)}{\sum_{j=1}^{M}\exp\left(-\frac{1}{2}\mathrm{SIC}_{c,j}\right)} ,\ \text{respectively}.$$
We label the estimators that result from the model selection and averaging estimators based on the AIC and the SIC the AIC$_{c}$, SIC$_{c}$, SAIC$_{c}$ and SSIC$_{c}$ estimators respectively.

We evaluate the performance of the above methods except MCV$_{0}$ with respect to the following final prediction error measure across $R$ replications:
 $$\mathrm{CPE}=\frac{1}{R}\sum_{r=1}^{R}\mathrm{PE}(r),$$
where
$$\mathrm{PE}(r)=\frac{1}{n_sK}\sum_{s=1}^{n_s}\sum_{k=1}^{K}\rho_{\tau_{k}}\left\{y_{s}-\sum_{m=1}^{M}\hat{w}_{m,k}\left(\tilde{b}_{k(m)}+\boldsymbol{x}_{s(m)}^{T}\hat{\bbeta}_{k(m)}\right)\right\}$$ for the method MCV$_{0}$ and
$$\mathrm{PE}(r)=\frac{1}{n_sK}\sum_{s=1}^{n_s}\sum_{k=1}^{K}\rho_{\tau_{k}}\left\{y_{s}-\sum_{m=1}^{M}\hat{w}_{m}\left(\hat{b}_{k,m}+\boldsymbol{x}_{s(m)}^{T}\hat{\bbeta}_{(m)}\right)\right\}$$
for others
is the prediction error from the $r^{th}$ replication based on the out-of-sample observations $\{\boldsymbol{x}_{s},y_{s}\}_{s=1}^{n_s}$ that vary across the replications and a given averaging/selection method that uses $\hat{w}_{m}$ as the weight for the $m^{th}$ model.  We set $n_s=n$.
\subsection{Simulation Settings}
For comparison, we consider the three similar simulation setups as that in \cite{lu2015jackknife} firstly.

\textbf{Setting 1}
$$y_{i}=\nu\sum_{j=1}^{1000}j^{-1}x_{ij}+\varepsilon_{i},~~~~~~~~~~~~i=1,2,\cdots,n,$$
where $x_{ij}$ ($j=1,2,\cdots$) follow $\mathrm{IID}$ $ N(0,1)$ distribution, $\nu$ is varied so that $R^{2}=\left\{Var(y_{i})-Var(\varepsilon_{i})\right\}/Var(y_{i})= 0.1,0.3,\cdots,0.9$, and $\varepsilon_{i}$ is an error term.  We consider the following situations: (i) $\varepsilon_{i}\sim N(0,1)$ and  (ii) $\varepsilon_{i}=\sum_{j=1}^{6}x_{ij}^{2}\epsilon_{i}$, representing homoscedastic and heteroscedastic errors, respectively, where $\epsilon_{i}\sim N(0,1)$ and is independent of $x_{ij}$. Except the quantiles are taken as $\tau_{1}=0.05$,
$\tau_{2}=0.5$ and $\tau_{3}=0.95$ with $K=3$,
the other values are the same as that in \cite{lu2015jackknife}. That is
the sample size $n=50$, $n=100$, $n=150$, the number of candidate model $M=\lfloor3n^{1/3}\rfloor$ where $\lfloor\cdot\rfloor$ represents the integer part of $\cdot$ and the number of replications $R=200$.
For ease of implementation, we consider nested models and each model contains an intercept term. Furthermore,
the first model contains the first variable $\{x_{i1}\}$, the second model contains the first and second variables $\{x_{i1},x_{i2}\}$, and so on.


\textbf{Setting 2}
$$y_{i}=\nu\sum_{j=1}^{30}j^{-1}x_{ij}+\varepsilon_{i},~~~~~~~~~~~~i=1,2,\cdots,n,$$
where $x_{ij}$ ($j=1,2,\cdots$) are $\mathrm{IID}$ $\chi^{2}(1)$ and the error term $\varepsilon_{i}$ be considered the following situations: (i) $\varepsilon_{i}=\epsilon_{i}$ and  (ii) $\varepsilon_{i}=\sum_{j=1}^{30}j^{-1}x_{ij}\epsilon_{i}$, representing homoscedastic and heteroscedastic errors, respectively, where $\epsilon_{i}$ is normalized $\chi^{2}(3)$ with mean $0$ and variance $1$ and independent of $x_{ij}$. Except we fix $M=20$ for each case,
the other values are same as that in Setting 1.

\textbf{Setting 3}
$$y_{i}=\nu\sum_{j=1}^{25}j^{-1}\Phi(x_{ij})+\varepsilon_{i},~~~~~~~~~~~~i=1,2,\cdots,n,$$
where $x_{ij}$ ($j=1,2,\cdots$) follow $\mathrm{IID}$ $ N(0,1)$ distribution and $\Phi(\cdot)$ is the standard normal distribution function. We implement the following cases for $\varepsilon_{i}$: (i) $\varepsilon_{i}=\epsilon_{i}$ and  (ii) $\varepsilon_{i}=\left(0.01+\sum_{j=1}^{11}x_{ij}^{2}\right)\epsilon_{i}$, representing homoscedastic and heteroscedastic errors, respectively, where $\epsilon_{i}\sim N(0,1)$ and is independent of $x_{ij}$.
The other values are the same as that in Setting 2.

Next, we adopt the data generating processes similar as that in \cite{Zhao2015Empirical}.

\textbf{Setting 4}
$$y_{i}=\nu\cdot\boldsymbol{x}_{i}^{T}\bbeta+\varepsilon_{i},~~~~~~~~~~~~i=1,2,\cdots,n,$$
where $\bbeta=(1,0,0.5,-0.1)^{T}$, each of $x_{ij}$, $j=1,\cdots,4,$ is $\mathrm{IID}$ $N(0,1)$ and mutually independent of each other, $\nu$ is varied so that $R^{2}=\left\{\mathrm{Var}(y_{i})-\mathrm{Var}(\varepsilon_{i})\right\}/\mathrm{Var}(y_{i})= 0.1,0.3,\cdots,0.9$. The response $y_{i}$ is generated according to the model. Here we considered three different error distributions.\\
\textbf{Case} $\boldsymbol{1}$ The error distribution follows the normal distribution
$N(0,0.5)$.\\
\textbf{Case} $\boldsymbol{2}$ The error distribution follows the chi-square distribution with one degree of freedom $\chi_{1}^{2}$.\\
\textbf{Case} $\boldsymbol{3}$ The error distribution follows the mixture of normal distribution $0.5N(0,1)+0.5N(0, 0.5)$.

Furthermore, $8$ candidate models are considered with each containing the first regressor. The quantiles are taken as $\tau_{k}=k/6$ with $K=5$, the sample size is taken as $n=50$, $n=100$ and $n=200$ for each case, we take 200 simulation runs.
\subsection{Simulation results}
Figure \ref{Fig1}-\ref{Fig6} report the $\mathrm{CPEs}$ of the various estimators under \textbf{Setting} $\boldsymbol{1}$-\textbf{Setting} $\boldsymbol{3}$, respectively. The results are expressed in terms of $R^{2}$. Of all homoscedastic settings considered,
MCV$_{c}$ estimator clearly denominates the other six estimators, including commonly used model averaging and model selection estimators. Under the heteroscedastic setting, the situation is somewhat different. Next, we give a detailed description.

Figure \ref{Fig1} reports the CPEs of the various estimators under the homoscedastic error of Setting 1. The results show that MCV$_{c}$ estimators are superior to other estimators in terms of minimizing CPEs. When $R^{2}$ is small to moderate, the CPEs of MCV$_{0}$ estimators less than that of SAIC$_{c}$, AIC$_{c}$ and CV$_{c}$ estimators, however these CPEs are greater than that of SBIC$_{c}$ and BIC$_{c}$ estimators. Differently, when $R^{2}=0.9$, We can't discern the relative merits of SAIC$_{c}$, SBIC$_{c}$, AIC$_{c}$, BIC$_{c}$ and CV$_{c}$ estimators but they all are better than MCV$_{0}$ estimators in terms of minimizing CPEs. In the case of heteroscedasticity presented in Figure \ref{Fig2}, the performance of these estimates is relatively clear. The MCV$_{c}$ estimators are usually the most popular estimator, and the SBIC$_{c}$ and BIC$_{c}$ estimators typically have a smaller CPEs than the other remaining estimators. The final conclusion is that MCV$_{c}$ estimator is the most favored one in Setting 1.

For the results of Setting 2 reported in Figure \ref{Fig3}-\ref{Fig4}, under the homoscedastic error, the performance of each estimators is very similar to that of Setting 1, but that is very different in the case of heteroscedasticity. From Figure \ref{Fig4}, of all cases considered,
the MCV$_{0}$ estimator produces the best estimates.
When the sample size and $R^{2}$ are not large, the performances of MCV$_{c}$ are slightly better than the other remaining estimators, containing SAIC$_{c}$, SBIC$_{c}$, AIC$_{c}$, BIC$_{c}$ and CV$_{c}$ estimators. But when $R^2$ is moderate to large, all estimators except MCV$_{0}$ estimators perform similarly.

For the results of Setting 3 shown in Figure \ref{Fig5}-\ref{Fig6}, under the homoscedastic error, SAIC$_{c}$, AIC$_{c}$ and CV$_{c}$ estimators are the least preferred estimators which enjoy the largest CPEs, while MCV$_{c}$ estimators have the smallest CPEs. When $n=50$,
the performance of MCV$_{0}$ estimators are superior to any of the three model selection composite quantile estimators and the model averaging composite quantile estimators apart from MCV$_{c}$ estimators. However, when $n=100$ and $n=150$, the
CPEs of SBIC$_{c}$ and BIC$_{c}$ estimators are very close and smaller than that of
MCV$_{0}$ estimators when $R^{2}$ is small. When $R^{2}$ is small to moderate, the
CPEs of these three estimators are very close. Under the heteroscedastic error,
the estimators produced by MCV$_{0}$ are worse than the estimators produced by MCV$_{c}$, implying that the composite quantile estimator typically having an advantage over the estimator using one quantile.
For the comparison of MCV$_{c}$ with SAIC$_{c}$, SBIC$_{c}$, AIC$_{c}$, BIC$_{c}$ and CV$_{c}$ estimators, the description of the performances of them are exactly the same as that in setting 2 under the heteroscedastic error.

Table \ref{Fig7}-\ref{Fig9} report the $\mathrm{CPEs}$ of the various estimators under \textbf{Case} $\boldsymbol{1}$-\textbf{Case} $\boldsymbol{3}$ in \textbf{Setting} $\boldsymbol{4}$, respectively. The results are expressed in terms of $R^{2}$.
The results show that the MCV$_{c}$ estimator is most frequently the estimator that enjoys the smallest $\mathrm{CPEs}$.
The MCV$_{0}$ estimator produces the worst estimates, with
the composite quantile estimator typically having an advantage over the estimator using one quantile. Of all cases considered, the model averaging composite quantile estimators are more or less superior to any of the three model selection composite quantile estimators.

\begin{figure}[b]
\centering
\subfigure{
\includegraphics[width=8cm,height=6.8cm]{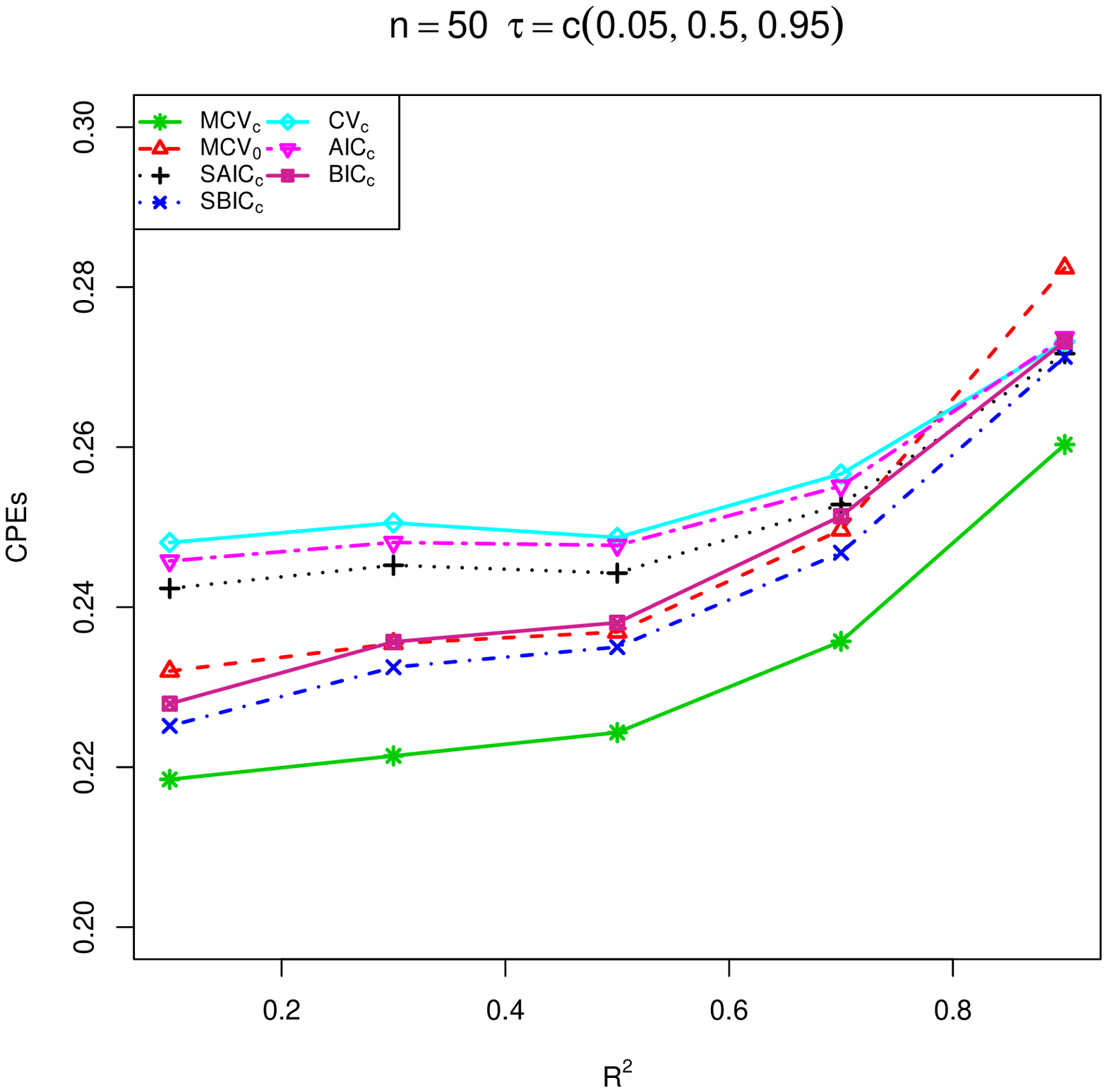}}
\subfigure{
\includegraphics[width=8cm,height=6.8cm]{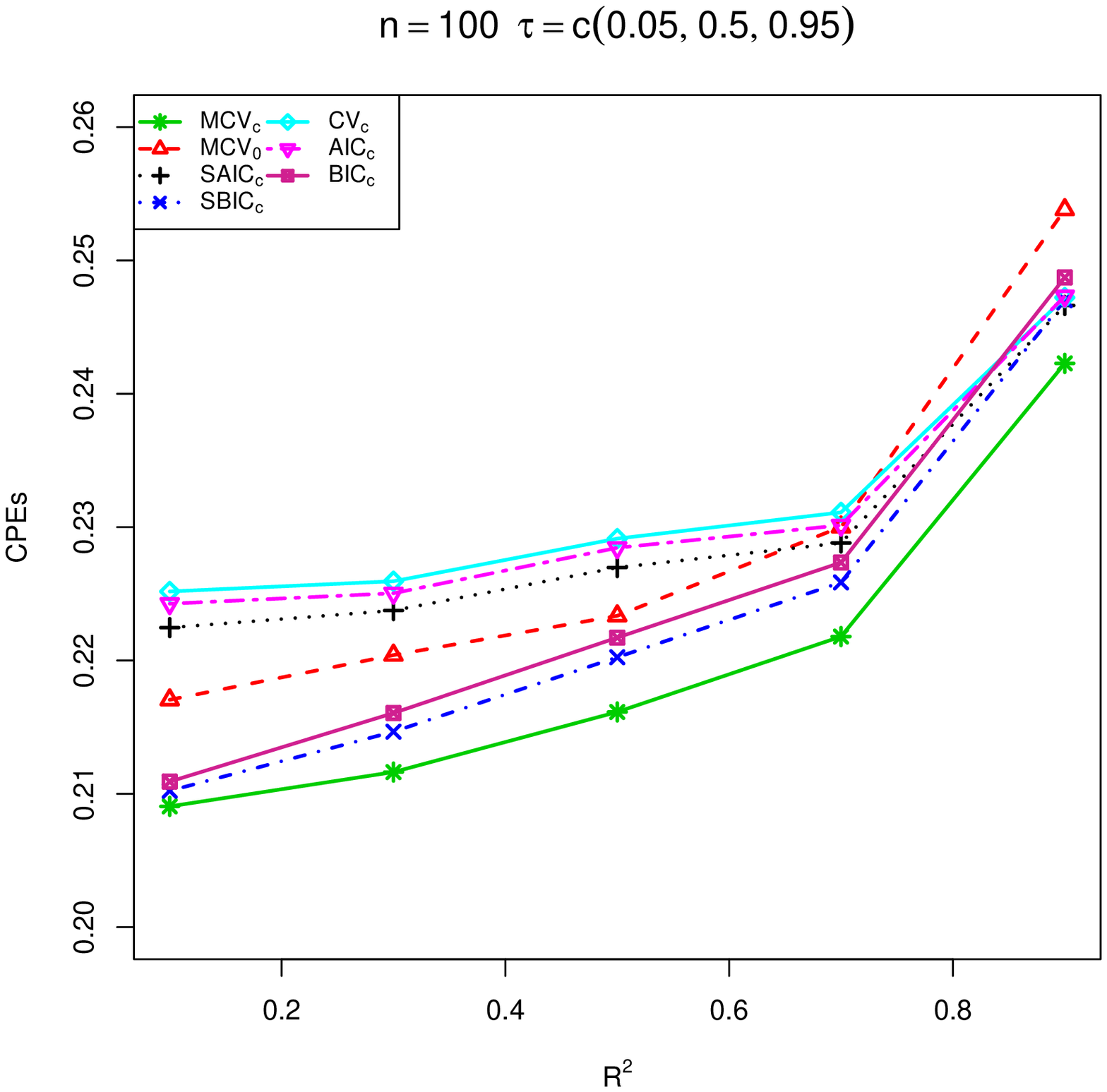}}
\subfigure{
\includegraphics[width=8cm,height=6.8cm]{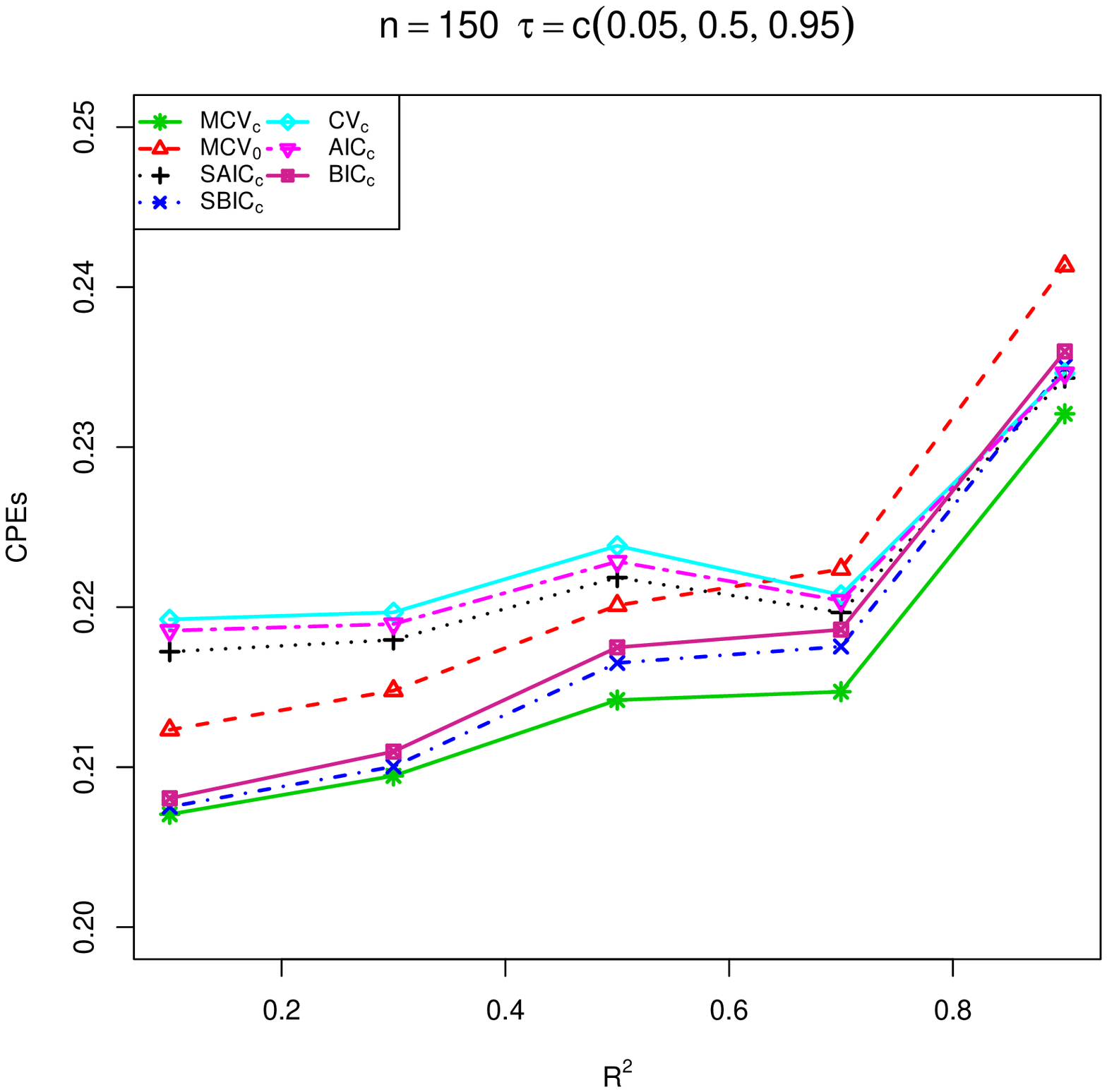}}
\caption{Out-of-sample performance of Setting 1: Homoscedasticity.}
\label{Fig1}
\end{figure}
\begin{figure}[b]
\centering
\subfigure{
\includegraphics[width=8cm,height=6.8cm]{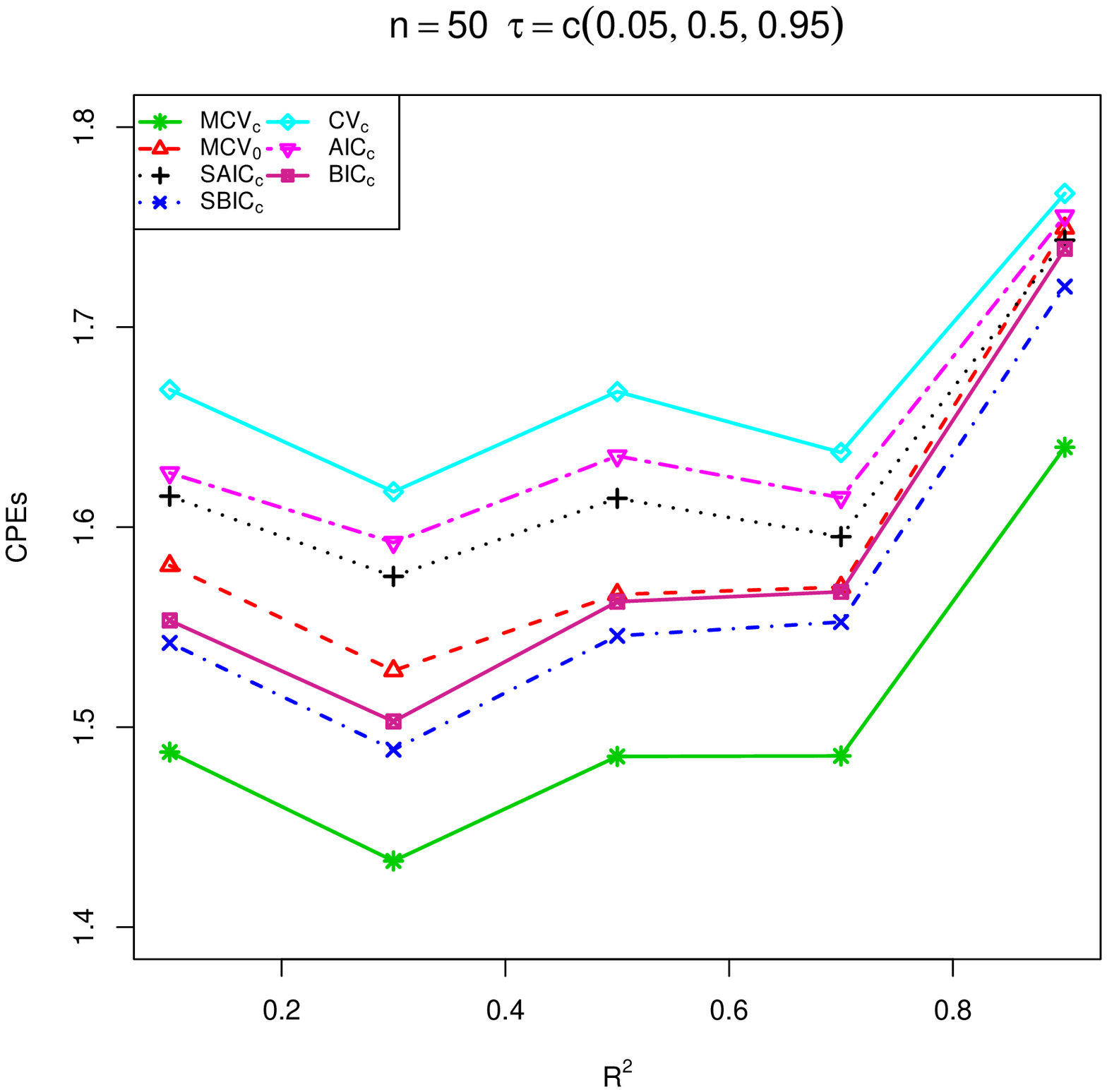}}
\subfigure{
\includegraphics[width=8cm,height=6.8cm]{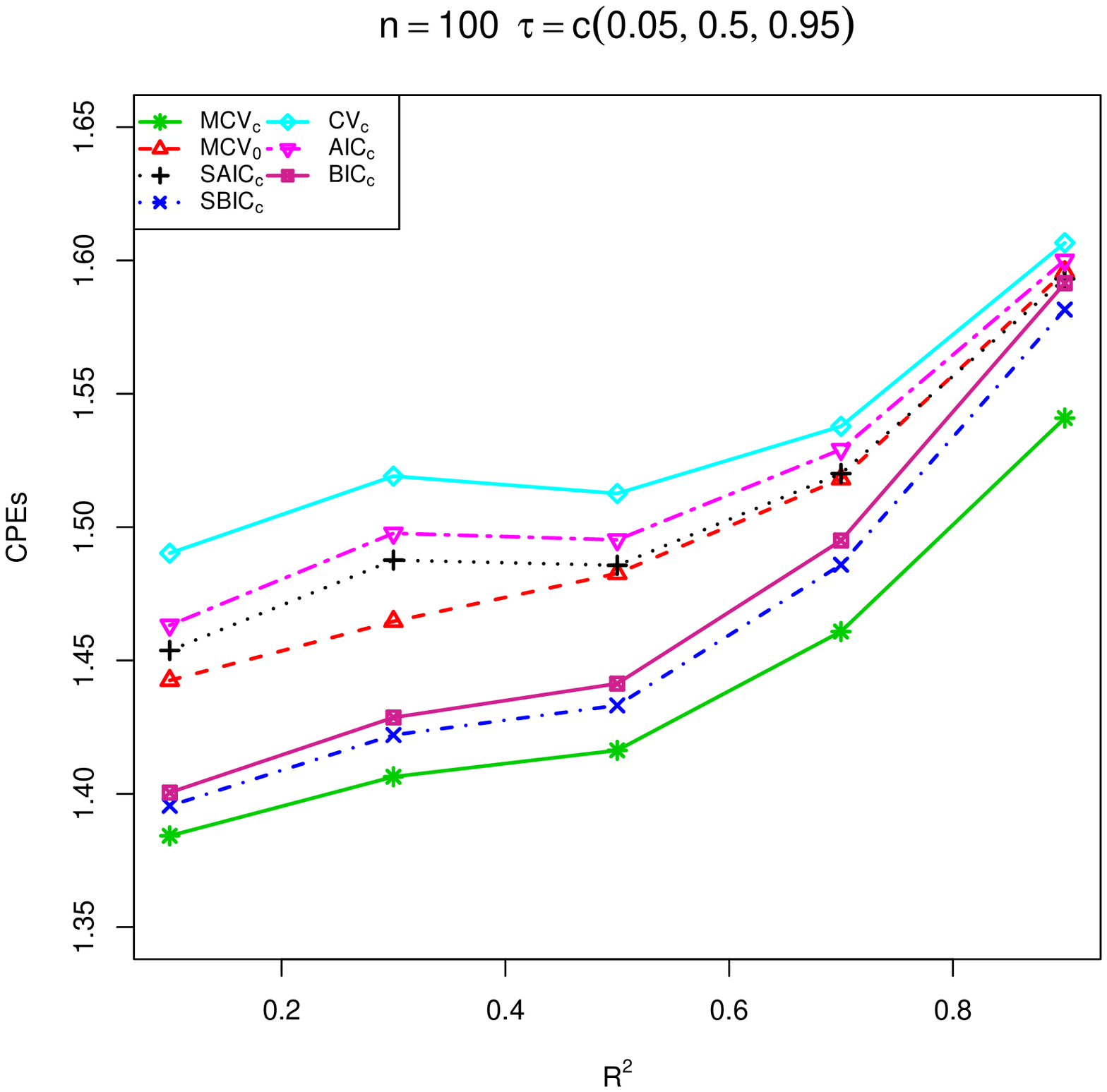}}
\subfigure{
\includegraphics[width=8cm,height=6.8cm]{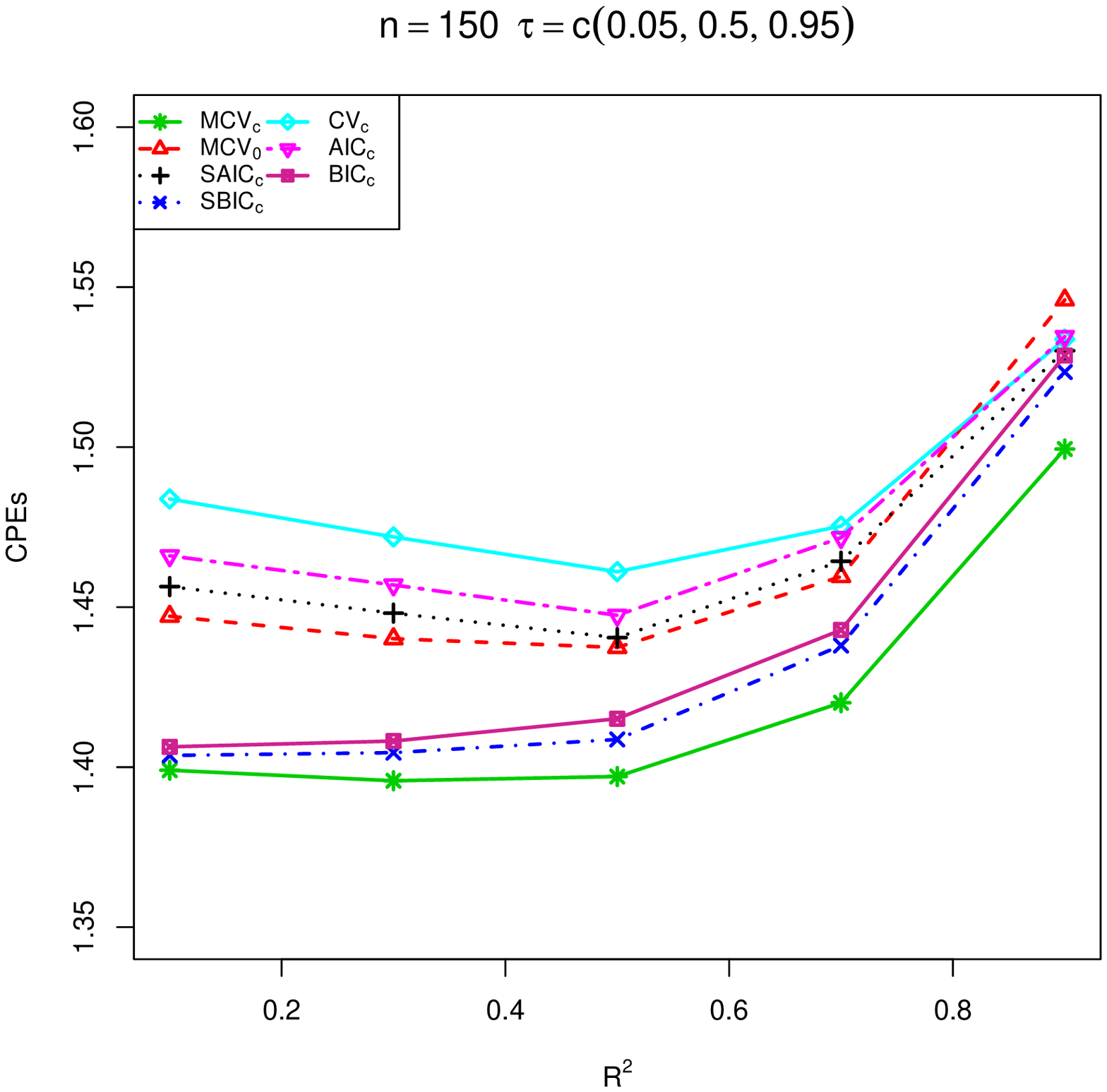}}
\caption{Out-of-sample performance of Setting 1: Heteroscedasticity.}
\label{Fig2}
\end{figure}
\begin{figure}[b]
\centering
\subfigure{
\includegraphics[width=8cm,height=6.8cm]{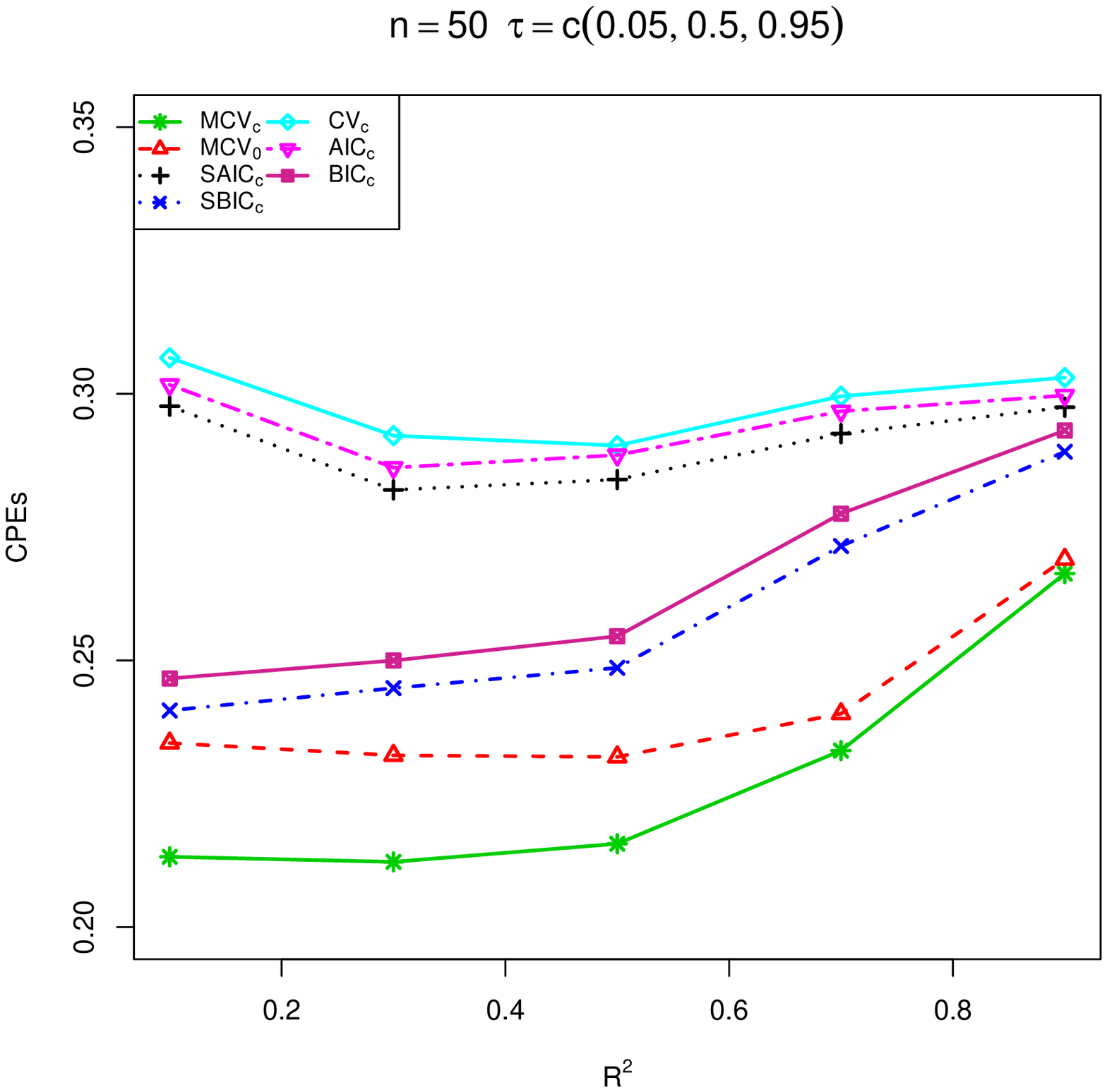}}
\subfigure{
\includegraphics[width=8cm,height=6.8cm]{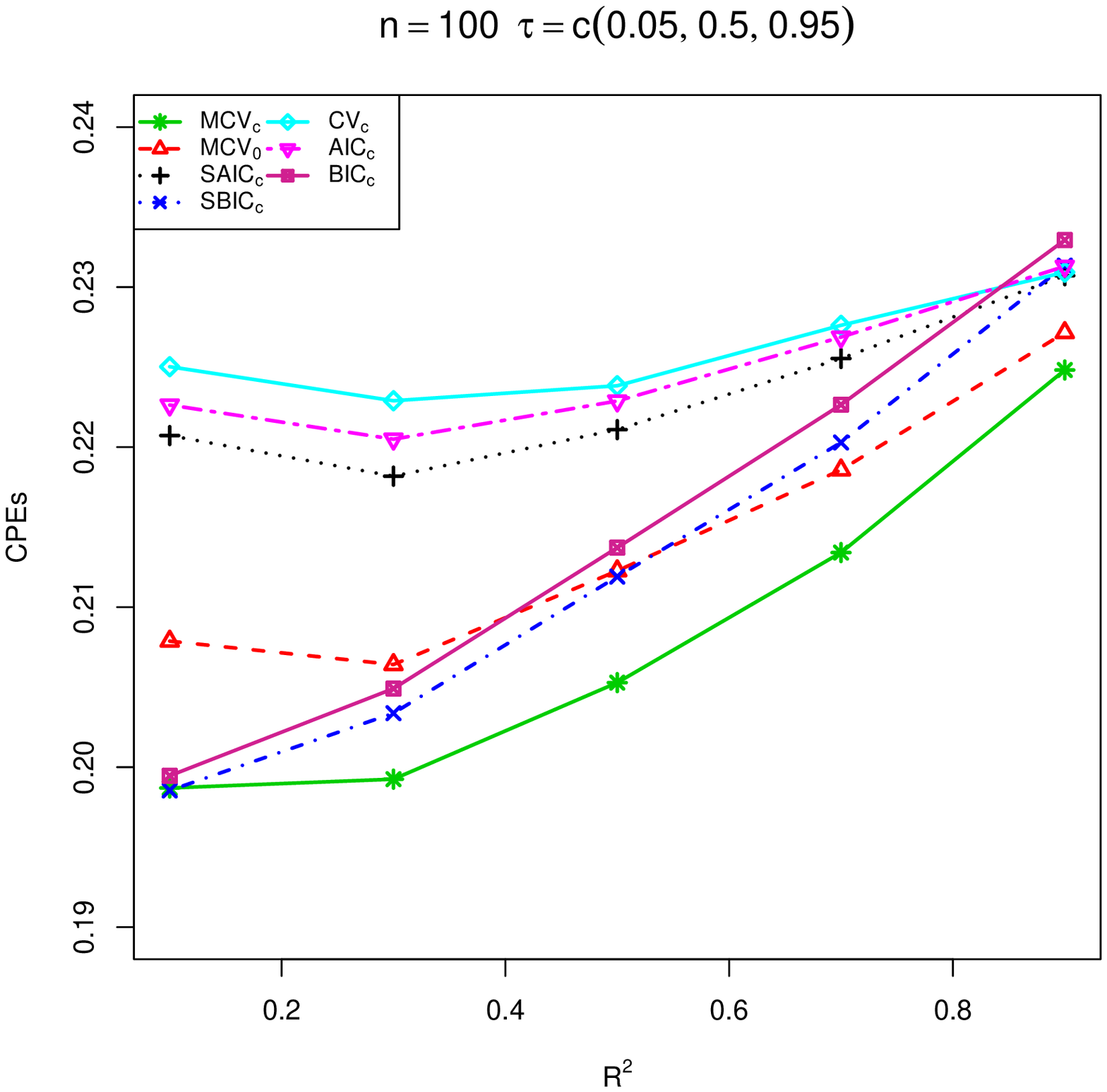}}
\subfigure{
\includegraphics[width=8cm,height=6.8cm]{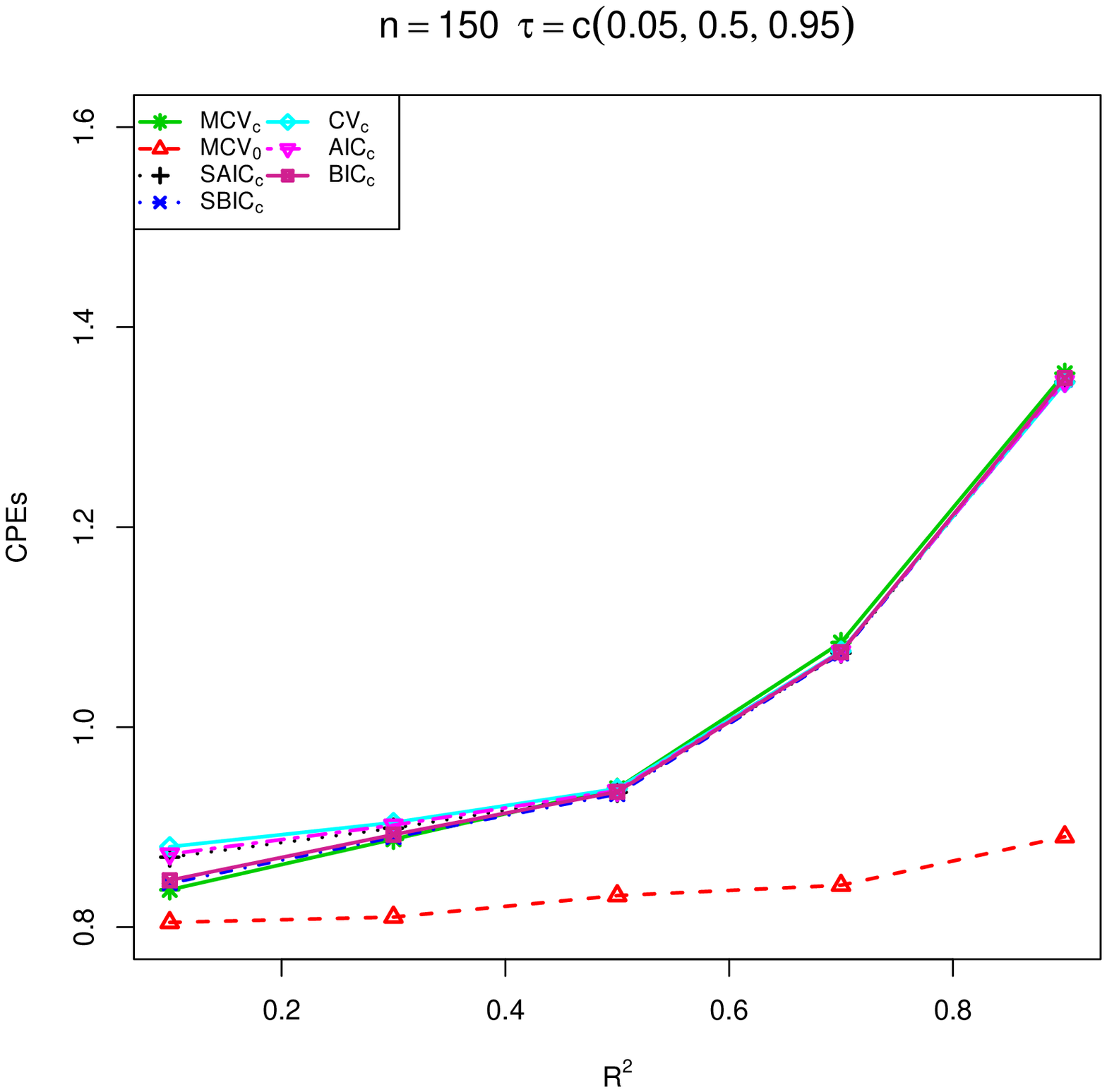}}
\caption{Out-of-sample performance of Setting 2: Homoscedasticity.}
\label{Fig3}
\end{figure}
\begin{figure}[b]
\centering
\subfigure{
\includegraphics[width=8cm,height=6.8cm]{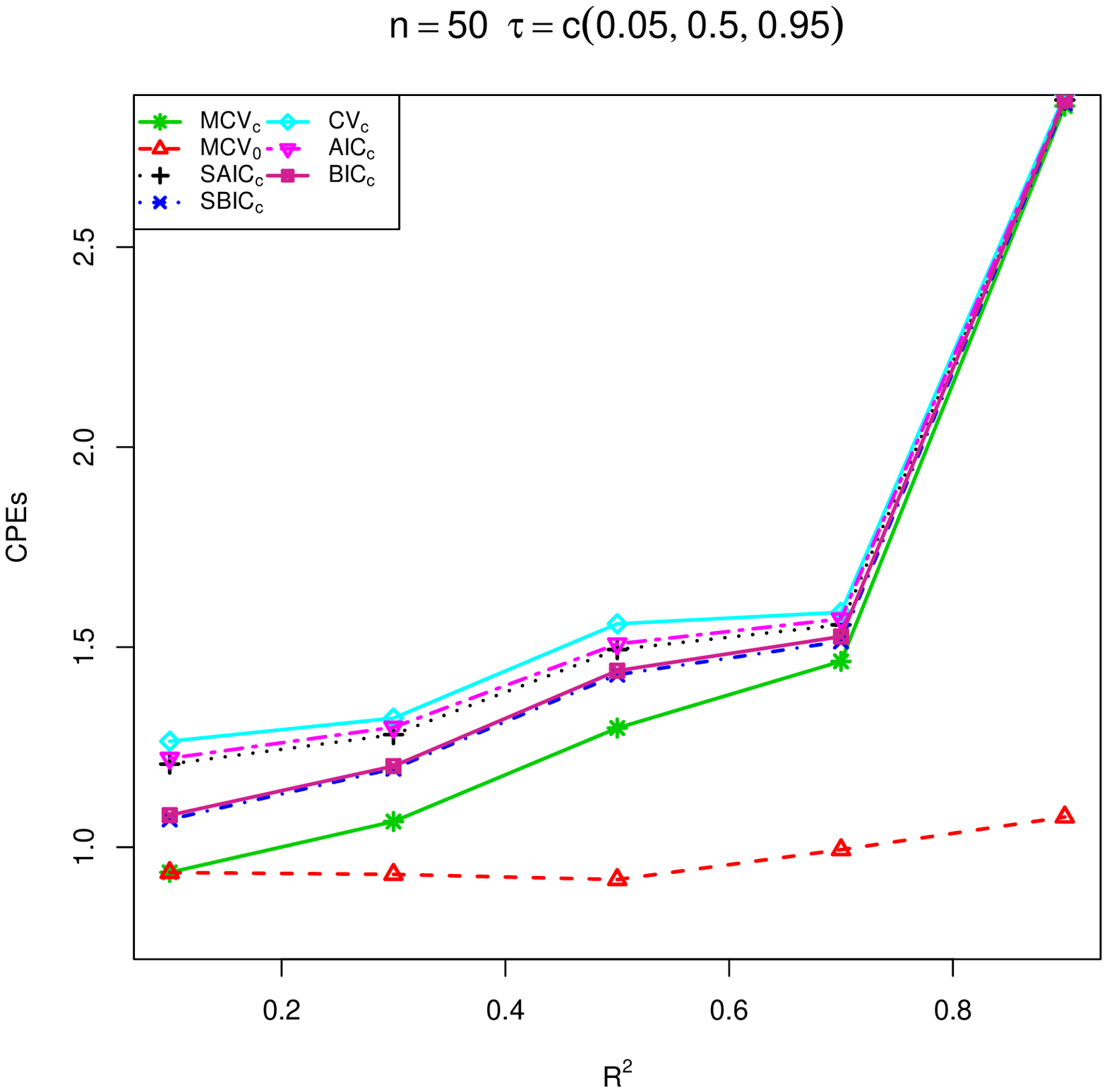}}
\subfigure{
\includegraphics[width=8cm,height=6.8cm]{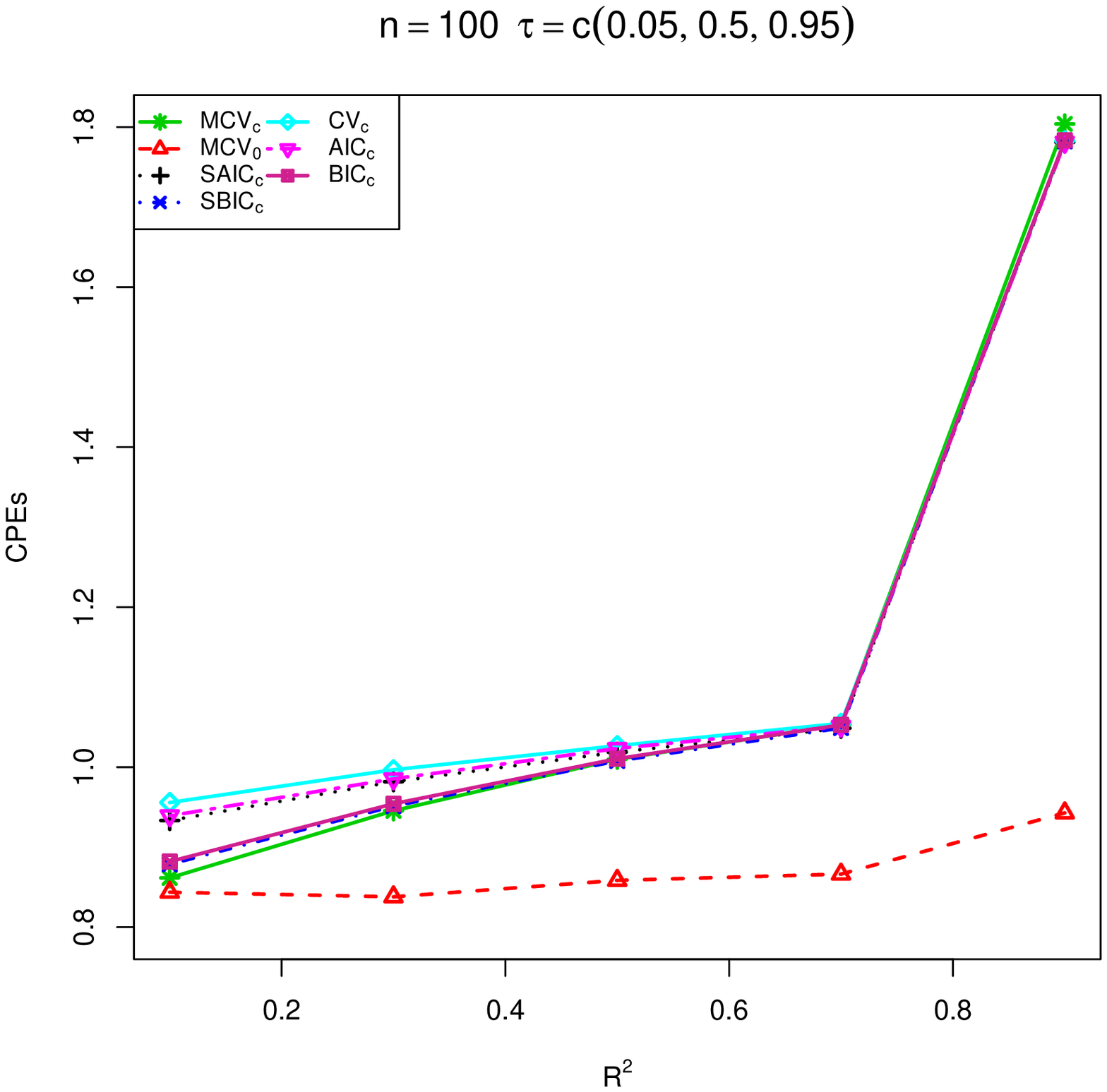}}
\subfigure{
\includegraphics[width=8cm,height=6.8cm]{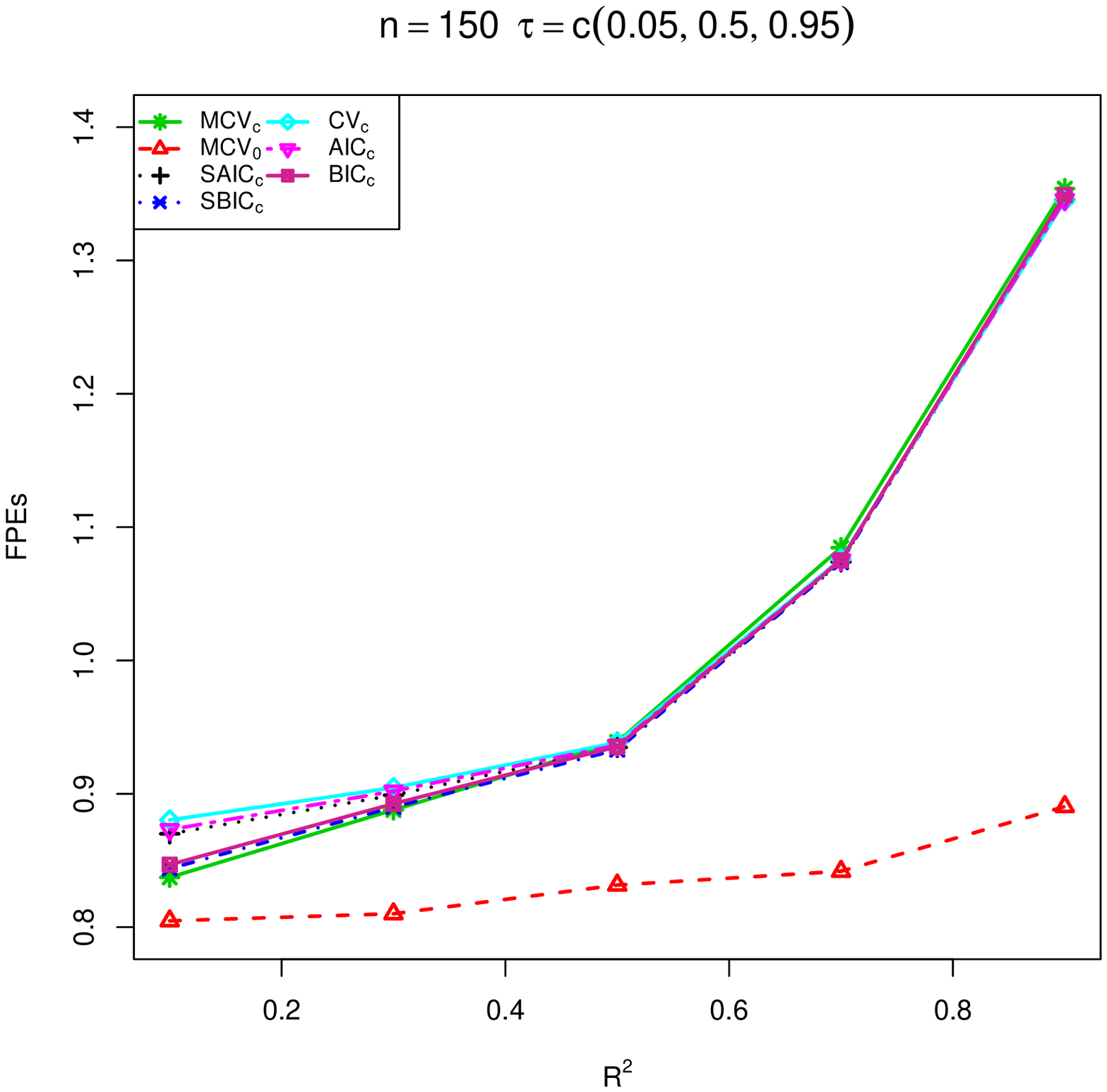}}
\caption{Out-of-sample performance of Setting 2: Heteroscedasticity.}
\label{Fig4}
\end{figure}
\begin{figure}[b]
\centering
\subfigure{
\includegraphics[width=8cm,height=6.8cm]{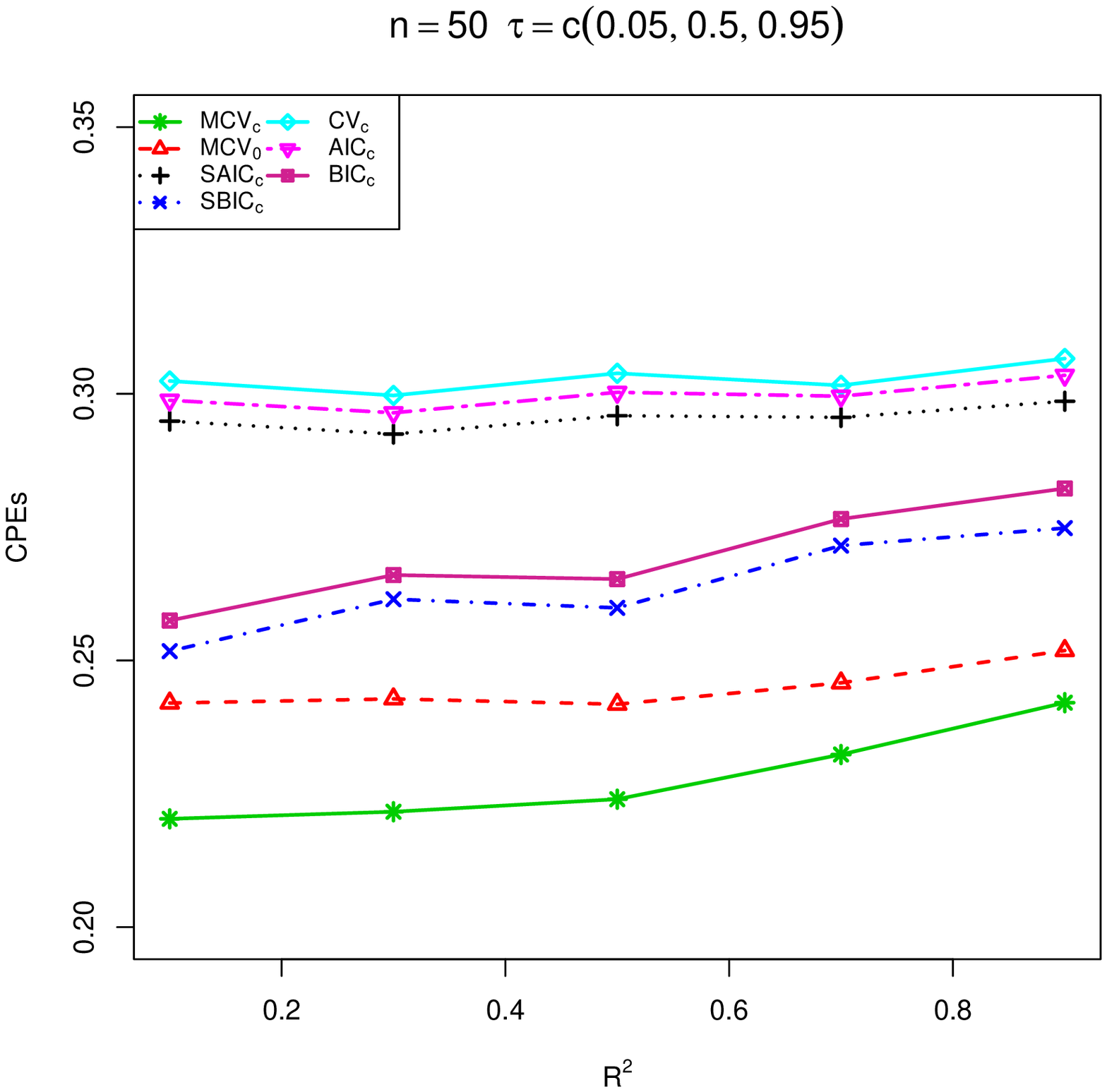}}
\subfigure{
\includegraphics[width=8cm,height=6.8cm]{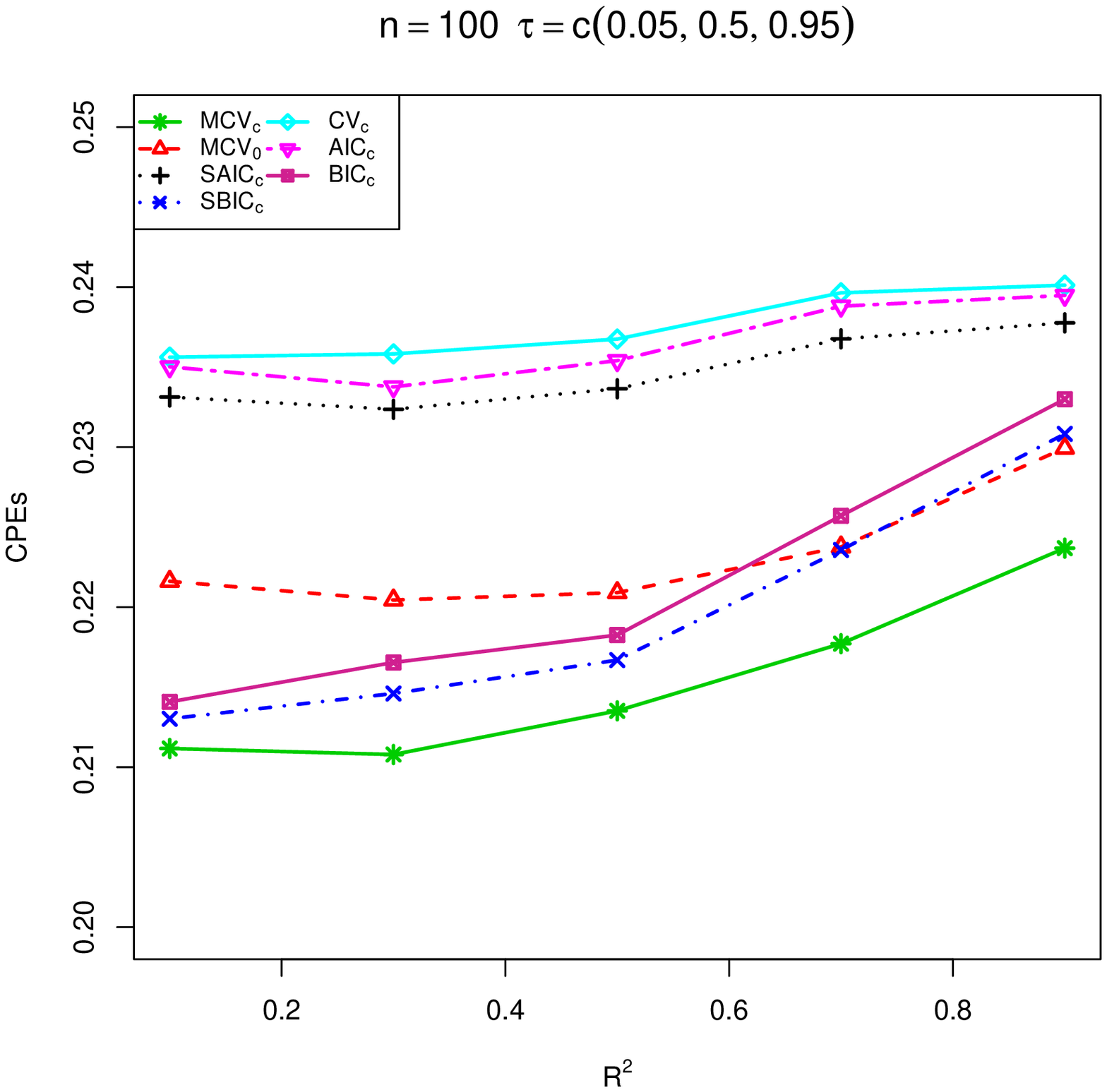}}
\subfigure{
\includegraphics[width=8cm,height=6.8cm]{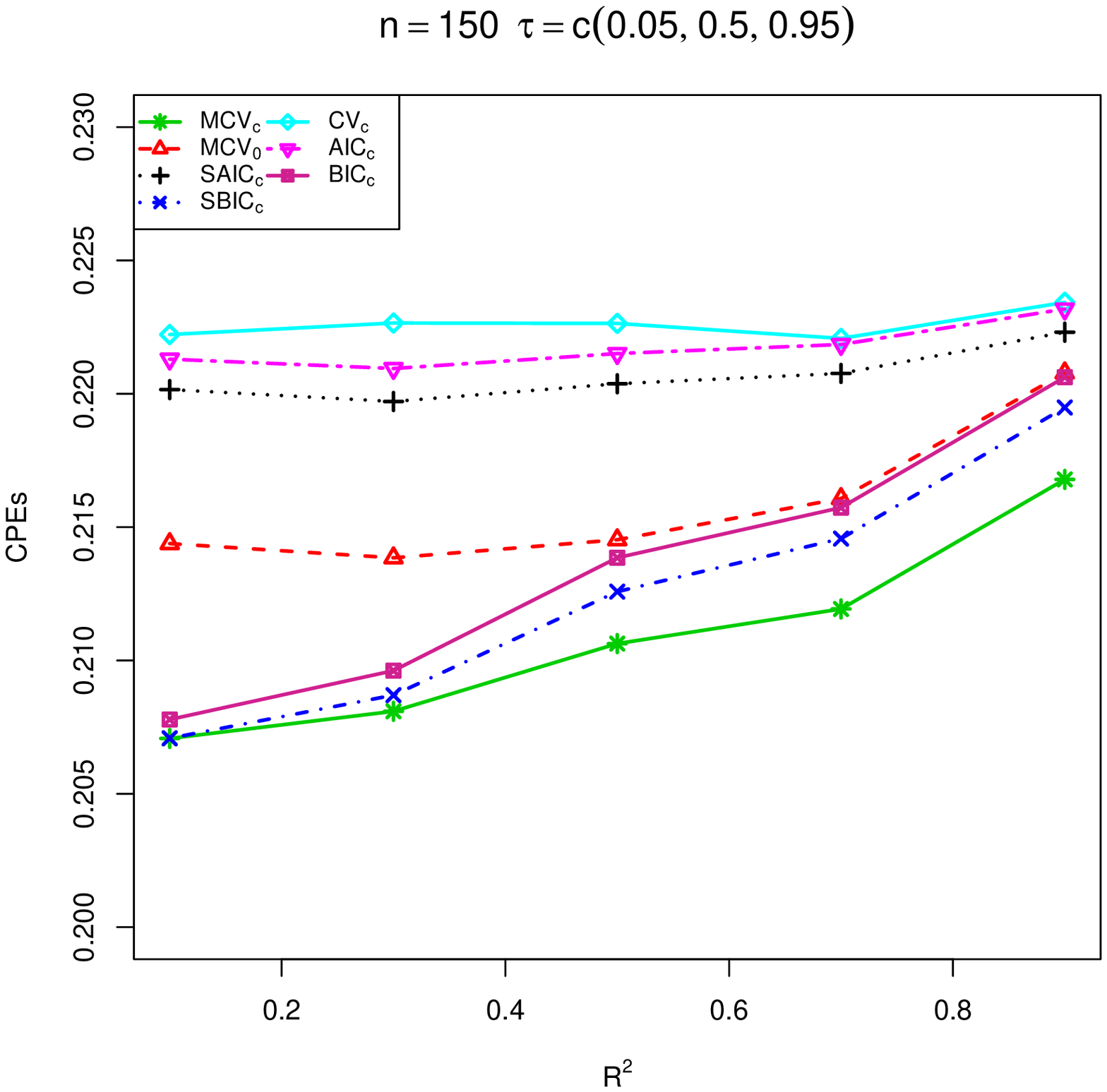}}
\caption{Out-of-sample performance of Setting 3: Homoscedasticity.}
\label{Fig5}
\end{figure}
\begin{figure}[b]
\centering
\subfigure{
\includegraphics[width=8cm,height=6.8cm]{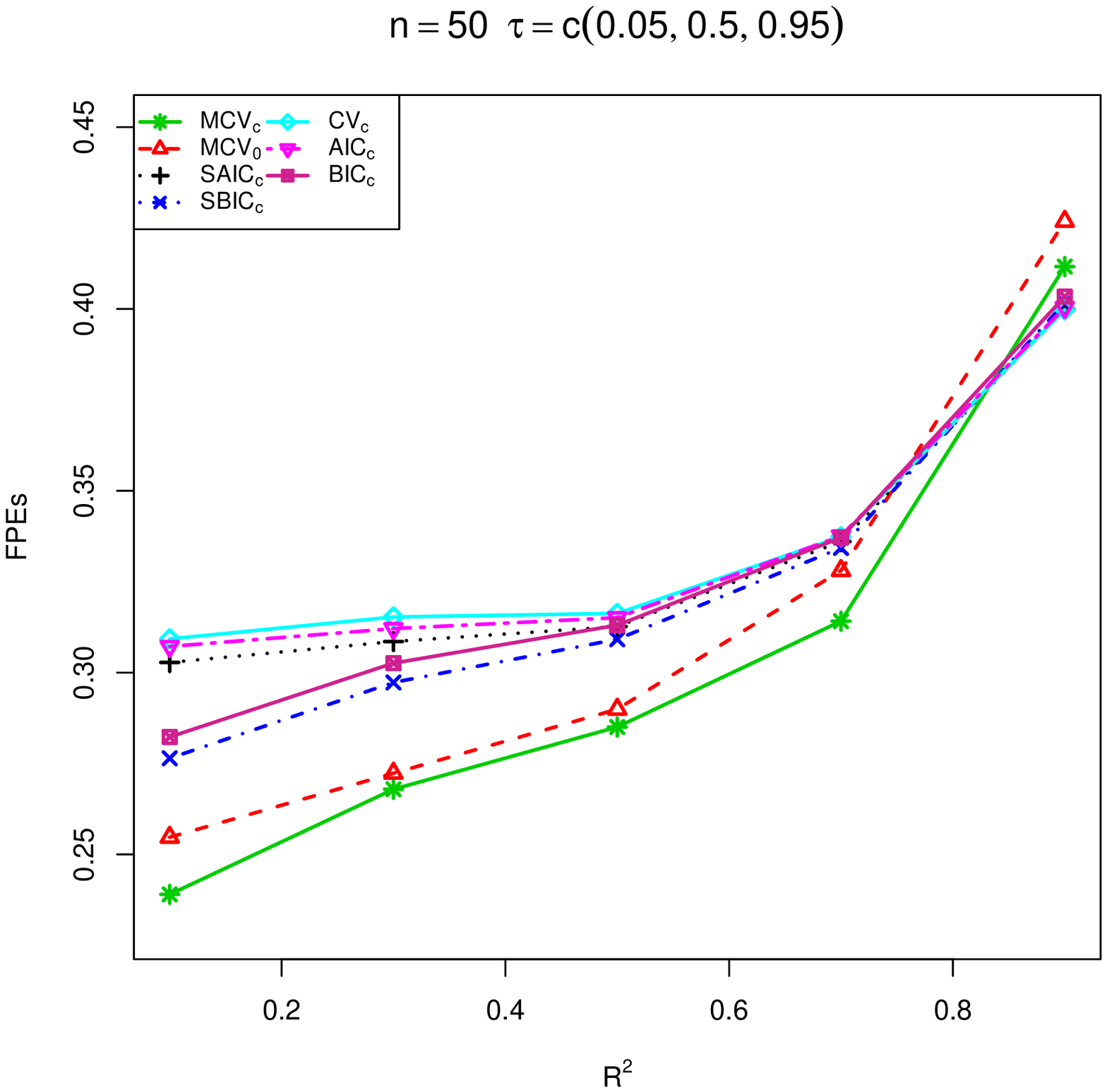}}
\subfigure{
\includegraphics[width=8cm,height=6.8cm]{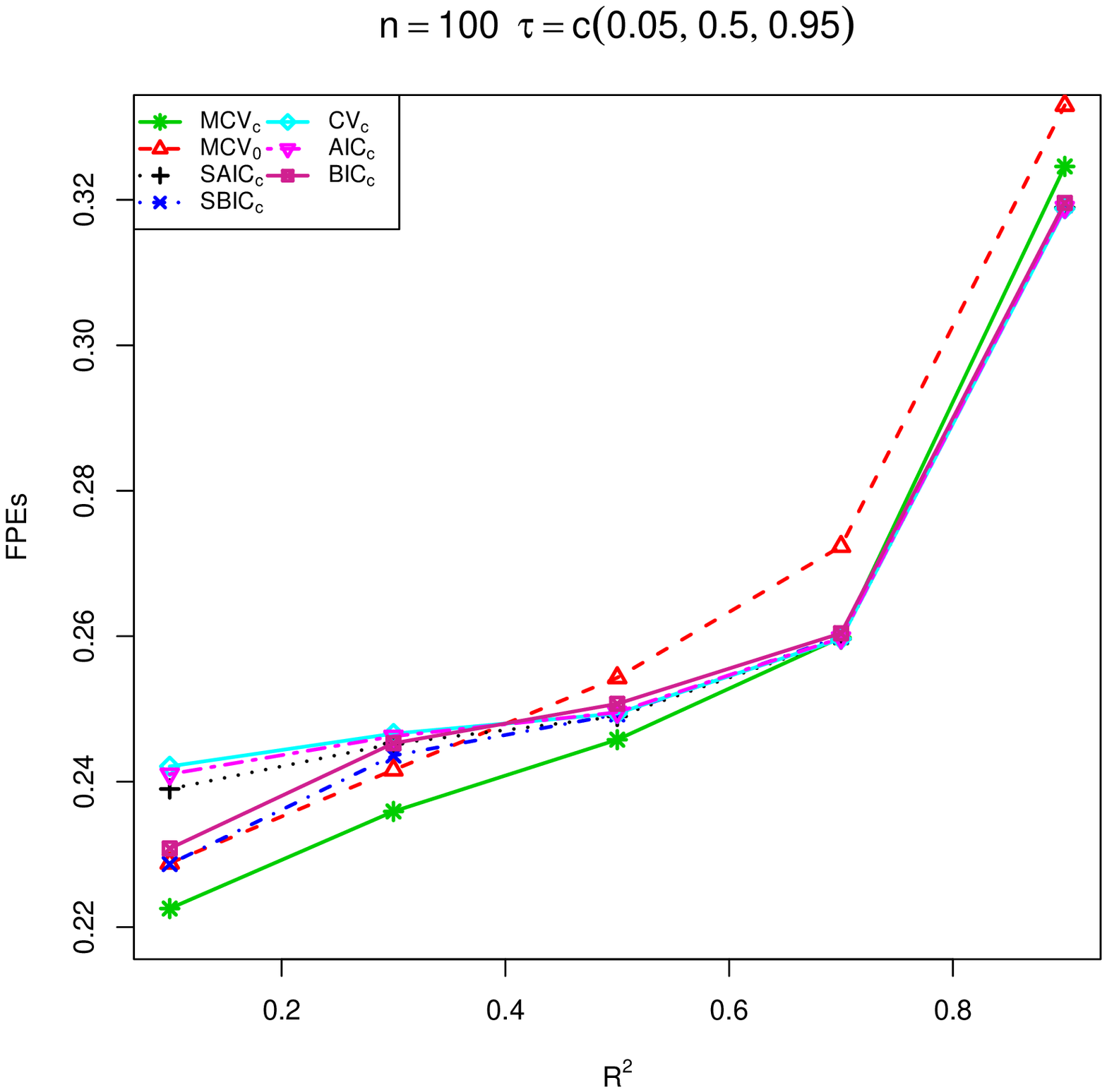}}
\subfigure{
\includegraphics[width=8cm,height=6.8cm]{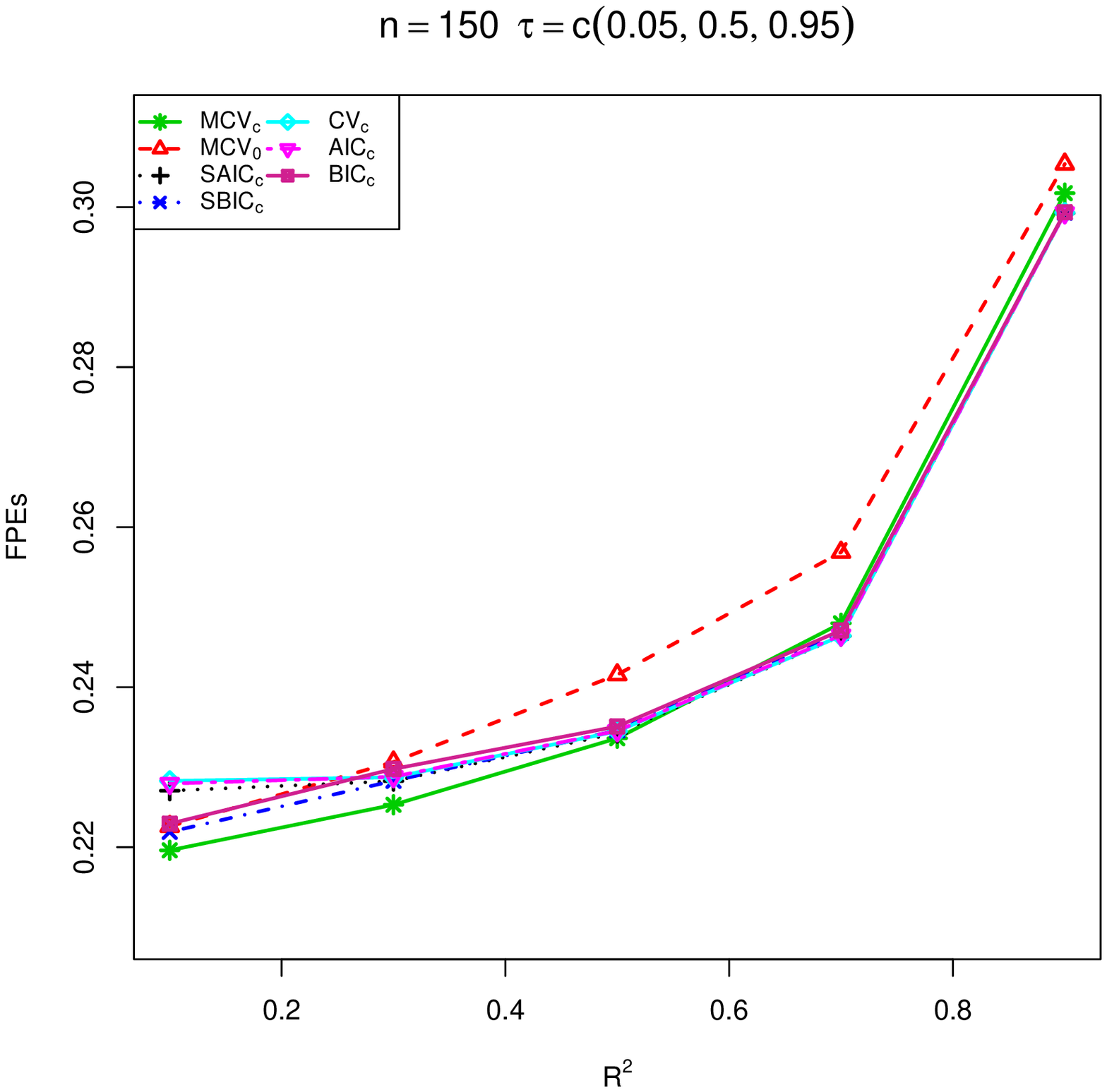}}
\caption{Out-of-sample performance of Setting 3: Heteroscedasticity.}
\label{Fig6}
\end{figure}
%
\begin{figure}[b]
\centering
\subfigure{
\includegraphics[width=8cm,height=6.8cm]{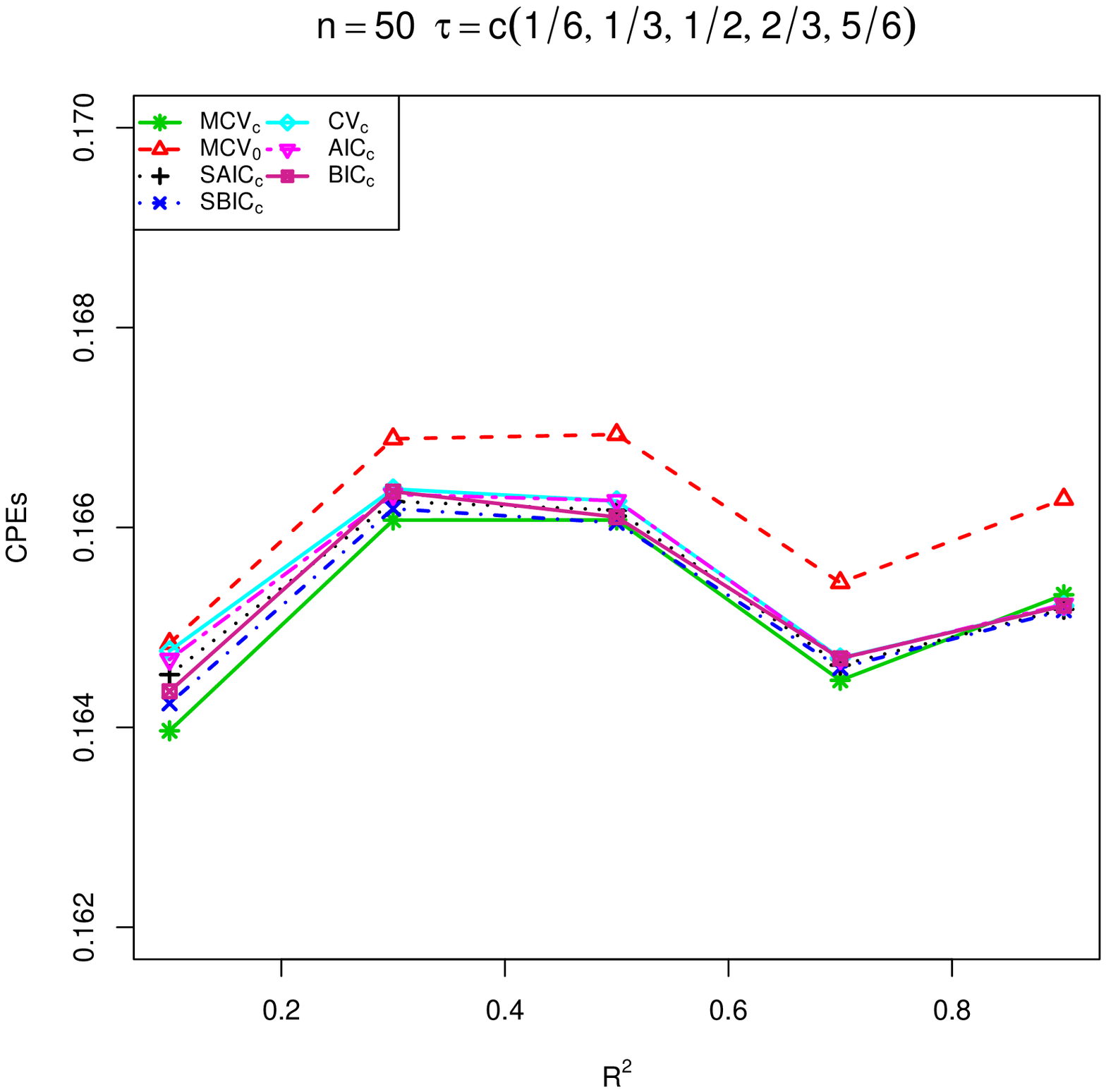}}
\subfigure{
\includegraphics[width=8cm,height=6.8cm]{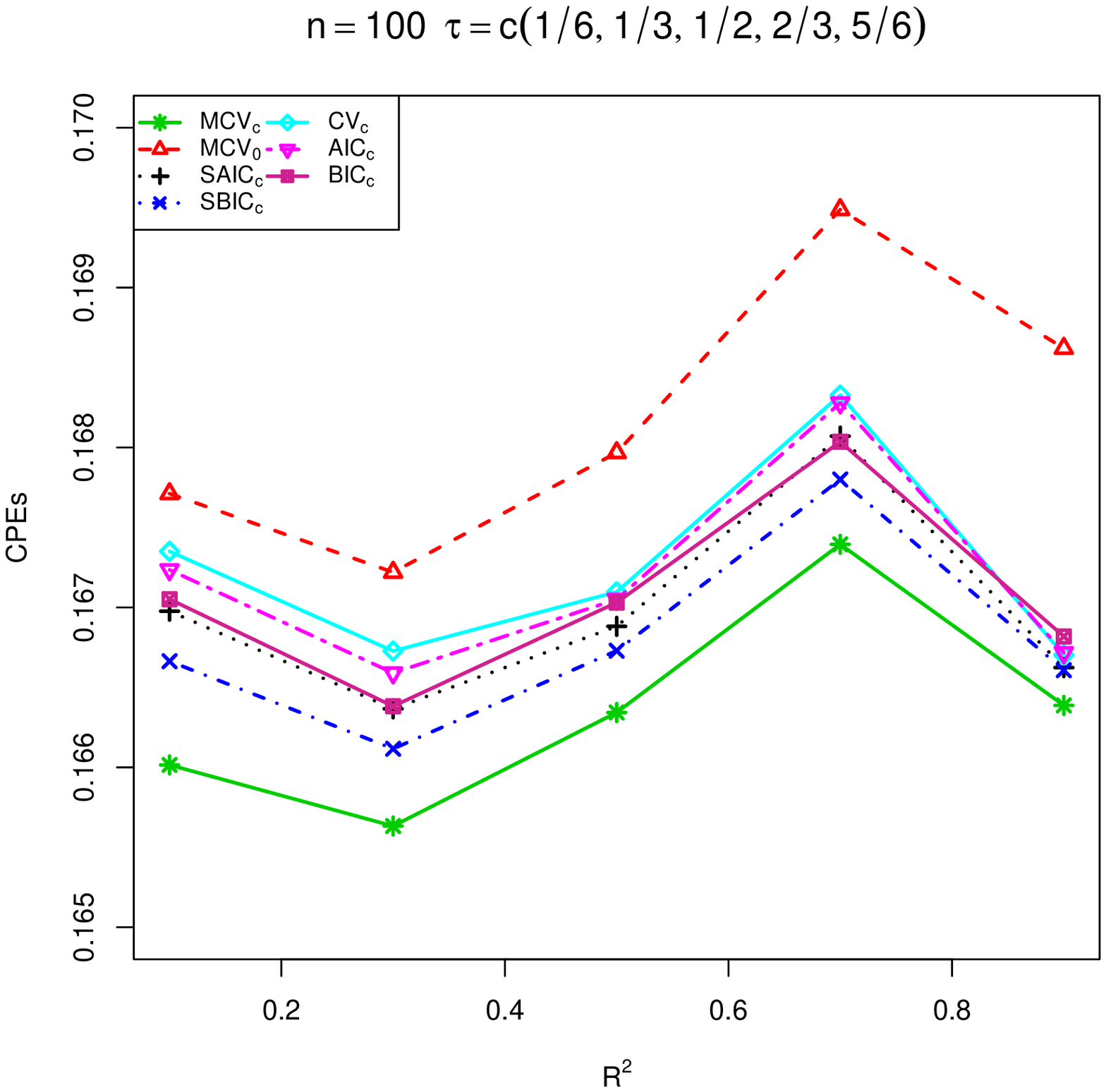}}
\subfigure{
\includegraphics[width=8cm,height=6.8cm]{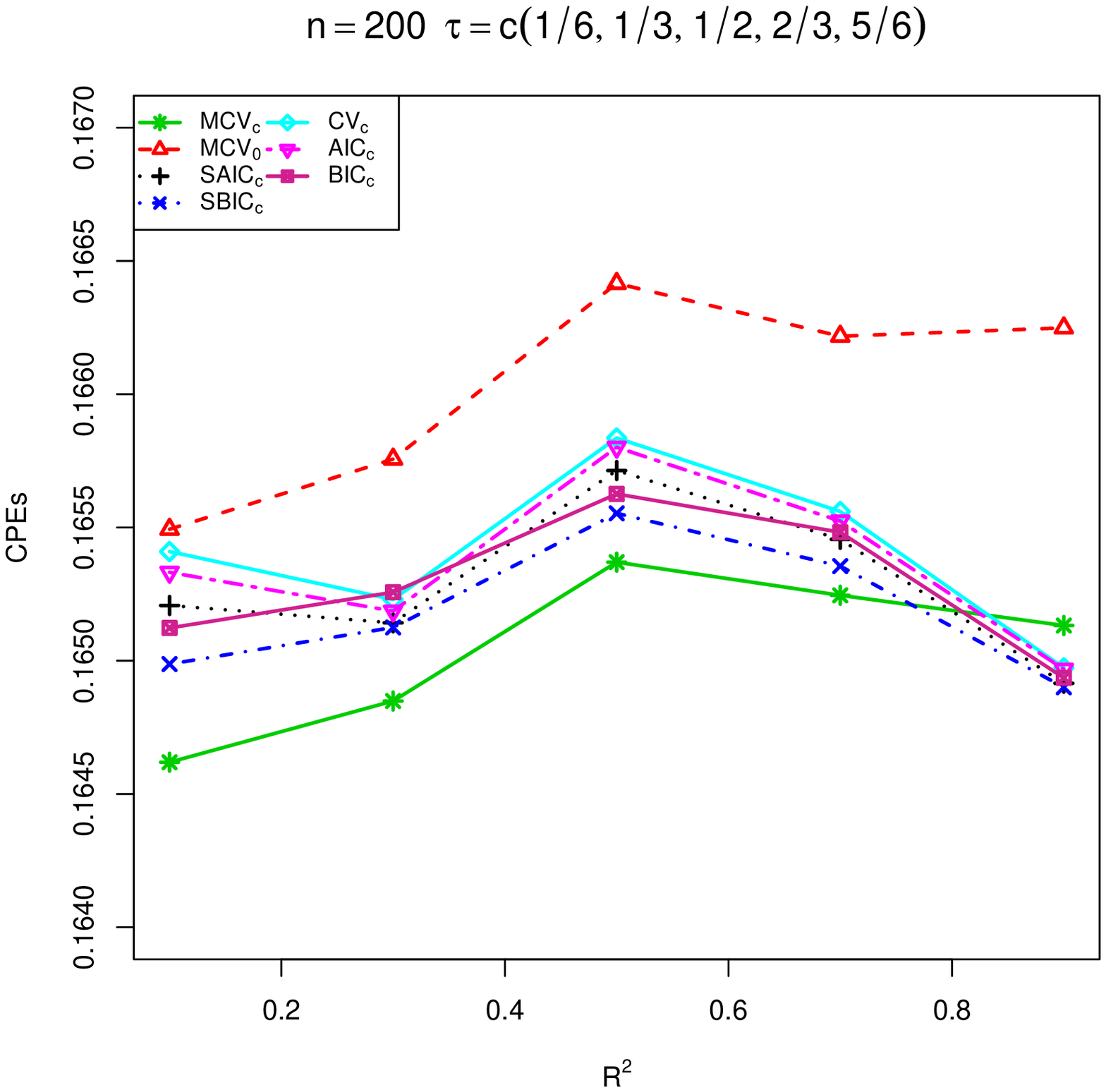}}
\caption{$\mathrm{CPEs}$ of estimators for \textbf{Case} $\boldsymbol{1}$ of \textbf{Setting} $\boldsymbol{4}$.}
\label{Fig7}
\end{figure}
\begin{figure}[b]
\centering
\subfigure{
\includegraphics[width=8cm,height=6.8cm]{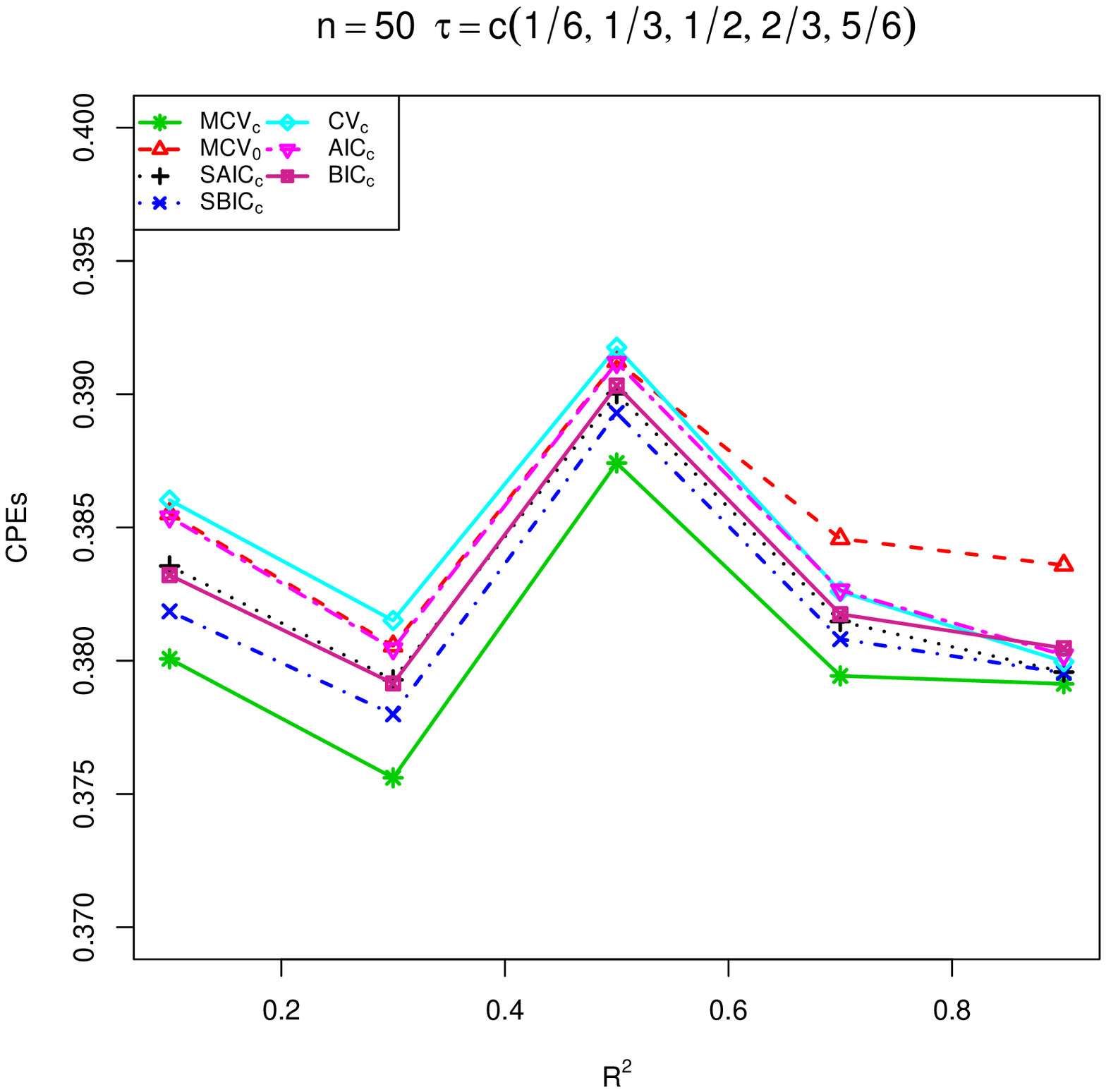}}
\subfigure{
\includegraphics[width=8cm,height=6.8cm]{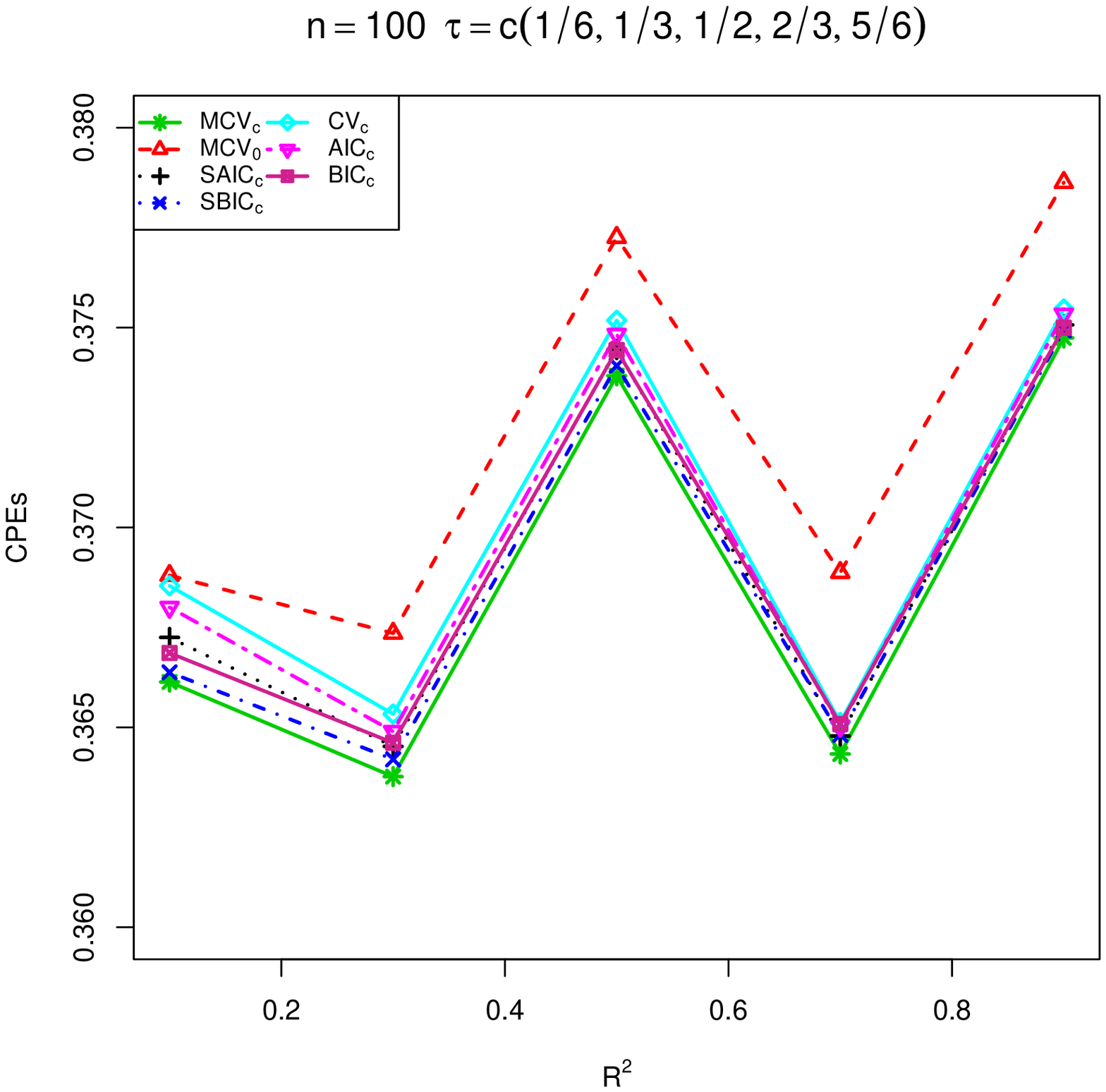}}
\subfigure{
\includegraphics[width=8cm,height=6.8cm]{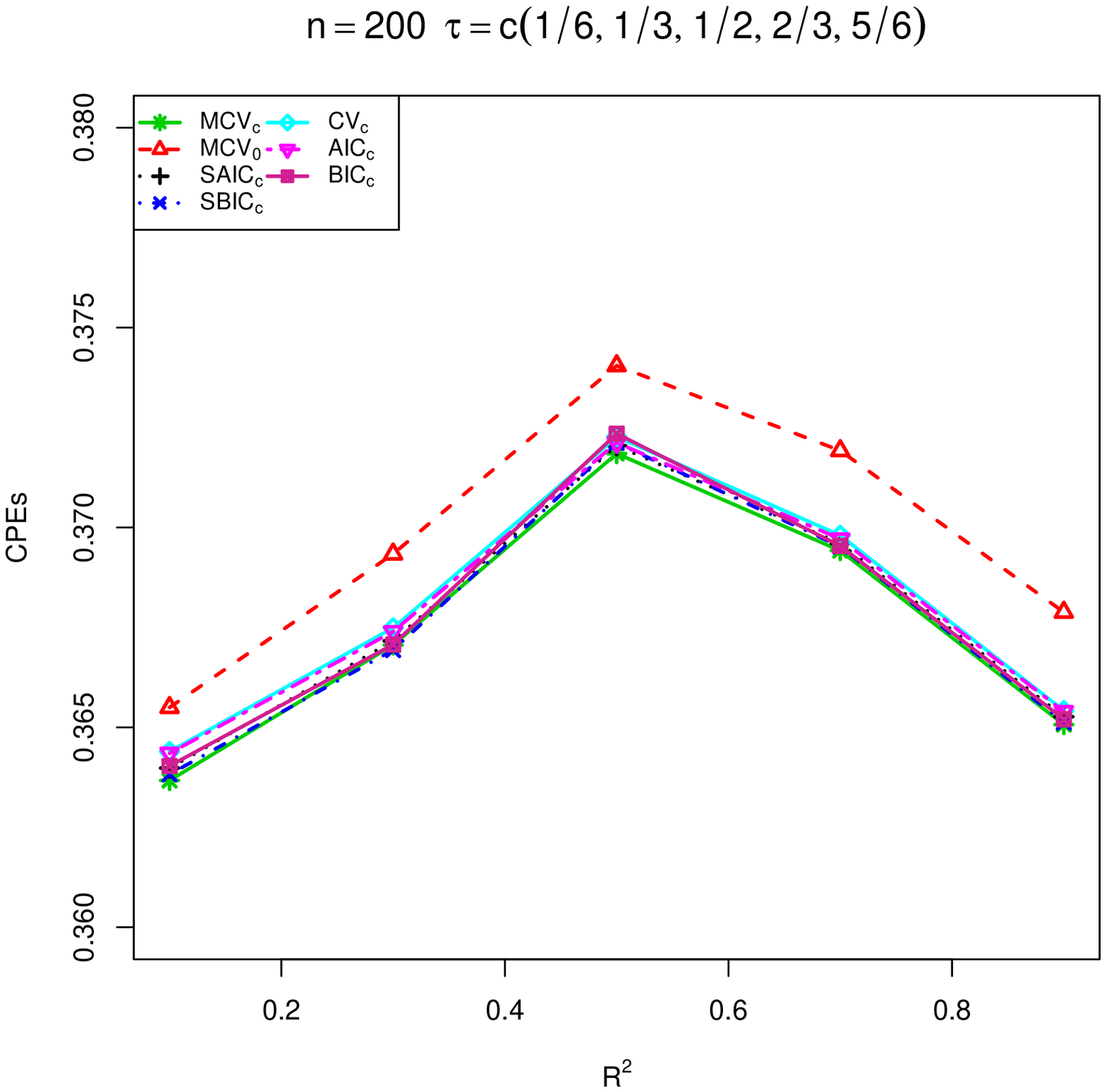}}
\caption{$\mathrm{CPEs}$ of estimators for \textbf{Case} $\boldsymbol{2}$ of \textbf{Setting} $\boldsymbol{4}$.}
\label{Fig8}
\end{figure}
\begin{figure}[b]
\centering
\subfigure{
\includegraphics[width=8cm,height=6.8cm]{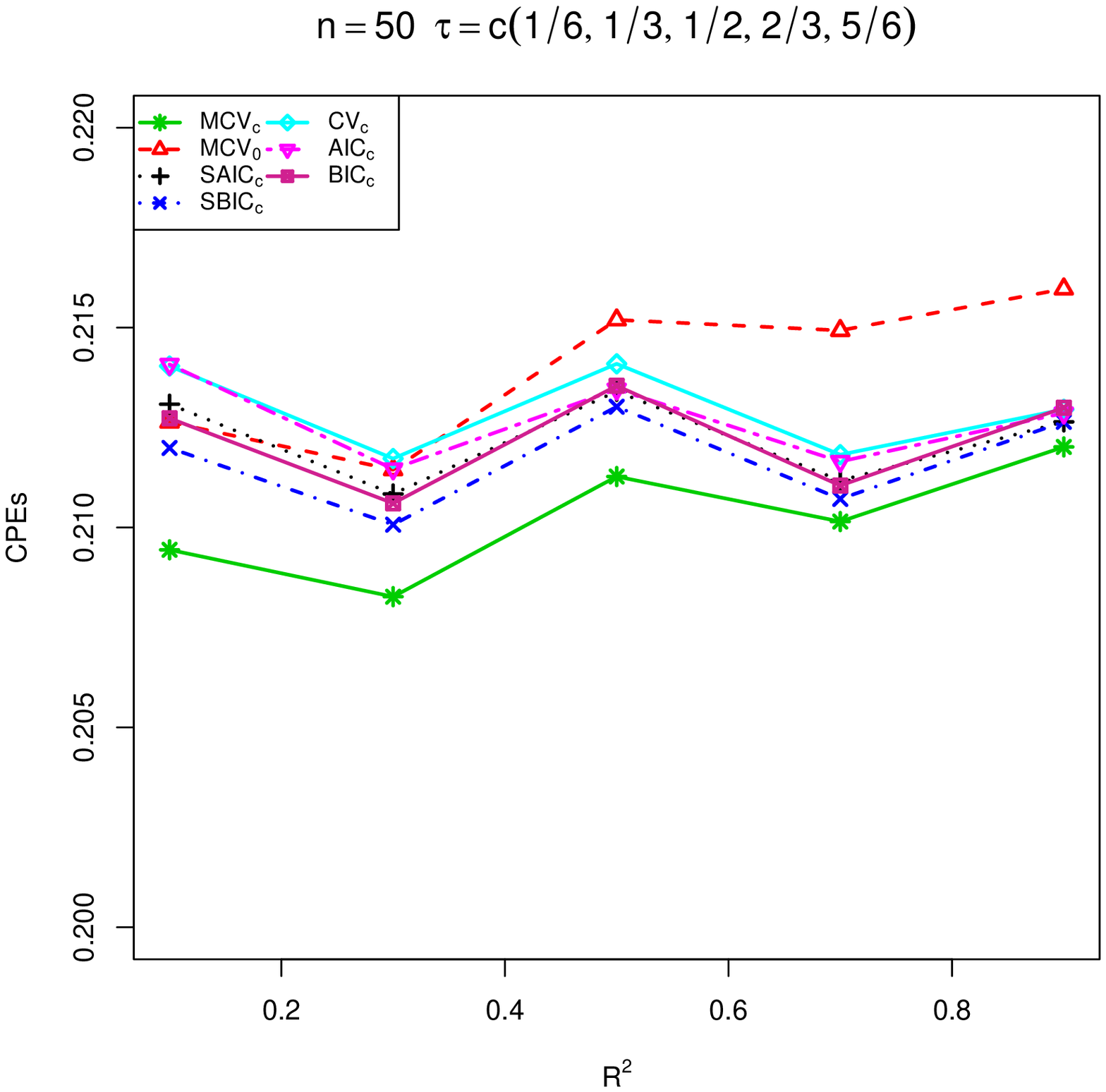}}
\subfigure{
\includegraphics[width=8cm,height=6.8cm]{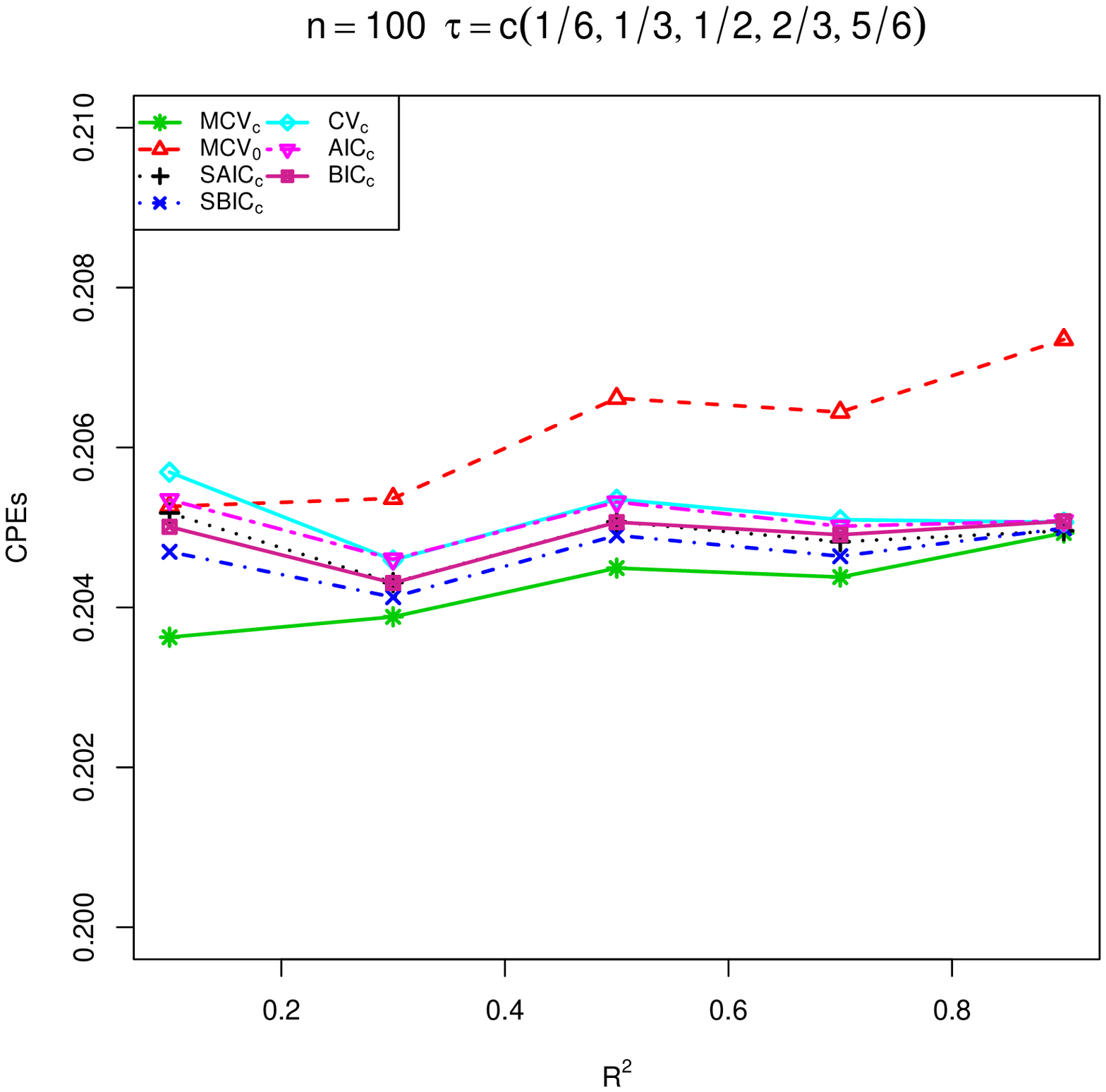}}
\subfigure{
\includegraphics[width=8cm,height=6.8cm]{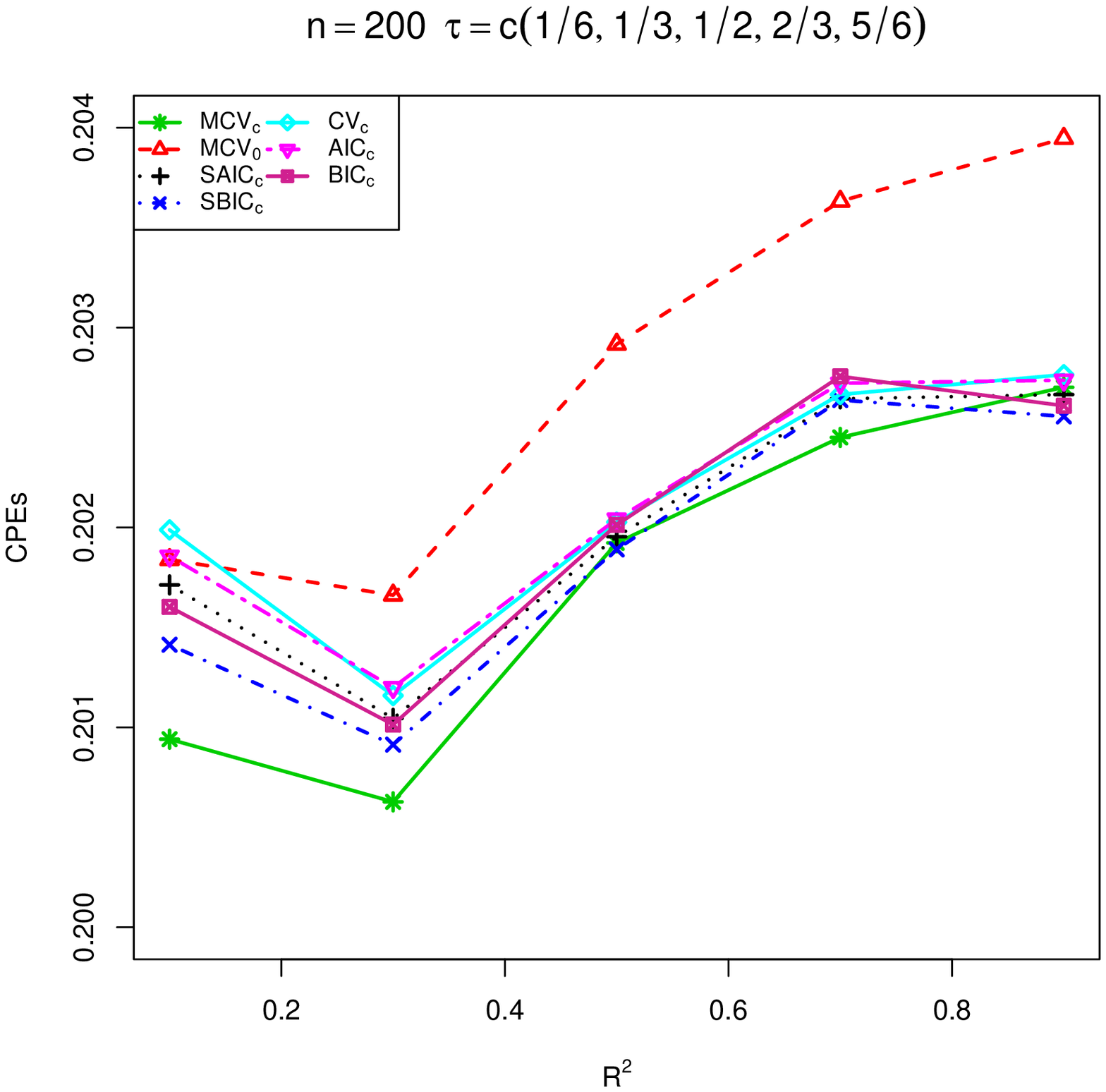}}
\caption{$\mathrm{CPEs}$ of estimators for \textbf{Case} $\boldsymbol{3}$ of \textbf{Setting} $\boldsymbol{4}$.}
\label{Fig9}
\end{figure}
\section{Real data example}\label{sec:Real}
In this section, we apply the methods examined in Section \ref{sec:simu} to the two real example investigated in \cite{lu2015jackknife}. One is about the excess stock returns, the other one is about wages.
\subsection{Quantile forecast of excess stock returns}
In financial risk management, Value at Risk (VaR) is widely used to measure the
risk of loss on a specific portfolio of financial assets.
It can be calculated by the quantiles of stock returns.
In this subsection, we apply our MCV$_{c}$ estimators to predict the quantiles of stock returns.

The first data example, size $ T = 672 $, is monthly from January 1950 to December 2005. The dependent variable $y$ is the excess stock returns, defined as the monthly returns of S$\&$P 500 index minus the risk-free rate. The detailed explanations of $12$ regressors in the data set were shown in \cite{Sainan2014Robustify}.
We sort the $12$ regressions based on the absolute value of their correlation with the dependent variable and then construct 12 candidate nested models with regressors
$\{1,x_{1}\}$, $\{1,x_{1},x_{2}\}$, $\cdots$, $\{1,x_{1},x_{2},\cdots,x_{12}\}$, respectively.
Adopting the in-sample size ($T_{1}$) of $72$, $96$, $120$, $144$, and $180$,
we construct one-period-ahead prediction of the quantiles of excess stock returns and compare the above forecast methods with the out-of-sample $\mathrm{CPE}=\frac{1}{(T-T_1)K}\sum_{t=T_{1}+1}^{T}\sum_{k=1}^{K}\rho_{\tau_{k}}\left(y_{t}-\hat{y}_{t,k}\right)$, where $\hat{y}_{t,k}$ is the one-period ahead prediction of the $\tau_{k}^{th}$ quantile of the excess return at time $t$ using data from period $t$ to period $t-T_{1}+1$.
The quantiles are taken as $\tau_{1}=0.05$,
$\tau_{2}=0.5$ and $\tau_{3}=0.95$ with $K=3$.

The results of the out-of-sample $\mathrm{CPEs}$ are presented in Table \ref{tab:stock}. It is clear that MCV$_{c}$ dominates all other model averaging and model selection methods.
But the performance of MCV$_{0}$ is different from the simulation results.
It produces the second best estimates. Other model averaging and model selection estimators can be the least preferred estimators.
\begin{table}[h]
  \centering
  \caption{$\mathrm{CPEs}$ of prediction for the quantiles of the excess stock returns$(\times 10^{-2})$}
  \vspace{0.3cm}
    \begin{tabular}{c|ccccccc}
    \toprule
    $T_{1}$ & MCV$_{c}$ & MCV$_{0}$ & SAIC$_{c}$ & SSIC$_{c}$ & CV$_{c}$ & AIC$_{c}$ & SIC$_{c}$ \\
    \midrule
    72    & 0.8033$^\sharp$  & 0.8435  & 0.9302  & 0.9297  & 0.9299  & 0.9307$^\flat$  & 0.9300  \\
    96    & 0.8240$^\sharp$  & 0.8414  & 0.9229  & 0.9230$^\flat$  & 0.9229  & 0.9230  & 0.9231  \\
    120   & 0.8602$^\sharp$  & 0.8636  & 0.9387$^\flat$  & 0.9387  & 0.9387  & 0.9387  & 0.9387  \\
    144   & 0.8719$^\sharp$  & 0.8726  & 0.9638$^\flat$  & 0.9638  & 0.9638  & 0.9638  & 0.9638  \\
    180   & 0.8459$^\sharp$  & 0.8678  & 0.9102  & 0.9099  & 0.9110$^\flat$  & 0.9103  & 0.9103  \\
    \bottomrule
    \end{tabular}%
  \label{tab:stock}%
\end{table}%
\subsection{Quantile forecast of wages}
In labor economics, the quantiles of wages are often used to characterize wage inequality (see, e.g. \cite{Angrist2009Mostly}, Chapter 7).
In this subsection, our new averaging estimators are applied to predict it.

We consider the sample of the US Current Population Survey for the year 1976 from \cite{Wooldridge2003Introductory}. The dependent variable is the logarithm of average hourly earnings with the sample size $n=526$ and $20$ regressors.
We sort these regressions based on their absolute value of the correlation with the dependent variable and construct 10 candidate nested models using the first 10 regressions. Then we randomly divide the data into a training sample $\{\boldsymbol{x}_{s},y_{s}\}_{s=1}^{n_{1}}$ and an evaluation sample $\{\boldsymbol{x}_{t},y_{t}\}_{t=1}^{n_{2}}$ containing $n_{1}$ and $n_{2}=n-n_{1}$ observations, respectively. We set different training sample sizes $n_{1}=100$, $150$, and $200$.
We evaluate these methods with respect to the out-of-sample $\mathrm{CPE}=\frac{1}{n_2K}\sum_{t=1}^{n_{2}}\sum_{k=1}^{K}\rho_{\tau_{k}}\left(y_{t}-\hat{y}_{t}\right)$, where $\hat{y}_{t}$ is the predicted value of $y_{t}$.  The quantiles are taken as $\tau_{1}=0.05$, $\tau_{2}=0.5$ and $\tau_{3}=0.95$ with $K=3$. We repeat the exercise of splitting the sample $100$ times, and a comparison of the methods is based on the $\mathrm{CPEs}$, which is the average of the $\mathrm{PE}$ across the $100$ replications.

The results of the out-of-sample $\mathrm{CPEs}$ are presented in Table \ref{tab:wage}. It is found that MCV$_{c}$ always yields the best results.
MCV$_{0}$ is inferior to other model averaging and model selection estimators.
The results from this real data analysis are consistent with those obtained from the simulation study.
\begin{table}[h]
  \centering
  \caption{$\mathrm{CPEs}$ of prediction for the quantiles of wages$(\times 10^{-1})$}
  \vspace{0.3cm}
    \begin{tabular}{c|ccccccc}
    \toprule
    $n_{1}$ & MCV$_{c}$ & MCV$_{0}$ & SAIC$_{c}$ & SSIC$_{c}$ & CV$_{c}$ & AIC$_{c}$ & SIC$_{c}$ \\
    \midrule
    100   & 0.6657$^\sharp$  & 0.8380$^\flat$  & 0.8285  & 0.8284  & 0.8290  & 0.8286  & 0.8291  \\
    150   & 0.5717$^\sharp$  & 0.8089$^\flat$  & 0.7996  & 0.7994  & 0.7998  & 0.7997  & 0.7995  \\
    200   & 0.4889$^\sharp$  & 0.7949$^\flat$  & 0.7900  & 0.7898  & 0.7901  & 0.7901  & 0.7901  \\
    \bottomrule
    \end{tabular}%
  \label{tab:wage}%
\end{table}%
\section{Conclusions} \label{sec:con}
In this paper, we propose a frequentist model averaging method for CQR.
We have demonstrated a model averaging estimator with weights obtained based on a cross-validation procedure that results in an optimal asymptotic property and has excellent finite sample properties relative to other methods, including MCV$_{0}$.
We have shown that model averaging represents a credible alternative to model selection in this context, and deserves further attention from both applied and theoretical researchers.

There are two main works may be extended. For example, while our model averaging method is derived based on a delete-one CV method, other model averaging methods that have been developed in the literature; for e.g.,
MSE minimization or Kullback-Leibler type methods \citep{liang2011optimal,zhang2016optimal,zhang2015kullback} may be usefully extended to the context here.
Besides, these existing methods above were built on uniform weights for loss function in composite quantile regression.
\cite{Bradic2011Penalized} pointed out that using uniform weights for loss function in composite quantile regression may be not optimal in general. Therefore, \cite{Jiang2012Oracle} considered a weighted CQR estimation approach and studied model selection for nonlinear models with a diverging number of parameters.
Due to the effectiveness and robustness of weighted CQR estimation. The Weighted CQR estimation is enjoying rising popularity. More works for it studies can be found in
\citep{Jiang2014Weighted,Liu2015Weighted,Guo2017Robust,Jiang2017Weighted}. Our JMA method can be developed in this literature.
These remain for future research.

\appendix
\renewcommand\thesection{\appendixname~\Alph{section}}
\renewcommand\theequation{\Alph{section}.\arabic{equation}}
\renewcommand\thelem{\Alph{section}.\arabic{lem}}
\section{Assumptions and Lemmas}
\label{sec:AL}
For $m=1,\ldots,M$, let $u_{ki(m)}=\mu_{i}-\left(1,\boldsymbol{x}_{i(m)}^{T}\right)\left(b_{k(m)}^{*},\bbeta_{(m)}^{*T}\right)^{T}$ and $\psi_{\tau}(\varepsilon)=\tau-\textbf{1}\{\varepsilon\leq 0\}$. Define
\begin{eqnarray}
A_{k(m)}&=&\mathrm{E}\left\{f(-u_{ki(m)}\mid \boldsymbol{x}_{i})\left(1,\boldsymbol{x}_{i(m)}^{T}\right)^{T}\left(1,\boldsymbol{x}_{i(m)}^{T}\right)\right\}\non\\
&=&\mathrm{E}\left\{f_{y\mid \boldsymbol{x}}\left(\left(1,\boldsymbol{x}_{i(m)}^{T}\right)\left(b_{k(m)}^{*},\bbeta_{(m)}^{*T}\right)^{T}\right)
\left(1,\boldsymbol{x}_{i(m)}\boldsymbol{x}_{i(m)}^{T}\right\}\left(1,\boldsymbol{x}_{i(m)}^{T}\right)\right\},\non\\
B_{k(m)}&=&\mathrm{E}\left\{\psi_{\tau_{k}}(\varepsilon_{i}+u_{ki(m)})^{2}\left(1,\boldsymbol{x}_{i(m)}^{T}\right)^{T}\left(1,\boldsymbol{x}_{i(m)}^{T}\right)\right\},\non
\end{eqnarray}
$$
A_{(m)}=
\begin{pmatrix}
f(-u_{1i(m)}\mid \boldsymbol{x}_{i})&&&&\\
&\ddots&&\\
&&f(-u_{Ki(m)}\mid \boldsymbol{x}_{i})&\\
&&&\sum_{k=1}^{K}f(-u_{ki(m)}\mid \boldsymbol{x}_{i})\boldsymbol{x}_{i(m)}\boldsymbol{x}_{i(m)}^{T}
\end{pmatrix},
$$
$$
B_{(m)}=
\begin{pmatrix}
\psi_{\tau_{1}}\left(\varepsilon_{i}+u_{1i(m)}\right)^{2}&&&&\\
&\ddots&&\\
&&\psi_{\tau_{K}}\left(\varepsilon_{i}+u_{Ki(m)}\right)^{2}&\\
&&&\left\{\sum_{k=1}^{K}\psi_{\tau_{k}}\left(\varepsilon_{i}+u_{ki(m)}\right)\right\}^{2}
\boldsymbol{x}_{i(m)}\boldsymbol{x}_{i(m)}^{T}
\end{pmatrix},
$$
\begin{eqnarray}
\bar{A}_{(m)}&=&\mathrm{E}\left(A_{0(m)}\right),
\bar{B}_{(m)}=\mathrm{E}\left(B_{0(m)}\right)~
\text{and}~
V_{(m)}=\bar{A}_{(m)}^{-1}\bar{B}_{(m)}\bar{A}_{(m)}^{-1}.\non
\end{eqnarray}

Let $\overline{k}=\max_{1\leq m\leq M}\tilde{k}_{m}$ and $\|A\|=\{\text{tr}(AA^{T})\}^{1/2}$ for an real matrix $A$.
For a symmetric matrix $A$ , we adopt $\lambda_{\max}(A)$ and $\lambda_{\min}(A)$ to denote its largest and smallest eigenvalues, respectively.
We require the following uniformly integrable assumptions:
\begin{condition}
\label{a1}%
\emph{(i)} $(y_{i},\boldsymbol{x}_{i}), i=1,\ldots,n,$ are $\mathrm{IID}$ such that (\ref{eq2.1}) holds.

\emph{(ii)} $P(\varepsilon_{i}\leq b_{k}\mid\boldsymbol{x}_{i})=\tau_{k} \ a.s.$.

\emph{(iii)} $\max_{1\leq k\leq K}\mathrm{E}\left(\mu_{i}+b_{k}\right)^{4}<\infty$ and $\mathrm{sup}_{j\geq1}\mathrm{E}(x_{ij}^{8})\leq c_{\boldsymbol{x}}$ for some $c_{\boldsymbol{x}}<\infty$.
\end{condition}
\begin{condition}
\label{a2}
\emph{(i)} $f_{y\mid \boldsymbol{x}}(\cdot\mid \boldsymbol{x}_{i})$ continuous over its support a.s.. $f_{y\mid \boldsymbol{x}}(\cdot\mid \boldsymbol{x}_{i})$ and its first order derivative are bounded above by a finite constant $\overline{c}_{f}$.

\emph{(ii)} There exist constants $\underline{c}_{A_{(m)}}$ and $\overline{c}_{A_{(m)}}$ that may depend on $k_{m}$ such that $0<\underline{c}_{A_{(m)}}\leq \lambda_{\min}(A_{k(m)}) \leq \lambda_{\max}(A_{k(m)})\leq\overline{c}_{A_{(m)}} <\infty$ and $\overline{c}_{f} \lambda_{\max}\left[\mathrm{E}\left\{\left(1,\boldsymbol{x}_{i(m)}^{T}\right)^{T}
\left(1,\boldsymbol{x}_{i(m)}^{T}\right)\right\}\right]\leq \overline{c}_{A_{(m)}} $
for all $k=1,\cdots,K$.

\emph{(iii)} There exist constants $\underline{c}_{\bar{A}_{(m)}}$, $\bar{c}_{\bar{A}_{(m)}}$,   $\underline{c}_{\bar{B}_{(m)}}$ and $\bar{c}_{\bar{B}_{(m)}}$that may depend on $\tilde{k}_{m}$ and $K$ such that
$0<\underline{c}_{\bar{A}_{(m)}}\leq \lambda_{\min}(\bar{A}_{(m)})\leq \lambda_{\max}(\bar{A}_{(m)})\leq \overline{c}_{\bar{A}_{(m)}}$ and
$0<\underline{c}_{\bar{B}_{(m)}}\leq \lambda_{\min}(\bar{B}_{(m)})\leq \lambda_{\max}(\bar{B}_{(m)})\leq \overline{c}_{\bar{B}_{(m)}}$.

\emph{(iv)} $\bar{c}_{A(m)}K/(K+\tilde{k}_{m})=O(\underline{c}_{A(m)}^{2})$
 \label{con:Ps}
\end{condition}
\begin{condition}
\label{a3}
Let $\underline{c}_{A}=\mathrm{min}_{1\leq m\leq M}~\underline{c}_{A_{(m)}}$,
$\overline{c}_{A}=\mathrm{max}_{1\leq m\leq M}~\overline{c}_{A_{(m)}}$,
$\underline{c}_{\bar{A}}=\mathrm{min}_{1\leq m\leq M}~\underline{c}_{\bar{A}_{(m)}}$,
$\overline{c}_{\bar{A}}=\mathrm{max}_{1\leq m\leq M}~\overline{c}_{\bar{A}_{(m)}}$,
$\underline{c}_{\bar{B}}=\mathrm{min}_{1\leq m\leq M}\underline{c}_{\bar{B}_{(m)}}$
and
$\overline{c}_{\bar{B}}=\mathrm{max}_{1\leq m\leq M}~\overline{c}_{\bar{B}_{(m)}}$.

\emph{(i)} $\bar{c}_{A}\bar{k}^{4}K^{2}/(n\underline{c}_{\bar{B}})=o(1)$, $\overline{k}^{4}K^{4}(\log n)^{4}/(n\underline{c}_{\bar{B}}^{2})=o(1)$,
$\left(K+\bar{k}\right)^{1/2}\bar{k}^{3/2}/\left(n^{1/2}\underline{c}_{A}\right)=o(1)$,
and
$0<\underline{c}\leq\underline{c}_{A}\leq \overline{c}_{A}\leq \overline{c}<\infty$.

\emph{(ii)} $\left\{L\bar{k}^{4}\log n/(n\underline{c}_{A}^{2})\right\}^{1/2}=o(1)$ and $nMn^{-LK\bar{k}\underline{c}_{A}\underline{c}_{\bar{A}}^{2}/(2\bar{c}_{\bar{A}}\bar{c}_{\bar{B}})}=o(1)$ for a sufficiently large constant $L$.
\end{condition}

\begin{remark}
Assumptions \emph{\ref{a1}-\ref{a3}} are regular conditions on quantile regression somewhat different from the assumptions in \emph{\cite{lu2015jackknife}} because we consider multiple quantiles. Assumption \emph{\ref{a1}} and the Part \emph{(ii)} of Assumption \emph{\ref{a2}} are same assumptions in \emph{\cite{lu2015jackknife}} and there are some detailed explanations for reference.

Next, we explain other assumptions.
The Part(i) of Assumption \emph{\ref{a2}} is a commonly used condition for inference in the quantile regression \emph{(}e.g., see \emph{\cite{koenker2005quantile}} and \emph{\cite{Wang2009INFERENCE}}\emph{)}.
When $K=1$, the Part \emph{(iii)} of Assumption \emph{\ref{a2}} is the assumption in \emph{\cite{lu2015jackknife}}, which is often assumed in typical QR models.
We suppose $\underline{c}_{\bar{A}_{(m)}}$, $\bar{c}_{\bar{A}_{(m)}}$,   $\underline{c}_{\bar{B}_{(m)}}$ and $\bar{c}_{\bar{B}_{(m)}}$ that may not only depend on $\tilde{k}_{m}$ but also associate with $K$.
When $K=1$, the Part\emph{(iv)} of Assumption \emph{\ref{a2}} is equivalent to the corresponding hypothesis about $\bar{c}_{A(m)}$, $\bar{c}_{B(m)}$ and $\underline{c}_{A(m)}$.

Assumption \emph{\ref{a3}} is also the general form hypothesis of Assumption \emph{A.3} in \emph{\cite{lu2015jackknife}} except the condition $0<\underline{c}\leq\underline{c}_{A}\leq \overline{c}_{A}\leq \overline{c}<\infty$, imposing restrictions on the largest number of regression in each model, the number of composite quantiles and candidate models, and the constants $\underline{c}_{A}$, $\bar{c}_{A}$,  $\underline{c}_{\bar{A}}$,
$\bar{c}_{\bar{A}}$, $\underline{c}_{\bar{B}}$ and $\bar{c}_{\bar{B}}$.
\end{remark}

\begin{remark}
Note that these assumptions do not indicate that the real model cannot be in the candidate model set.
If the $m^{th}$ model is true model, we have $$u_{ki(m)}^{b}=\mu_{i}+b_{k}-\left(1,\boldsymbol{x}_{i(m)}^{T}\right)\left(b_{k(m)}^{*},\bbeta_{(m)}^{*T}\right)^{T}=0$$
It implies that $-u_{ki(m)}=b_{k}$. The Part \emph{(ii)-(iii)} of Assumption \emph{\ref{a2}} are still hold by the condition $f(b_{k})>0$, which is a general assumption in \emph{\cite{koenker1978quantile}}.
\end{remark}

The following theorem studies the asymptotic behavior of estimators.
\begin{lem}
\label{lem:distr}
Suppose Assumptions \emph{\ref{a1}-\ref{a3}} hold. Let $C_{(m)}$ denote an $l_{m}\times (\tilde{k}_{m}+K)$ matrix such that $C_{0}=\mathrm{lim}_{n\rightarrow\infty}C_{(m)}C_{(m)}^{T}$ exists and is positive definite, where $l_{m}\in\{1,2,\cdots,\tilde{k}_{m}+K\}$. Then\\
\emph{(i)} $\left\|\left(\hat{b}_{1(m)},\cdots,\hat{b}_{K(m)},\hat{\bbeta}_{(m)}^{T}\right)^{T}-\left(b_{1(m)}^{*},\cdots,b_{K(m)}^{*},\bbeta_{(m)}^{*T}\right)^{T}\right\|=O_{p}\left(\sqrt{n^{-1}(K+\tilde{k}_{m})}\right) ;$\\
\emph{(ii)}
$\sqrt{n}C_{(m)}V_{(m)}^{-1/2}\left\{\left(\hat{b}_{1(m)},\cdots,\hat{b}_{K(m)},\hat{\bbeta}_{(m)}^{T}\right)^{T}-\left(b_{1(m)}^{*},\cdots,b_{K(m)}^{*},\bbeta_{(m)}^{*T}\right)^{T}\right\}\overset{d}{\rightarrow } N(0,C_{0}).$
\end{lem}
\begin{proof}
See \ref{sec:distr}.
\end{proof}
\begin{lem}
\label{lem:order}
Suppose Assumptions \emph{\ref{a1}-\ref{a3}} hold. Then\\
\emph{(i)} $\mathrm{max}_{1\leq i\leq n}\mathrm{max}_{1\leq m\leq M}\sum_{k=1}^{K}\left\|\left(\hat{b}_{ki(m)}-b_{k(m)}^{*},\left(\hat{\bbeta}_{i(m)}-\bbeta_{(m)}^{*}\right)^{T}\right)^{T}\right\|^{2}=O_{p}\left(n^{-1}K\bar{k}\log n \right);$\\
\emph{(ii)} $\mathrm{max}_{1\leq m\leq M}\sum_{k=1}^{K}\left\|\left(\hat{b}_{k(m)}-b_{k(m)}^{*},\left(\hat{\bbeta}_{(m)}-\bbeta_{(m)}^{*}\right)^{T}\right)^{T}\right\|^{2}=O_{p}\left(n^{-1}K\bar{k}\log n \right).$
\end{lem}
\begin{proof}
See \ref{sec:order}.
\end{proof}
\section{Proof of Lemma \ref{lem:distr}}\label{sec:distr}
(i) Let $a_{n}=\sqrt{n^{-1}(K+\tilde{k}_{m})}$. Let $\boldsymbol{v}_{(m)}=(v_{b_{1(m)}},\cdots,v_{b_{K(m)}},\boldsymbol{v}_{\bbeta(m)}^{T})^{T}\in\mathbb{R}^{\tilde{k}_{m}+K}$ such that $\|\boldsymbol{v}_{(m)}\|= C$ where $C$ is a large enough constant. We want to show that for any given $\varepsilon>0$ there is a large enough constant $C$ such that, for large $n$ we have
\be
&&\mathrm{P}\left(\inf_{\|\boldsymbol{v}_{(m)}\|= C}Q_{n(m)}
\left(\left(b_{1(m)}^{*},\cdots,b_{K(m)}^{*},\bbeta_{(m)}^{*T}\right)^{T}+a_{n}\boldsymbol{v}_{(m)}\right)\right.\non\\
&&\left.~~~~>Q_{n(m)}
\left(\left(b_{1(m)}^{*},\cdots,b_{K(m)}^{*},\bbeta_{(m)}^{*T}\right)^{T}\right)\right)\geq1-\varepsilon.
\ee
This implies that with probability approaching $1$ there is
a local minimum $\left(\hat{b}_{1(m)},\cdots,\hat{b}_{K(m)},\hat{\bbeta}_{(m)}^{T}\right)^{T}$ in the ball $\left\{\left(b_{1(m)}^{*},\cdots,b_{K(m)}^{*},\bbeta_{(m)}^{*T}\right)^{T}+a_{n}\boldsymbol{v}_{(m)}:\|\boldsymbol{v}_{(m)}\|\leq C\right\}$
such that
$$\left\|\left(\hat{b}_{1(m)},\cdots,\hat{b}_{K(m)},\hat{\bbeta}_{(m)}^{T}\right)^{T}-\left(b_{1(m)}^{*},\cdots,b_{K(m)}^{*},\bbeta_{(m)}^{*T}\right)^{T}\right\|=O_{p}(a_{n}).$$ It is also the global minimum
by the convexity of $Q_{n(m)}$.

Let $\boldsymbol{v}_{k(m)}=(v_{b_{k(m)}},\boldsymbol{v}_{\bbeta(m)}^{T})^{T}$ and $u_{ki(m)}=\mu_{i}-\left(1,\boldsymbol{x}_{i(m)}^{T}\right)\left(b_{k(m)}^{*},\bbeta_{(m)}^{*T}\right)^{T}$.
Then by Knight's identity
\be
Z_{n(m)}(\boldsymbol{v}_{(m)})&=&Q_{n(m)}
\left(\left(b_{1(m)}^{*},\cdots,b_{K(m)}^{*},\bbeta_{(m)}^{*T}\right)^{T}+a_{n}\boldsymbol{v}_{(m)}\right)\non\\
&&-Q_{n(m)}
\left(\left(b_{1(m)}^{*},\cdots,b_{K(m)}^{*},\bbeta_{(m)}^{*T}\right)^{T}\right)\non\\
&=&\sum_{i=1}^{n}\sum_{k=1}^{K}
\left[\rho_{\tau_{k}}\left(\varepsilon_{i}+u_{ki(m)}-a_{n}\left(1,\boldsymbol{x}_{i(m)}^{T}\right)\boldsymbol{v}_{k(m)}\right)
-\rho_{\tau_{k}}\left(\varepsilon_{i}+u_{ki(m)}\right)\right]\non\\
&=&-a_{n}\sum_{i=1}^{n}\sum_{k=1}^{K}\psi_{\tau_{k}}\left(\varepsilon_{i}+u_{ki(m)}\right)\left(1,\boldsymbol{x}_{i(m)}^{T}\right)\boldsymbol{v}_{k(m)}
\non\\
&&+\sum_{i=1}^{n}\sum_{k=1}^{K}\int_{0}^{a_{n}\left(1,\boldsymbol{x}_{i(m)}^{T}\right)\boldsymbol{v}_{k(m)}}\alpha_{ki(m)}(s)\mathrm{d}s\non\\
&=&-a_{n}\sum_{i=1}^{n}\sum_{k=1}^{K}\psi_{\tau_{k}}\left(\varepsilon_{i}+u_{ki(m)}\right)\left(1,\boldsymbol{x}_{i(m)}^{T}\right)\boldsymbol{v}_{k(m)}
\non\\
&&+\sum_{i=1}^{n}\sum_{k=1}^{K}\mathrm{E}\left\{\int_{0}^{a_{n}\left(1,\boldsymbol{x}_{i(m)}^{T}\right)\boldsymbol{v}_{k(m)}}\alpha_{ki(m)}(s)\mathrm{d}s|\boldsymbol{x}_{i}\right\}\non\\
\non\\
&&+\sum_{i=1}^{n}\sum_{k=1}^{K}\left[\int_{0}^{a_{n}\left(1,\boldsymbol{x}_{i(m)}^{T}\right)\boldsymbol{v}_{k(m)}}\alpha_{ki(m)}(s)\mathrm{d}s
\right.\non\\
&&\left.-\mathrm{E}\left\{\int_{0}^{a_{n}\left(1,\boldsymbol{x}_{i(m)}^{T}\right)\boldsymbol{v}_{k(m)}}\alpha_{ki(m)}(s)\mathrm{d}s|\boldsymbol{x}_{i}\right\}\right]\non\\
&=&Z_{n(m),1}(\boldsymbol{v}_{(m)})+Z_{n(m),2}(\boldsymbol{v}_{(m)})+Z_{n(m),3}(\boldsymbol{v}_{(m)}),
\ee
where $\alpha_{ki(m)}(s)=\boldsymbol{1}\{\varepsilon_{i}+u_{ki(m)}\leq s\}-\boldsymbol{1}\{\varepsilon_{i}+u_{ki(m)}\leq 0\}$.
The first order condition for the population minimization
problem (\ref{problem:2}) implies that
\be\label{d:0}
\mathrm{E}\left\{\left(\psi_{\tau_{1}}\left(\varepsilon_{i}+u_{1i(m)}\right),\cdots,
\psi_{\tau_{K}}\left(\varepsilon_{i}+u_{Ki(m)}\right),
\left\{\sum_{k=1}^{K}\psi_{\tau_{k}}\left(\varepsilon_{i}+u_{ki(m)}\right)\right\}\boldsymbol{x}_{i(m)}^{T}\right)^{T}\right\}=\boldsymbol{0},
\ee
 where $\psi_{\tau}(\varepsilon)=\tau-\boldsymbol{1}\{\varepsilon\leq0\}$.
By Assumption \ref{a2}(ii),
\be\mathrm{E}|Z_{n(m),1}(\boldsymbol{v}_{(m)})|^{2}
&=&\mathrm{E}\left|-a_{n}\sum_{i=1}^{n}\sum_{k=1}^{K}\psi_{\tau_{k}}\left(\varepsilon_{i}+u_{ki(m)}\right)\left(1,\boldsymbol{x}_{i(m)}^{T}\right)\boldsymbol{v}_{k(m)}
\right|^{2}\non\\
&\leq&a_{n}^{2}K\sum_{i=1}^{n}\sum_{k=1}^{K}\mathrm{E}\left|\psi_{\tau_{k}}\left(\varepsilon_{i}+u_{ki(m)}\right)\left(1,\boldsymbol{x}_{i(m)}^{T}\right)\boldsymbol{v}_{k(m)}
\right|^{2}\non\\
&\leq& a_{n}^{2}K\sum_{i=1}^{n}\sum_{k=1}^{K}\boldsymbol{v}_{k(m)}^{T}
\mathrm{E}\left[\left\{\psi_{\tau_{k}}\left(\varepsilon_{i}+u_{ki(m)}\right)\right\}^{2}\left(1,\boldsymbol{x}_{i(m)}^{T}\right)^{T}\left(1,\boldsymbol{x}_{i(m)}^{T}\right)\right]\boldsymbol{v}_{k(m)}\non\\
&\leq& \frac{\bar{c}_{A(m)}}{c_{f}}nK a_{n}^{2}\sum_{k=1}^{K}\|\boldsymbol{v}_{k(m)}\|^{2}
\ee
and therefore by Chebyshev's inequality
\be \label{Z:1}
Z_{n(m),1}(\boldsymbol{v}_{(m)})
&=&O_{p}\left(a_{n}(nK)^{1/2}\bar{c}_{A(m)}^{1/2}/c_{f}^{1/2}\right)\sqrt{\sum_{k=1}^{K}\|\boldsymbol{v}_{k(m)}\|^{2}}\non\\
&=&O_{p}\left(\left(\bar{c}_{A(m)}K\right)^{1/2}\Big/\left(c_{f}^{1/2}\sqrt{K+\tilde{k}_{m}}\right)n a_{n}^{2}\right)\sqrt{\sum_{k=1}^{K}\|\boldsymbol{v}_{k(m)}\|^{2}}.
\ee

For $Z_{n(m),2}(\boldsymbol{v}_{(m)})$, by Law of iterated expectations and Taylor expansion, we have
\be \label{Z:21}
Z_{n(m),2}(\boldsymbol{v}_{(m)})
&=&\sum_{i=1}^{n}\sum_{k=1}^{K}\left\{\int_{0}^{a_{n}\left(1,\boldsymbol{x}_{i(m)}^{T}\right)\boldsymbol{v}_{k(m)}}
F(-u_{ki(m)}+s|\boldsymbol{x}_{i})-F(-u_{ki(m)}|\boldsymbol{x}_{i})\mathrm{d}s\right\}\non\\
&=&a_{n}\sum_{i=1}^{n}\sum_{k=1}^{K}\left\{\int_{0}^{\left(1,\boldsymbol{x}_{i(m)}^{T}\right)\boldsymbol{v}_{k(m)}}
F(-u_{ki(m)}+a_{n}t|\boldsymbol{x}_{i})-F(-u_{ki(m)}|\boldsymbol{x}_{i})\mathrm{d}t\right\}\non\\
&=&a_{n}\sum_{i=1}^{n}\sum_{k=1}^{K}\left[\int_{0}^{\left(1,\boldsymbol{x}_{i(m)}^{T}\right)\boldsymbol{v}_{k(m)}}
\left\{f(-u_{ki(m)}|\boldsymbol{x}_{i})a_{n}t\right.\right.\non\\
&&\left.\left.+f'(-u_{ki(m)}+\iota_{1} a_{n}t|\boldsymbol{x}_{i})a_{n}^{2}t^{2}\right\}\mathrm{d}t\right],
\ee
where $\iota_{1}\in (0,1)$.
By Assumption \ref{a2}(ii), we have
\be
&&\mathrm{E}\left\{\sum_{k=1}^{K}\boldsymbol{v}_{k(m)}^{T}f(-u_{ki(m)}|\boldsymbol{x}_{i})\left(1,\boldsymbol{x}_{i(m)}^{T}\right)^{T}\left(1,\boldsymbol{x}_{i(m)}^{T}\right)\boldsymbol{v}_{k(m)}\right\}\non\\
&=&\sum_{k=1}^{K}\boldsymbol{v}_{k(m)}^{T}\mathrm{E}\left\{f(-u_{ki(m)}|\boldsymbol{x}_{i})\left(1,\boldsymbol{x}_{i(m)}^{T}\right)^{T}\left(1,\boldsymbol{x}_{i(m)}^{T}\right)\right\}\boldsymbol{v}_{k(m)}\non\\
&=&\sum_{k=1}^{K}\boldsymbol{v}_{k(m)}^{T}A_{k(m)}\boldsymbol{v}_{k(m)}\non\\
&\leq&\bar{c}_{A_{(m)}}\sum_{k=1}^{K}\left\|\boldsymbol{v}_{k(m)}\right\|^{2}.
\ee
Then by Law of large numbers, we obtain
\be
&&n^{-1}\sum_{i=1}^{n}\left\{\sum_{k=1}^{K}\boldsymbol{v}_{k(m)}^{T}f(-u_{ki(m)}|\boldsymbol{x}_{i})\left(1,\boldsymbol{x}_{i(m)}^{T}\right)^{T}\left(1,\boldsymbol{x}_{i(m)}^{T}\right)\boldsymbol{v}_{k(m)}\right\}\non\\
&=&\mathrm{E}\left\{\sum_{k=1}^{K}\boldsymbol{v}_{k(m)}^{T}f(-u_{ki(m)}|\boldsymbol{x}_{i})\left(1,\boldsymbol{x}_{i(m)}^{T}\right)^{T}\left(1,\boldsymbol{x}_{i(m)}^{T}\right)\boldsymbol{v}_{k(m)}\right\}+o_{p}(1)\non\\
&=&\sum_{k=1}^{K}\boldsymbol{v}_{k(m)}^{T}\mathrm{E}\left\{f(-u_{ki(m)}|\boldsymbol{x}_{i})\left(1,\boldsymbol{x}_{i(m)}^{T}\right)^{T}\left(1,\boldsymbol{x}_{i(m)}^{T}\right)\right\}\boldsymbol{v}_{k(m)}+o_{p}(1)\non\\
&=&\sum_{k=1}^{K}\boldsymbol{v}_{k(m)}^{T}A_{k(m)}\boldsymbol{v}_{k(m)}+o_{p}(1)\non\\
&\geq&\underline{c}_{A_{(m)}}\sum_{k=1}^{K}\left\|\boldsymbol{v}_{k(m)}\right\|^{2}+o_{p}(1).
\ee
Further, we derive that
\be
&&\frac{a_{n}^{3}\overline{c}_{f}
\sum_{k=1}^{K}\sum_{i=1}^{n}\frac{1}{3}\left\{\left(1,\boldsymbol{x}_{i(m)}^{T}\right)\boldsymbol{v}_{k(m)}\right\}^{3}}
{a_{n}^{2}
\sum_{k=1}^{K}\sum_{i=1}^{n}\frac{1}{2}f(-u_{ki(m)}|\boldsymbol{x}_{i})\left\{\left(1,\boldsymbol{x}_{i(m)}^{T}\right)\boldsymbol{v}_{k(m)}\right\}^{2}}\non\\
&=&\frac{\overline{c}_{f}a_{n}
\sum_{k=1}^{K}\sum_{i=1}^{n}\left\{\left(1,\boldsymbol{x}_{i(m)}^{T}\right)\boldsymbol{v}_{k(m)}\right\}^{3}}
{6\sum_{k=1}^{K}\sum_{i=1}^{n}f(-u_{ki(m)}|\boldsymbol{x}_{i})\left\{\left(1,\boldsymbol{x}_{i(m)}^{T}\right)\boldsymbol{v}_{k(m)}\right\}^{2}}\non\\
&\leq&\frac{\overline{c}_{f}a_{n}
\sum_{k=1}^{K}\sum_{i=1}^{n}\left\{\left\|\left(1,\boldsymbol{x}_{i(m)}^{T}\right)^{T}\right\|^{3}\left\|\boldsymbol{v}_{k(m)}\right\|^{3}\right\}}
{6\sum_{k=1}^{K}\sum_{i=1}^{n}\boldsymbol{v}_{k(m)}^{T}f(-u_{ki(m)}|\boldsymbol{x}_{i})\left(1,\boldsymbol{x}_{i(m)}^{T}\right)^{T}\left(1,\boldsymbol{x}_{i(m)}^{T}\right)\boldsymbol{v}_{k(m)}}\non\\
&=&\frac{\overline{c}_{f}a_{n}
n^{-1}\sum_{i=1}^{n}\left\|\left(1,\boldsymbol{x}_{i(m)}^{T}\right)^{T}\right\|^{3}\sum_{k=1}^{K}\left\|\boldsymbol{v}_{k(m)}\right\|^{3}}
{6n^{-1}\sum_{i=1}^{n}\left\{\sum_{k=1}^{K}\boldsymbol{v}_{k(m)}^{T}f(-u_{ki(m)}|\boldsymbol{x}_{i})\left(1,\boldsymbol{x}_{i(m)}^{T}\right)^{T}\left(1,\boldsymbol{x}_{i(m)}^{T}\right)\boldsymbol{v}_{k(m)}\right\}}\non\\
&\leq&\frac{\overline{c}_{f}a_{n}
n^{-1}\sum_{i=1}^{n}\left\|\left(1,\boldsymbol{x}_{i(m)}^{T}\right)^{T}\right\|^{3}\sum_{k=1}^{K}\left\|\boldsymbol{v}_{k(m)}\right\|^{3}}
{6\underline{c}_{A_{(m)}}\sum_{k=1}^{K}\left\|\boldsymbol{v}_{k(m)}\right\|^{2}+o_{p}(1)}\non\\
&\leq&\frac{\overline{c}_{f}a_{n}
n^{-1}\sum_{i=1}^{n}\left\|\left(1,\boldsymbol{x}_{i(m)}^{T}\right)^{T}\right\|^{3}\max_{1\leq k\leq K}\left\|\boldsymbol{v}_{k(m)}\right\|\cdot\sum_{k=1}^{K}\left\|\boldsymbol{v}_{k(m)}\right\|^{2}}
{6\underline{c}_{A_{(m)}}\sum_{k=1}^{K}\left\|\boldsymbol{v}_{k(m)}\right\|^{2}+o_{p}(1)}.
\ee
By the part (iii) of Assumption \ref{a1} and H\"{o}lder's inequalities,
\be
&&\mathrm{E}\left\|\left(1,\boldsymbol{x}_{i(m)}^{T}\right)^{T}\right\|^{3}
=\mathrm{E}\left(1+\sum_{j=1}^{\tilde{k}_{m}}x_{ij(m)}^{2}\right)^{3/2}\non\\
&\leq&\left(1+\tilde{k}_{m}\right)^{1/2}
\mathrm{E}\left(1+\sum_{j=1}^{\tilde{k}_{m}}x_{ij(m)}^{3}\right)\non\\
&\leq&\left(1+\tilde{k}_{m}\right)^{3/2}c_{\boldsymbol{x}}.\non
\ee
By $\|\boldsymbol{v}_{(m)}\|= C>0$, we have $\max_{1\leq k\leq K}\left\|\boldsymbol{v}_{k(m)}\right\|\geq C$
and $\sum_{k=1}^{K}\left\|\boldsymbol{v}_{k(m)}\right\|^{2}\geq C^{2}$.
Therefore,
\be
&&\frac{a_{n}^{3}\overline{c}_{f}
\sum_{k=1}^{K}\sum_{i=1}^{n}\frac{1}{3}\left\{\left(1,\boldsymbol{x}_{i(m)}^{T}\right)\boldsymbol{v}_{k(m)}\right\}^{3}}
{a_{n}^{2}
\sum_{k=1}^{K}\sum_{i=1}^{n}\frac{1}{2}f(-u_{ki(m)}|\boldsymbol{x}_{i})\left\{\left(1,\boldsymbol{x}_{i(m)}^{T}\right)\boldsymbol{v}_{k(m)}\right\}^{2}}\non\\
&\leq&\frac{\overline{c}_{f}a_{n}
n^{-1}\sum_{i=1}^{n}\left\|\left(1,\boldsymbol{x}_{i(m)}^{T}\right)^{T}\right\|^{3}\cdot C}
{6\underline{c}_{A(m)}+o_{p}(1)}\non\\
&=&O_{p}\left(\frac{a_{n}\tilde{k}_{m}^{3/2}}{\underline{c}_{A(m)}}\right)
=O_{p}\left(\frac{\left(K+\tilde{k}_{m}\right)^{1/2}\tilde{k}_{m}^{3/2}}{n^{1/2}\underline{c}_{A(m)}}\right)
=o_{p}(1)
\ee
Then because the first order derivative of $f_{y\mid \boldsymbol{x}}(\cdot\mid \boldsymbol{x}_{i})$ is bounded above by a finite constant $\overline{c}_{f}$, we have
\be\label{Z:22}
&&a_{n}\sum_{i=1}^{n}\sum_{k=1}^{K}\left[\int_{0}^{\left(1,\boldsymbol{x}_{i(m)}^{T}\right)\boldsymbol{v}_{k(m)}}
\left\{f'(-u_{ki(m)}+\iota_{1} a_{n}t|\boldsymbol{x}_{i})a_{n}^{2}t^{2}\right\}\mathrm{d}t\right]\non\\
&\leq&a_{n}^{3}\overline{c}_{f}
\sum_{k=1}^{K}\sum_{i=1}^{n}\frac{1}{3}\left\{\left(1,\boldsymbol{x}_{i(m)}^{T}\right)\boldsymbol{v}_{k(m)}\right\}^{3}\non\\ &=&o_{p}\left[a_{n}^{2}
\sum_{k=1}^{K}\sum_{i=1}^{n}\frac{1}{2}f(-u_{ki(m)}|\boldsymbol{x}_{i})\left\{\left(1,\boldsymbol{x}_{i(m)}^{T}\right)\boldsymbol{v}_{k(m)}\right\}^{2}\right].
\ee
Therefore, by Law of large numbers and Assumption \ref{a2}(ii), we have
\be\label{Z:23}
Z_{n(m),2}(\boldsymbol{v}_{(m)})
&=&a_{n}^{2}
\sum_{k=1}^{K}\sum_{i=1}^{n}\frac{1}{2}f(-u_{ki(m)}|\boldsymbol{x}_{i})\left\{\left(1,\boldsymbol{x}_{i(m)}^{T}\right)\boldsymbol{v}_{k(m)}\right\}^{2}\{1+o_{p}(1)\}\non\\
&=&a_{n}^{2}
\left[\sum_{k=1}^{K}\frac{1}{2}\boldsymbol{v}_{k(m)}^{T}\left\{\sum_{i=1}^{n}f(-u_{ki(m)}\mid \boldsymbol{x}_{i})\left(1,\boldsymbol{x}_{i(m)}^{T}\right)^{T}\left(1,\boldsymbol{x}_{i(m)}^{T}\right)\right\}\boldsymbol{v}_{k(m)}\right]\{1+o_{p}(1)\}\non\\
&=&na_{n}^{2}
\left(\sum_{k=1}^{K}\frac{1}{2}\boldsymbol{v}_{k(m)}^{T}A_{k(m)}\boldsymbol{v}_{k(m)}\right)\{1+o_{p}(1)\}\non\\
&=& O_{p}\left(\underline{c}_{A_{m}}n a_{n}^{2}\sum_{k=1}^{K}\|\boldsymbol{v}_{k(m)}\|^{2}\right).
\ee

Noting that $\mathrm{E}(Z_{n(m),3}(\boldsymbol{v}_{(m)})|\boldsymbol{X})=0$ and by Assumption \ref{a2}(ii),
\be
&&\mathrm{Var}(Z_{n(m),3}(\boldsymbol{v}_{(m)})|\boldsymbol{X})\non\\
&=&\mathrm{E}\left(\sum_{i=1}^{n}\sum_{k=1}^{K}\left[\int_{0}^{a_{n}\left(1,\boldsymbol{x}_{i(m)}^{T}\right)\boldsymbol{v}_{k(m)}}\alpha_{ki(m)}(s)\mathrm{d}s
-\mathrm{E}\left\{\int_{0}^{a_{n}\left(1,\boldsymbol{x}_{i(m)}^{T}\right)\boldsymbol{v}_{k(m)}}\alpha_{ki(m)}(s)\mathrm{d}s|\boldsymbol{x}_{i}\right\}\right]
\right)^{2}\non\\
&=&\sum_{i=1}^{n}\mathrm{E}\left(\sum_{k=1}^{K}\left[\int_{0}^{a_{n}\left(1,\boldsymbol{x}_{i(m)}^{T}\right)\boldsymbol{v}_{k(m)}}\alpha_{ki(m)}(s)\mathrm{d}s
-\mathrm{E}\left\{\int_{0}^{a_{n}\left(1,\boldsymbol{x}_{i(m)}^{T}\right)\boldsymbol{v}_{k(m)}}\alpha_{ki(m)}(s)\mathrm{d}s|\boldsymbol{x}_{i}\right\}\right]
\right)^{2}\non\\
&\leq&\sum_{i=1}^{n}\mathrm{E}\left[\sum_{k=1}^{K}\left\{\int_{0}^{a_{n}\left(1,\boldsymbol{x}_{i(m)}^{T}\right)\boldsymbol{v}_{k(m)}}\alpha_{ki(m)}(s)\mathrm{d}s|\boldsymbol{x}_{i}\right\}\right]^{2}\non\\
&\leq&K a_{n}^{2}\sum_{k=1}^{K}\sum_{i=1}^{n}\left\{\left(1,\boldsymbol{x}_{i(m)}^{T}\right)\boldsymbol{v}_{k(m)}\right\}^{2}\non\\
&=&n K a_{n}^{2}\sum_{k=1}^{K}\boldsymbol{v}_{k(m)}^{T}\left[\mathrm{E}\left\{\left(1,\boldsymbol{x}_{i(m)}^{T}\right)^{T}\left(1,\boldsymbol{x}_{i(m)}^{T}\right)\right\}+o_{p}(1)\right]\boldsymbol{v}_{k(m)}\non\\
&=&O_{p}\left(\frac{\bar{c}_{A(m)}}{c_{f}}n K a_{n}^{2}\sum_{k=1}^{K}\|\boldsymbol{v}_{k(m)}\|^{2}\right).
\ee
Then we have
\be \label{Z:3}
Z_{n(m),3}(\boldsymbol{v}_{(m)})
&=&O_{p}( a_{n}n^{1/2} K^{1/2}\bar{c}_{A(m)}^{1/2}/c_{f}^{1/2})\sqrt{\sum_{k=1}^{K}\|\boldsymbol{v}_{k(m)}\|^{2}}\non\\
&=&O_{p}\left(\left(\bar{c}_{A(m)}K\right)^{1/2}\Big/\left(c_{f}^{1/2}\sqrt{K+\tilde{k}_{m}}\right)n a_{n}^{2}\right)\sqrt{\sum_{k=1}^{K}\|\boldsymbol{v}_{k(m)}\|^{2}}.
\ee
Observe that $\bar{c}_{A(m)}K/(K+\tilde{k}_{m})=O(\underline{c}_{A(m)}^{2})$ under Assumption \ref{a2}(iv), By (\ref{Z:1})-(\ref{Z:3}) and allowing $\sum_{k=1}^{K}\|\boldsymbol{v}_{k(m)}\|^{2}$ to be sufficiently
large, both $Z_{n(m),1}(\boldsymbol{v}_{(m)})$ and $Z_{n(m),3}(\boldsymbol{v}_{(m)})$ are dominated by $Z_{n(m),2}(\boldsymbol{v}_{(m)})$ which
is positive with probability approaching $1$ . This implies that
$Z_{n(m)}(\boldsymbol{v}_{(m)})>0$ with probability approaching $1$. This proves (i).

(ii) Let
\be
\hat{\Delta}_{k(m)}&=&\sqrt{n}\left\{\left(\hat{b}_{k(m)},\hat{\bbeta}_{(m)}^{T}\right)^{T}-\left(b_{k(m)}^{*},\bbeta_{(m)}^{*T}\right)^{T}\right\},\non\\
\hat{\Delta}_{(m)}&=&\sqrt{n}\left\{\left(\hat{b}_{1(m)},\cdots,\hat{b}_{K(m)},\hat{\bbeta}_{(m)}^{T}\right)^{T}-\left(b_{1(m)}^{*},\cdots,b_{K(m)}^{*},\bbeta_{(m)}^{*T}\right)^{T}\right\},\non\\
\Delta_{k(m)}&=&\sqrt{n}\left\{\left(b_{k(m)},\bbeta_{(m)}^{T}\right)^{T}-\left(b_{k(m)}^{*},\bbeta_{(m)}^{*T}\right)^{T}\right\},\non\\
\Delta_{(m)}&=&\sqrt{n}\left\{\left(b_{1(m)},\cdots,b_{K(m)},\bbeta_{(m)}^{T}\right)^{T}-\left(b_{1(m)}^{*},\cdots,b_{K(m)}^{*},\bbeta_{(m)}^{*T}\right)^{T}\right\},\non\\
\Delta&=&\sqrt{n}\left\{\left(b_{1},\cdots,b_{K},\bbeta^{T}\right)^{T}-\left(b_{1}^{*},\cdots,b_{K}^{*},\bbeta^{*T}\right)^{T}\right\}\non\\ \text{and}\non\\
\Delta_{k}&=&\sqrt{n}\left\{\left(b_{k},\bbeta^{T}\right)^{T}-\left(b_{k}^{*},\bbeta^{*T}\right)^{T}\right\}.\non
\ee
It follows that $\hat{\Delta}_{(m)}=\mathop{\arg\min}_{\Delta_{(m)}}
\sum_{i=1}^{n}\sum_{k=1}^{K}\rho_{\tau_{k}}\left(y_{i}-\left(1,\boldsymbol{x}_{i(m)}^{T}\right)\left\{\left(b_{k(m)}^{*},\bbeta_{(m)}^{*T}\right)^{T}+n^{-\frac{1}{2}}\Delta_{k(m)}\right\}\right)$.
Let $V_{(m)}(\Delta)=n^{-\frac{1}{2}}\sum_{i=1}^{n}
\left(\psi_{\tau_{1}}\left(\varepsilon_{1i(m)}(\Delta_{1})\right),\cdots,
\psi_{\tau_{K}}\left(\varepsilon_{Ki(m)}(\Delta_{K})\right),
\left\{\sum_{k=1}^{K}\psi_{\tau_{k}}\left(\varepsilon_{ki(m)}(\Delta_{k})\right)\right\}\boldsymbol{x}_{i(m)}^{T}\right)^{T}$
and $\bar{V}_{(m)}(\Delta)=\mathrm{E}\left[V_{(m)}(\Delta)\right]$, where $\varepsilon_{ki(m)}(\Delta_{k})=y_{i}-\left(1,\boldsymbol{x}_{i(m)}^{T}\right)\left\{\left(b_{k(m)}^{*},\bbeta_{(m)}^{*T}\right)^{T}+n^{-\frac{1}{2}}\Delta_{k}\right\}$. Define the weight norm $\|\cdot\|_{c_{(m)}}$ by $\|A\|_{c_{(m)}}=\|c_{(m)}A\|$, where $c_{(m)}$ is an an arbitrary $l_{m}\times(\tilde{k}_{m}+K)$ matrix with $\|c_{(m)}\|\leq \underline{c}_{\bar{B}_{(m)}}^{-\frac{1}{2}}L_{c}$
for a large constant $L_{c}<\infty$. We want to show that for any large
constant $L<\infty$,
\be\label{V}
\sup_{\left\|\Delta\right\|\leq\sqrt{\tilde{k}_{m}+K}L}
\left\|V_{(m)}(\Delta)-V_{(m)}(0)-\bar{V}_{(m)}(\Delta)+\bar{V}_{(m)}(0)\right\|_{\boldsymbol{c}_{(m)}}=o_{p}(1),
\ee
\be\label{V:bar}
\sup_{\left\|\Delta\right\|\leq\sqrt{\tilde{k}_{m}+K}L}
\left\|\bar{V}_{(m)}(\Delta)-\bar{V}_{(m)}(0)+A_{(m)}\Delta\right\|_{\boldsymbol{c}_{(m)}}=o_{p}(1),
\ee
\be\label{V:hat}
\left\|V_{(m)}(\hat{\Delta}_{(m)})\right\|_{\boldsymbol{c}_{(m)}}=o_{p}(1).
\ee
(\ref{V})-(\ref{V:bar}) and the result in part (i) imply that
$\left\|V_{(m)}(\hat{\Delta}_{(m)})-V_{(m)}(0)+A_{(m)}\hat{\Delta}_{(m)}\right\|_{c_{(m)}}=o_{p}(1)$
and consequently we have by Assumption \ref{a2}(ii)-(iii),
\be
\hat{\Delta}_{(m)}&=&\sqrt{n}\left\{\left(\hat{b}_{1(m)},\cdots,\hat{b}_{K(m)},\hat{\bbeta}_{(m)}^{T}\right)^{T}-\left(b_{1(m)}^{*},\cdots,b_{K(m)}^{*},\bbeta_{(m)}^{*T}\right)^{T}\right\}\non\\
&=&A_{(m)}^{-1}V_{(m)}(0)+A_{(m)}^{-1}V_{(m)}(\hat{\Delta}_{(m)})+A_{(m)}^{-1}R_{(m)},
\ee and
\be
C_{(m)}V_{(m)}^{-\frac{1}{2}}\hat{\Delta}_{(m)}
&=&\sqrt{n}C_{(m)}V_{(m)}^{-\frac{1}{2}}\left\{\left(\hat{b}_{1(m)},\cdots,\hat{b}_{K(m)},\hat{\bbeta}_{(m)}^{T}\right)^{T}-\left(b_{1(m)}^{*},\cdots,b_{K(m)}^{*},\bbeta_{(m)}^{*T}\right)^{T}\right\}\non\\
&=&C_{(m)}V_{(m)}^{-\frac{1}{2}}A_{(m)}^{-1}V_{(m)}(0)+C_{(m)}V_{(m)}^{-\frac{1}{2}}A_{(m)}^{-1}V_{(m)}(\hat{\Delta}_{(m)})+C_{(m)}V_{(m)}^{-\frac{1}{2}}A_{(m)}^{-1}R_{(m)}\non\\
&=&T_{1(m)}+T_{2(m)}+T_{3(m)},\text{say},
\ee
where $\|R_{(m)}\|_{c_{(m)}}=o_{p}(1)$ for any $c_{(m)}$ with $\|c_{(m)}\|\leq \underline{c}_{\bar{B}_{(m)}}^{-\frac{1}{2}}L_{c}$, $V_{(m)}^{\frac{1}{2}}$ denotes the symmetric square root of $V_{(m)}$ and $V_{(m)}^{-\frac{1}{2}}$ the inverse of
$V_{(m)}^{\frac{1}{2}}$.

Let
\be
\boldsymbol\gamma_{ni}&=&n^{-\frac{1}{2}}C_{(m)}V_{(m)}^{-\frac{1}{2}}A_{(m)}^{-1}\tilde{\boldsymbol\gamma}_{ni},\ee
where $\tilde{\boldsymbol\gamma}_{ni}=
\left(\psi_{\tau_{1}}\left(\varepsilon_{i}+u_{1i(m)}\right),\cdots,\psi_{\tau_{K}}\left(\varepsilon_{i}+u_{Ki(m)}\right),
\sum_{k=1}^{K}\psi_{\tau_{k}}\left(\varepsilon_{i}+u_{ki(m)}\right)
\boldsymbol{x}_{i(m)}^{T}\right)^{T}$,
then $T_{1(m)}=\sum_{i=1}^{n}\boldsymbol\gamma_{ni}$. Noting that $\mathrm{E}(\boldsymbol\gamma_{ni})=0$ by (\ref{d:0}), we have
\be
\mathrm{Var}(T_{1(m)})&=&\sum_{i=1}^{n}\mathrm{Var}(\boldsymbol\gamma_{ni})\non\\
&=&n^{-1}\sum_{i=1}^{n}C_{(m)}V_{(m)}^{-\frac{1}{2}}A_{(m)}^{-1}
\mathrm{E}\left(\tilde{\boldsymbol\gamma}_{ni}\tilde{\boldsymbol\gamma}_{ni}^{T}\right)A_{(m)}^{-1}V_{(m)}^{-\frac{1}{2}}C_{(m)}^{T}\non\\
&=&C_{(m)}V_{(m)}^{-\frac{1}{2}}A_{(m)}^{-1}B_{(m)}A_{(m)}^{-1}V_{(m)}^{-\frac{1}{2}}C_{(m)}^{T}=C_{(m)}C_{(m)}^{T}.
\ee
By the fact that $\mathrm{tr}(AB)\leq\lambda_{\max}(A)\mathrm{tr}(B)$ for symmetric matrix $A$ and positive semi-definite matrix $B$ (e.g., Bernstein (2005, Proposition 8.4.13)), and
By the part (iii) of Assumption \ref{a1} and H\"{o}lder's inequalities,
$\mathrm{E}\left\|\boldsymbol{x}_{i(m)}\right\|^{4}
=\mathrm{E}\left(\sum_{j=1}^{\tilde{k}_{m}}x_{ij(m)}^{2}\right)^{2}
\leq\tilde{k}_{m}
\mathrm{E}\left(\sum_{j=1}^{\tilde{k}_{m}}x_{ij(m)}^{4}\right)
\leq\tilde{k}_{m}^{2}c_{\boldsymbol{x}}$, we have
\be \label{eta4}
&&\mathrm{E}\left\|\left(C_{(m)}C_{(m)}^{T}\right)^{-\frac{1}{2}}\boldsymbol\gamma_{ni}\right\|^{4}\non\\
&=&n^{-2}\mathrm{E}\left(\left[\mathrm{tr}\left\{\left(C_{(m)}C_{(m)}^{T}\right)^{-\frac{1}{2}}C_{(m)}V_{(m)}^{-\frac{1}{2}}A_{(m)}^{-1}
\tilde{\boldsymbol\gamma}_{ni}\tilde{\boldsymbol\gamma}_{ni}^{T}A_{(m)}^{-1}V_{(m)}^{-\frac{1}{2}}C_{(m)}^{T}\left(C_{(m)}C_{(m)}^{T}\right)^{-\frac{1}{2}}\right\}
\right]^{2}\right)\non\\
&=&n^{-2}\mathrm{E}\left(\left[\mathrm{tr}\left\{
\tilde{\boldsymbol\gamma}_{ni}\tilde{\boldsymbol\gamma}_{ni}^{T}A_{(m)}^{-1}V_{(m)}^{-\frac{1}{2}}C_{(m)}^{T}\left(C_{(m)}C_{(m)}^{T}\right)^{-1}C_{(m)}V_{(m)}^{-\frac{1}{2}}A_{(m)}^{-1}
\right\}\right]^{2}\right)\non\\
&\leq&n^{-2}\mathrm{E}\left(\left\{\mathrm{tr}\left(
\tilde{\boldsymbol\gamma}_{ni}\tilde{\boldsymbol\gamma}_{ni}^{T}\right)\right\}^{2}\left\{\lambda_{\max}\left(A_{(m)}^{-1}V_{(m)}^{-\frac{1}{2}}C_{(m)}^{T}\left(C_{(m)}C_{(m)}^{T}\right)^{-1}
C_{(m)}V_{(m)}^{-\frac{1}{2}}A_{(m)}^{-1}\right)
\right\}^{2}\right)\non\\
&\leq&
n^{-2}\mathrm{E}\left\{K+K^{2}\mathrm{tr}(\boldsymbol{x}_{i(m)}\boldsymbol{x}_{i(m)}^{T})\right\}^{2}
\left\{\lambda_{\max}\left(\left(C_{(m)}C_{(m)}^{T}\right)^{-1}C_{(m)}C_{(m)}^{T}\right)\right\}^{2}\non\\
&&\left\{\lambda_{\max}\left(A_{(m)}^{-1}V_{(m)}^{-1}A_{(m)}^{-1}\right)\right\}^{2}\non\\
&=&n^{-2}\mathrm{E}\left\{K+K^{2}\mathrm{tr}(\boldsymbol{x}_{i(m)}\boldsymbol{x}_{i(m)}^{T})\right\}^{2}
\left\{\lambda_{\max}\left(A_{(m)}^{-1}V_{(m)}^{-1}A_{(m)}^{-1}\right)\right\}^{2}\non\\
&\leq&n^{-2}\left(2K^{2}+2K^{4}\mathrm{E}\left\|\boldsymbol{x}_{i(m)}\right\|^{4}\right)
\left\{\lambda_{\max}\left(B_{(m)}^{-1}\right)\right\}^{2}\non\\
&=&O(n^{-2}K^{4}\tilde{k}_{m}^{2}\underline{c}_{\bar{B}_{(m)}}^{-2}).
\ee
Then for any $\epsilon>0$,
\be
&&\sum_{i=1}^{n}\mathrm{E}\left[\left\|\left(C_{(m)}C_{(m)}^{T}\right)^{-\frac{1}{2}}\boldsymbol\gamma_{ni}\right\|^{2}
\boldsymbol{1}\left\{\left\|\left(C_{(m)}C_{(m)}^{T}\right)^{-\frac{1}{2}}\boldsymbol\gamma_{ni}\right\|\geq\epsilon\right\}\right]\non\\
&=&n\mathrm{E}\left[\left\|\left(C_{(m)}C_{(m)}^{T}\right)^{-\frac{1}{2}}\boldsymbol\gamma_{ni}\right\|^{2}
\boldsymbol{1}\left\{\left\|\left(C_{(m)}C_{(m)}^{T}\right)^{-\frac{1}{2}}\boldsymbol\gamma_{ni}\right\|\geq\epsilon\right\}\right]\non\\
&\leq& n\mathrm{E}\left\{\left\|\left(C_{(m)}C_{(m)}^{T}\right)^{-\frac{1}{2}}\boldsymbol\gamma_{ni}\right\|^{4}\right\}^{1/2}
\left\{P\left(\left\|\left(C_{(m)}C_{(m)}^{T}\right)^{-\frac{1}{2}}\boldsymbol\gamma_{ni}\right\|\geq\epsilon\right)\right\}^{1/2}\non\\
&\leq& n\epsilon^{-2}\mathrm{E}\left(\left\|\left(C_{(m)}C_{(m)}^{T}\right)^{-\frac{1}{2}}\boldsymbol\gamma_{ni}\right\|^{4}\right)
=O\left(n^{-1}K^{4}\tilde{k}_{m}^{2}\underline{c}_{\bar{B}_{(m)}}^{-2}\right)=o(1)
\ee
by Assumption \ref{a3}(i).
Thus $\boldsymbol\gamma_{ni}$ satisfies the conditions of the Lindeberg-Feller central limit theorem and we have
\be \label{T1}T_{1(m)}\overd N(0,C_{0}).\ee
For $T_{2(m)}$ and $T_{3(m)}$, we take $c_{(m)}=C_{(m)}V_{(m)}^{-\frac{1}{2}}A_{(m)}^{-1}$. By the fact that $\mathrm{tr}(AB)\leq\lambda_{\max}(A)\mathrm{tr}(B)$ for symmetric matrix $A$ and positive semi-definite matrix $B$ and that $\lambda_{\max}(A^{T}A)=\lambda_{\max}(A A^{T})$ for any matrix $A$, we have
\be
\left\|c_{(m)}\right\|
&=&\left\{\mathrm{tr}\left(V_{(m)}^{-1/2}A_{(m)}^{-1}A_{(m)}^{-1}V_{(m)}^{-1/2}C_{(m)}^{T}C_{(m)}\right)\right\}^{1/2}\non\\
&\leq&\left\|C_{(m)}\right\|\left\{\lambda_{\max}\left(V_{(m)}^{-1/2}A_{(m)}^{-1}A_{(m)}^{-1}V_{(m)}^{-1/2}\right)\right\}^{1/2}\non\\
&=&\left\|C_{(m)}\right\|\left\{\lambda_{\max}\left(A_{(m)}^{-1}V_{(m)}^{-1}A_{(m)}^{-1}\right)\right\}^{1/2}\non\\
&=&\left\|C_{(m)}\right\|\left\{\lambda_{\max}\left(B_{(m)}^{-1}\right)\right\}^{1/2}
\leq\left\|C_{(m)}\right\|\underline{c}_{\bar{B}_{(m)}}^{-1/2}\leq L_{c}\underline{c}_{\bar{B}_{(m)}}^{-1/2}\non
\ee
for sufficiently large $L_{c}$. Then by (\ref{V:hat})
\be\label{T2}
\left\|T_{2(m)}\right\|=\left\|V_{(m)}(\hat{\Delta}_{(m)})\right\|_{C_{(m)}V_{(m)}^{-\frac{1}{2}}A_{(m)}^{-1}}=o_{p}(1)
\ee and
\be\label{T3}
\left\|T_{3(m)}\right\|=\left\|R_{2(m)}\right\|_{C_{(m)}V_{(m)}^{-\frac{1}{2}}A_{(m)}^{-1}}=o_{p}(1)
\ee
Combining (\ref{T1})-(\ref{T3}) yields the claim in part (ii).

Below we demonstrate (\ref{V})¨C(\ref{V:hat}) hold under Assumptions \ref{a1}-\ref{a3}. Since $l_{m}$ is fixed, without loss of generality we assume that $l_{m}=1$. Denote $\boldsymbol{c}_{(m)}=(c_{1(m)},\cdots,c_{K(m)},\tilde{\boldsymbol{c}}_{(m)})$, where
 $\tilde{\boldsymbol{c}}_{(m)}$ is a $\tilde{k}_{m}\times 1$ vector. First, we show (\ref{V}). Write $a_{ki(m)}=\left(c_{k(m)},\boldsymbol{\tilde{c}}_{(m)}\right)\left(1,\boldsymbol{x}_{ki(m)}^{T}\right)^{T}=a_{ki(m)}^{+}-a_{ki(m)}^{-}$,
where $a_{ki(m)}^{+}=\max\{a_{ki(m)},0\}$ and $a_{ki(m)}^{-}=\max\{-a_{ki(m)},0\}$. Then by Minkowski's inequality we have
\be\label{V1}
&&\sup_{\left\|\Delta\right\|\leq\sqrt{\tilde{k}_{m}+K}L}
\left\|V_{(m)}(\Delta)-V_{(m)}(0)-\bar{V}_{(m)}(\Delta)+\bar{V}_{(m)}(0)\right\|_{c_{(m)}}\non\\
&\leq&\sup_{\left\|\Delta\right\|\leq\sqrt{\tilde{k}_{m}+K}L}
\left|\Psi_{(m)}(\Delta)-\Psi_{(m)}(0)-\bar{\Psi}_{(m)}(\Delta)+\bar{\Psi}_{(m)}(0)\right|\non\\
&&+\sup_{\left\|\Delta\right\|\leq\sqrt{\tilde{k}_{m}+K}L}
\left|V_{(m)}^{+}(\Delta)-V_{(m)}^{+}(0)-\bar{V}_{(m)}^{+}(\Delta)+\bar{V}_{(m)}^{+}(0)\right|\non\\
&&+\sup_{\left\|\Delta\right\|\leq\sqrt{\tilde{k}_{m}+K}L}
\left|V_{(m)}^{-}(\Delta)-V_{(m)}^{-}(0)-\bar{V}_{(m)}^{-}(\Delta)+\bar{V}_{(m)}^{-}(0)\right|,
\ee
where
$\Psi_{(m)}(\Delta)=n^{-\frac{1}{2}}\sum_{i=1}^{n}\sum_{k=1}^{K}c_{k}\psi_{\tau_{k}}\left(\varepsilon_{ki(m)}(\Delta_{k})\right)$, $\bar{\Psi}_{(m)}(\Delta)=\mathrm{E}\left[\Psi_{(m)}(\Delta)\right]$,
$V_{(m)}^{+}(\Delta)=n^{-\frac{1}{2}}\sum_{i=1}^{n}\sum_{k=1}^{K}\psi_{\tau_{k}}\left(\varepsilon_{ki(m)}(\Delta_{k})\right)a_{ki(m)}^{+}$,
$\bar{V}_{(m)}^{+}(\Delta)=\mathrm{E}\left\{V_{(m)}^{+}(\Delta)\right\}$, and
$V_{(m)}^{-}(\Delta)$ and $\bar{V}_{(m)}^{-}(\Delta)$ are analogously defined.
It suffices to show that each term on the right hand side of
(\ref{V1}) is $o_{p}(1)$. We only show the first and second terms are $o_{p}(1)$ as the third term can be treated analogously.

Let $\boldsymbol{D}=\left\{\Delta\in \mathbb{R}^{\tilde{k}_{m}+K}:\left\|\Delta\right\|\leq\sqrt{\tilde{k}_{m}+K}L\right\}$ for some $L<\infty$.
Let $\left|t\right|_{\infty}$ denote the maximum of the absolute values of the coordinates of $t$. By selecting $N_{1}=(2n^{2})^{\tilde{k}_{m}+K}$ grid points, $\tilde{\Delta}_{1},\cdots,\tilde{\Delta}_{N_{1}}$, we can cover
$\boldsymbol{D}$ by cubes $\boldsymbol{D}_{s}=\left\{\Delta\in \mathbb{R}^{\tilde{k}_{m}+K}:\left|\Delta-\tilde{\Delta}_{s}\right|_{\infty}\leq\delta_{n}\right\}$ , $s=1,\cdots,N_{1}$, with sides of length $\delta_{n}$
where $\delta_{n}=L\left(\tilde{k}_{m}+K\right)^{1/2}/n^{2}$.
Denote $$\tilde{\Delta}_{(s)}=\sqrt{n}\left\{\left(\tilde{b}_{1 (s)},\cdots,\tilde{b}_{K (s)},\tilde{\bbeta}^{T}\right)^{T}-\left(\tilde{b}_{1 (s)}^{*},\cdots,\tilde{b}_{K (s)}^{*},\tilde{\bbeta}^{*T}\right)^{T}\right\} $$
and
$$
\tilde{\Delta}_{k (s)}=\sqrt{n}\left\{\left(\tilde{b}_{k (s)},\tilde{\bbeta}^{T}\right)^{T}-\left(\tilde{b}_{k (s)}^{*},\tilde{\bbeta}^{*T}\right)^{T}\right\}.
$$
Let $$\varepsilon_{ki(m)}=y_{i}-\left(1,\boldsymbol{x}_{i(m)}^{T}\right)\left(b_{k(m)}^{*},\bbeta_{(m)}^{*T}\right)^{T}$$
and $$\psi_{ksi(m)}(\delta)=\psi_{\tau_{k}}\left(\varepsilon_{ki(m)}-n^{-\frac{1}{2}}\left(1,\boldsymbol{x}_{i(m)}^{T}\right)\tilde{\Delta}_{k(s)}+n^{-\frac{1}{2}}\delta\left\|\left(1,\boldsymbol{x}_{i(m)}^{T}\right)^{T}\right\|\right).$$
In view of the fact that $\psi_{\tau_{k}}(\cdot)$ is monotone and by Minkowski's inequality, we can readily show that
\be
&&\sup_{\left\|\Delta\right\|\leq\sqrt{\tilde{k}_{m}+K}L}
\left|V_{(m)}^{+}(\Delta)-V_{(m)}^{+}(0)-\bar{V}_{(m)}^{+}(\Delta)+\bar{V}_{(m)}^{+}(0)\right|\non\\
&\leq&\mathrm{max}_{1\leq s\leq N_{1}}
\left|V_{(m)}^{+}(\tilde{\Delta}_{s})-V_{(m)}^{+}(0)-\bar{V}_{(m)}^{+}(\tilde{\Delta}_{s})+\bar{V}_{(m)}^{+}(0)\right|\non\\
&&+\mathrm{max}_{1\leq s\leq N_{1}}
\left|n^{-\frac{1}{2}}\sum_{i=1}^{n}\mathrm{E}\left\{\sum_{k=1}^{K}\psi_{ksi(m)}(\delta_{n})a_{ki(m)}^{+}\right\}
-\mathrm{E}\left\{\sum_{k=1}^{K}\psi_{ksi(m)}(-\delta_{n})a_{ki(m)}^{+}\right\}\right|\non\\
&&+\mathrm{max}_{1\leq s\leq N_{1}}
\left|n^{-\frac{1}{2}}\sum_{i=1}^{n}\left(
\sum_{k=1}^{K}\left\{\psi_{ksi(m)}(\delta_{n})-\psi_{ksi(m)}(0)\right\}a_{ki(m)}^{+}\right.\right.\non\\
&&\left.\left.-\mathrm{E}\left[\sum_{k=1}^{K}\left\{\psi_{ksi(m)}(\delta_{n})-\psi_{ksi(m)}(0)\right\}a_{ki(m)}^{+}
\right]\right)\right|\non\\
&=&I_{1(m)}+I_{2(m)}+I_{3(m)},~\text{say},
\ee
For $I_{1(m)}$, we apply Lagrange Mean Value Theorem, Assumption \ref{a2}(i), and the fact that
$a_{i(m)}^{+}\leq\left|\left(c_{k(m)},\boldsymbol{\tilde{c}}_{(m)}\right)\left(1,\boldsymbol{x}_{i(m)}^{T}\right)^{T}\right|
\leq \left\|\left(c_{k(m)},\boldsymbol{\tilde{c}}_{(m)}\right)\right\|\left\|\left(1,\boldsymbol{x}_{i(m)}^{T}\right)\right\|$ to obtain
\be
I_{2(m)}&=&\mathrm{max}_{1\leq s\leq N_{1}}
\Bigg|
n^{1/2}\mathrm{E}\left[ \sum_{k=1}^{K}  F\left(-u_{ki(m)}+n^{-1/2}\left(1,\boldsymbol{x}_{i(m)}^{T}\right)\tilde{\Delta}_{k(s)}\right.\right.\non\\
&&\left.\left.+n^{-1/2}\delta_{n}\left\|\left(1,\boldsymbol{x}_{i(m)}^{T}\right)^{T}\right\|\Big|\boldsymbol{x}_{i(m)}\right)a_{ki(m)}^{+}     \right]-\mathrm{E}\left[ \sum_{k=1}^{K} F\left(-u_{ki(m)}\right.\right.\non\\
&&\left.\left.+n^{-1/2}\left(1,\boldsymbol{x}_{i(m)}^{T}\right)\tilde{\Delta}_{k(s)}-n^{-1/2}\delta_{n}\left\|\left(1,\boldsymbol{x}_{i(m)}^{T}\right)^{T}\right\|\Big|\boldsymbol{x}_{i(m)}\right)a_{ki(m)}^{+}     \right]\Bigg|\non\\
&=&2f\left(\xi\Big|\boldsymbol{x}_{i(m)}\right)\delta_{n}\mathrm{E}\left\{\sum_{k=1}^{K}\left\|\left(1,\boldsymbol{x}_{i(m)}^{T}\right)^{T}\right\|a_{ki(m)}^{+}\right\}\non\\
&\leq& 2c_{f}\delta_{n}\mathrm{E}\left\{\sum_{k=1}^{K}\left\|\left(1,\boldsymbol{x}_{i(m)}^{T}\right)^{T}\right\|a_{ki(m)}^{+}\right\}\non\\
&\leq & 2c_{f}\delta_{n}\sum_{k=1}^{K}\left\|\left(c_{k(m)},\boldsymbol{\tilde{c}}_{(m)}\right)\right\|\mathrm{E}\left\|\left(1,\boldsymbol{x}_{i(m)}^{T}\right)^{T}\right\|^{2}\non\\
&=&O\left(\delta_{n}\underline{c}_{B_{(m)}}^{-1/2}\tilde{k}_{m}\right)=O\left(\underline{c}_{B_{(m)}}^{-1/2}(\tilde{k}_{m}+K)^{1/2}\tilde{k}_{m}/n^{2}\right)=o(1),
\ee
where $\xi$ is a value between $-u_{ki(m)}+n^{-1/2}\left(1,\boldsymbol{x}_{i(m)}^{T}\right)\tilde{\Delta}_{k(s)}-n^{-1/2}\delta_{n}\left\|\left(1,\boldsymbol{x}_{i(m)}^{T}\right)^{T}\right\|$
and $-u_{ki(m)}+n^{-1/2}\left(1,\boldsymbol{x}_{i(m)}^{T}\right)\tilde{\Delta}_{k(s)}+n^{-1/2}\delta_{n}\left\|\left(1,\boldsymbol{x}_{i(m)}^{T}\right)^{T}\right\|$.

For $I_{1(m)}$, let $$\eta_{kis(m)}=n^{1/2}\left\{\eta_{kis(m),0}-\mathrm{E}\left(\eta_{kis(m),0}\right)\right\},$$
$$\eta_{kis(m),0}=\left[\psi_{\tau_{k}}\left(\varepsilon_{ki(m)}-n^{-1/2}\left(1,\boldsymbol{x}_{i(m)}^{T}\right)\tilde{\Delta}_{k(s)}\right)-\psi_{\tau_{k}}\left(\varepsilon_{ki(m)}\right)\right]a_{ki(m)}^{+}$$
and $e_{1n}=\left(nK\tilde{k}_{m}^{2}\underline{c}_{B_{(m)}}^{-2}\right)^{1/4}$. Note that
\be
&&V_{(m)}^{+}(\Delta_{s})-V_{(m)}^{+}(0)-\bar{V}_{(m)}^{+}(\Delta_{s})+\bar{V}_{(m)}^{+}(0)\non\\
&&=n^{-1}\sum_{i=1}^{n}\sum_{k=1}^{K}\eta_{kis(m)}\boldsymbol{1}\left\{a_{ki(m)}^{+}\leq e_{1n}/n^{1/2}\right\}
+n^{-1}\sum_{i=1}^{n}\sum_{k=1}^{K}\eta_{kis(m)}\boldsymbol{1}\left\{a_{ki(m)}^{+}> e_{1n}/n^{1/2}\right\}\non\\
&&=D_{1s}+D_{2s},~ \text{say}.
\ee
It suffices to prove $I_{1(m)}=o_{p}(1)$ by showing that
\be\label{D12}\mathrm{max}_{1\leq s\leq N_{1}}\|D_{ls}\|=o_{p}(1)
\ \text{for}\ l=1\ \text{and}\ 2.\ee
Note that \be
&&\mathrm{Var}\left[\sum_{k=1}^{K}\eta_{kis(m)}\boldsymbol{1}\left\{a_{ki(m)}^{+}\leq e_{1n}/n^{1/2}\right\}\right]\non\\
&\leq& nK\mathrm{E}\left\{\left|\psi_{\tau_{k}}\left(\varepsilon_{ki(m)}
-n^{-1/2}\left(1,\boldsymbol{x}_{i(m)}^{T}\right)\tilde{\Delta}_{k(s)}\right)
-\psi_{\tau_{k}}\left(\varepsilon_{ki(m)}\right)\right|\left(a_{i(m)}^{+}\right)^{2}\right\}\non\\
&\leq& C_{1}\underline{c}_{B_{(m)}}^{-1}n^{1/2}K\tilde{k}_{m}
\ee
for some $C_{1}<\infty$.
By Boole's and Bernstein's inequalities, we have
\be
&&\mathrm{P}\left(\mathrm{max}_{1\leq s\leq N_{1}}\|D_{1s}\|\geq\epsilon\right)\non\\
&\leq& N_{1}\mathrm{max}_{1\leq s\leq N_{1}}
\mathrm{P}\left(\left\|\frac{1}{n}\sum_{i=1}^{n}\sum_{k=1}^{K}\eta_{kis(m)}\boldsymbol{1}\left\{a_{ki(m)}^{+}\leq e_{1n}/n^{1/2}\right\}\right\|\geq\epsilon\right)\non\\
&\leq& 2N_{1}\exp\left\{-\frac{n\epsilon^{2}}{2C_{1}\underline{c}_{B_{(m)}}^{-1}
n^{1/2}K\tilde{k}_{m}+4\epsilon n^{1/2}e_{1n}/3}\right\}\non\\
&\leq& 2\exp\left\{6\tilde{k}_{m}\log n\right\}\times\exp\left\{-7\tilde{k}_{m}\log n\right\}\non\\
&=&\exp\left\{-\tilde{k}_{m}\log n\right\}=o(1).
\ee
because $n/(\underline{c}_{B_{(m)}}^{-1}n^{1/2}K\tilde{k}_{m})
=n^{1/2}\underline{c}_{B_{(m)}}/(K\tilde{k}_{m})\gg K\tilde{k}_{m}\log n$
and
$n/(n^{1/2}e_{1n})
=n^{1/4}K^{-1/4}\tilde{k}_{m}^{-1/2}\underline{c}_{B_{(m)}}^{1/2}\gg K\tilde{k}_{m}\log n$
by Assumption \ref{a3}(i).
Let $\bar{a}_{i(m)}=a_{i(m)}\tilde{k}_{m}^{-1/2}\underline{c}_{B_{(m)}}^{1/2}$.
Noting that $\mathrm{E}\left\{\left|\boldsymbol{c}_{(m)}\left(1,\boldsymbol{x}_{i(m)}^{T}\right)\right|^{8}\right\}
=O(\tilde{k}_{m}^{4}\underline{c}_{B_{(m)}}^{-4})$
by arguments as used to obtain (\ref{eta4}), $\mathrm{E}\left(\left|\bar{a}_{i(m)}\right|^{8}\right)
=O(1)$.
By Boole's and Markov's inequalities, Assumption \ref{a1}(iii), the fact that
$nK\tilde{k}_{m}^{2}\underline{c}_{B_{(m)}}^{-2}/e_{1n}^{4}=O(1)$ by construction, and Cauchy-Schwarz inequalities,
\be
&&\mathrm{P}\left(\mathrm{max}_{1\leq s\leq N_{1}}\|D_{2s}\|\geq\epsilon\right)\non\\
&\leq&\mathrm{P}\left(\mathrm{max}_{1\leq i\leq n}\mathrm{max}_{1\leq k\leq K}a_{ki(m)}^{+}> e_{1n}/n^{1/2}\right)\non\\
&\leq&nK\mathrm{P}\left(|\bar{a}_{i(m)}|> \tilde{k}_{m}^{-1/2}\underline{c}_{B_{(m)}}^{1/2}e_{1n}\right)\non\\
&\leq&\frac{nK\tilde{k}_{m}^{2}\underline{c}_{B_{(m)}}^{-2}}{e_{1n}^{4}}
\mathrm{E}\left[\left|\bar{a}_{i(m)}\right|^{4}\boldsymbol{1}\left\{|\bar{a}_{i(m)}|> \tilde{k}_{m}^{-1/2}\underline{c}_{B_{(m)}}^{1/2}e_{1n}\right\}\right]\non\\
&\leq&\frac{nK\tilde{k}_{m}^{2}\underline{c}_{B_{(m)}}^{-2}}{e_{1n}^{4}}
\mathrm{E}\left(\left|\bar{a}_{i(m)}\right|^{8}\right)^{1/2}
\mathrm{E}\left[\boldsymbol{1}\left\{|\bar{a}_{i(m)}|> \tilde{k}_{m}^{-1/2}\underline{c}_{B_{(m)}}^{1/2}e_{1n}\right\}\right]^{1/2}
=o(1).
\ee
Thus (\ref{D12}) follows and we have shown $I_{1(m)}=o_{p}(1)$. By the same
token, we can show that $I_{3(m)}=o_{p}(1)$. Consequently (\ref{V}) follows.

Next, we show (\ref{V:bar}). By Assumptions \ref{a1} and \ref{a2},
\be
&&\sup_{\left\|\Delta\right\|\leq\sqrt{\tilde{k}_{m}+K}L}
\left\|\bar{V}_{(m)}(\Delta)-\bar{V}_{(m)}(0)+A_{(m)}\Delta\right\|_{\boldsymbol{c}_{(m)}}\non\\
&=&\sup_{\left\|\Delta\right\|\leq\sqrt{\tilde{k}_{m}+K}L}
\left\|n^{-1/2}\sum_{i=1}^{n}
\mathrm{E}\left(
\left(F\left(-u_{1i(m)}+n^{-1/2}\left(1,\boldsymbol{x}_{i(m)}^{T}\right)\Delta_{1}|\boldsymbol{x}_{i(m)}\right)\right.\right.\right.\non\\
&&-F\left(-u_{1i(m)}|\boldsymbol{x}_{i(m)}\right),\cdots,
F\left(-u_{Ki(m)}+n^{-1/2}\left(1,\boldsymbol{x}_{i(m)}^{T}\right)\Delta_{K}|\boldsymbol{x}_{i(m)}\right)\non\\
&&-F\left(-u_{Ki(m)}|\boldsymbol{x}_{i(m)}\right),
\sum_{k=1}^{K}\left[F\left(-u_{ki(m)}+n^{-1/2}\left(1,\boldsymbol{x}_{i(m)}^{T}\right)\Delta_{k}|\boldsymbol{x}_{i(m)}\right)\right.\non\\
&&\left.\left.\left.\left.-F\left(-u_{ki(m)}|\boldsymbol{x}_{i(m)}\right)\right]\boldsymbol{x}_{i(m)}\right)^{T}
\right)
-A_{(m)}\Delta\right\|_{\boldsymbol{c}_{(m)}}\non\\
&=&\sup_{\left\|\Delta\right\|\leq\sqrt{\tilde{k}_{m}+K}L}
\Bigg\|n^{-1}\sum_{i=1}^{n}
\mathrm{E}\left(\left[
\int_{0}^{1}
f\left(-u_{1i(m)}+sn^{-1/2}\left(1,\boldsymbol{x}_{i(m)}^{T}\right)\Delta_{1}|\boldsymbol{x}_{i(m)}\right)\right.\right.
\mathrm{d}s\left(1,\boldsymbol{x}_{i(m)}^{T}\right)\Delta_{1},\non\\
&&\cdots,\int_{0}^{1}
f\left(-u_{Ki(m)}+sn^{-1/2}\left(1,\boldsymbol{x}_{i(m)}^{T}\right)\Delta_{K}|\boldsymbol{x}_{i(m)}\right)
\mathrm{d}s\left(1,\boldsymbol{x}_{i(m)}^{T}\right)\Delta_{K},\non\\
&&\left.\left.\left.\int_{0}^{1}\sum_{k=1}^{K}f\left(-u_{ki(m)}+sn^{-1/2}\left(1,\boldsymbol{x}_{i(m)}^{T}\right)\Delta_{k}|\boldsymbol{x}_{i(m)}\right)
\mathrm{d}s\left(1,\boldsymbol{x}_{i(m)}^{T}\right)\Delta_{k}\boldsymbol{x}_{i(m)}
\right]\right)-A_{(m)}\Delta\right\|_{\boldsymbol{c}_{(m)}}\non\\
&=&\sup_{\left\|\Delta\right\|\leq\sqrt{\tilde{k}_{m}+K}L}
\sum_{k=1}^{K}\left\|n^{-1}\sum_{i=1}^{n}
\mathrm{E}\left[\int_{0}^{1}
f\left(-u_{ki(m)}+sn^{-1/2}\left(1,\boldsymbol{x}_{i(m)}^{T}\right)\Delta_{k}|\boldsymbol{x}_{i}\right)\right.\right.\non\\
&&\left.\left.-f\left(-u_{ki(m)}\right)\mathrm{d}s\left(1,\boldsymbol{x}_{i(m)}^{T}\right)^{T}\left(1,\boldsymbol{x}_{i(m)}^{T}\right)\Delta_{k}\right]\right\|_{\left(c_{k(m)},\boldsymbol{\tilde{c}}_{(m)}\right)}\non\\
&\leq&\overline{c}_{f}\sup_{\left\|\Delta\right\|\leq\sqrt{\tilde{k}_{m}+K}L}n^{-3/2}
\sum_{i=1}^{n}\sum_{k=1}^{K}\mathrm{E}\left\|
\Delta_{k}^{T}\left(1,\boldsymbol{x}_{i(m)}^{T}\right)^{T}\left(1,\boldsymbol{x}_{i(m)}^{T}\right)^{T}\left(1,\boldsymbol{x}_{i(m)}^{T}\right)\Delta_{k}\right\|_{\left(c_{k(m)},\boldsymbol{\tilde{c}}_{(m)}\right)}\non\\
&\leq&\overline{c}_{f}n^{-1/2}KL_{c}\underline{c}_{B_{(m)}}^{-1/2}\sup_{\left\|\Delta\right\|\leq\sqrt{\tilde{k}_{m}+K}L}
\left\{\mathrm{E}\left\|\Delta_{k}^{T}\left(1,\boldsymbol{x}_{i(m)}^{T}\right)^{T}\right\|^{2}\right\}^{1/2}
\left\{\mathrm{E}\left\|\left(1,\boldsymbol{x}_{i(m)}^{T}\right)^{T}\left(1,\boldsymbol{x}_{i(m)}^{T}\right)\Delta_{k}\right\|^{2}\right\}^{1/2}\non\\
&=&n^{-1/2}K\underline{c}_{B_{(m)}}^{-1/2}O\left(\bar{c}_{A_{(m)}}^{1/2}\tilde{k}_{m}^{1/2}\right)O\left(\tilde{k}_{m}^{3/2}\right)\non\\
&=&O\left(\underline{c}_{B_{(m)}}^{-1/2}\bar{c}_{A_{(m)}}^{1/2}\tilde{k}_{m}^{2}K/n^{1/2}\right)=o(1),
\ee
where we use the fact that $\mathrm{E}\left\|\Delta_{k}^{T}\left(1,\boldsymbol{x}_{i(m)}^{T}\right)^{T}\right\|^{2}
\leq\lambda_{\mathrm{max}}\left(\mathrm{E}\left\{\left(1,\boldsymbol{x}_{i(m)}^{T}\right)^{T}\left(1,\boldsymbol{x}_{i(m)}^{T}\right)\right\}\right)
\|\Delta_{k}\|^{2}=O\left(\bar{c}_{A_{(m)}}\tilde{k}_{m}\right)$ by Assumption \ref{a2}(ii).

Now we show (\ref{V:hat}). By the proof of Lemma A2 in Ruppert and
Carroll (1980) and Assumptions \ref{a2}(i)-(ii),
\be
\left\|V_{(m)}(\hat{\Delta}_{(m)})\right\|_{\boldsymbol{c}_{(m)}}
&=&\left\|n^{-\frac{1}{2}}\sum_{i=1}^{n}
\left(\psi_{\tau_{1}}\left(\varepsilon_{1i(m)}(\hat{\Delta}_{1(m)})\right),\cdots,
\psi_{\tau_{K}}\left(\varepsilon_{Ki(m)}(\hat{\Delta}_{K(m)})\right),\right.\right.\non\\
&&\left.\left.\sum_{k=1}^{K}\psi_{\tau_{k}}\left(\varepsilon_{ki(m)}(\Delta_{k})\right)\boldsymbol{x}_{i(m)}^{T}\right)^{T}\right\|_{c_{(m)}}\non\\
&=&\sum_{k=1}^{K}\left\|n^{-\frac{1}{2}}\sum_{i=1}^{n}
\psi_{\tau_{k}}\left(\varepsilon_{ki(m)}(\hat{\Delta}_{k(m)})\right)\left(1,\boldsymbol{x}_{i(m)}^{T}\right)^{T}\right\|_{\left(c_{k(m)},\boldsymbol{\tilde{c}}_{(m)}\right)}\non\\
&\leq& n^{-\frac{1}{2}}\sum_{k=1}^{K}\sum_{i=1}^{n}\boldsymbol{1}\left\{y_{i}-\left(1,\boldsymbol{x}_{i(m)}^{T}\right)\left(\hat{b}_{k(m)},\hat{\bbeta}_{(m)}^{T}\right)^{T}
=0\right\}\non\\
&&\left|\left(c_{k(m)},\boldsymbol{\tilde{c}}_{(m)}\right)\left(1,\boldsymbol{x}_{i(m)}^{T}\right)^{T}\right|\non\\
&\leq& n^{-\frac{1}{2}}\tilde{k}_{m}K\mathrm{max}_{1\leq i\leq n}\mathrm{max}_{1\leq k\leq K}\left|\left(c_{k(m)},\boldsymbol{\tilde{c}}_{(m)}\right)\left(1,\boldsymbol{x}_{i(m)}^{T}\right)^{T}\right|\non\\
&=&o_{p}(1).
\ee
Because by Boole's and Markov's inequalities
\be
&&\mathrm{P}\left(\mathrm{max}_{1\leq i\leq n}\mathrm{max}_{1\leq k\leq K}\left|\left(c_{k(m)},\boldsymbol{\tilde{c}}_{(m)}\right)\left(1,\boldsymbol{x}_{i(m)}^{T}\right)^{T}\right|\geq n^{1/2}/(\tilde{k}_{m}K)\right)\non\\
&\leq& nK\mathrm{P}\left(\left|\left(c_{k(m)},\boldsymbol{\tilde{c}}_{(m)}\right)\left(1,\boldsymbol{x}_{i(m)}^{T}\right)^{T}\right|\geq n^{1/2}/(\tilde{k}_{m}K)\right)\non\\
&\leq& \tilde{k}_{m}^{8}K^{9}n^{-3}\mathrm{E}\left\{\left|\left(c_{k(m)},\boldsymbol{\tilde{c}}_{(m)}\right)\left(1,\boldsymbol{x}_{i(m)}^{T}\right)^{T}\right|^{8}\right\}\non\\
&=&\tilde{k}_{m}^{8}K^{9}n^{-3}O\left(\tilde{k}_{m}^{4}\underline{c}_{B_{(m)}}^{-4}\right)
=O\left(n^{-3}\tilde{k}_{m}^{12}K^{9}\underline{c}_{B_{(m)}}^{-4}\right)=o(1)
\ee
by Assumption \ref{a3}(i). This completes the proof of part (ii).
\section{Proof of Lemma \ref{lem:order}}\label{sec:order}
We only prove (i) as the proof of (ii) is analogous. Let $\delta_{n}=Ln^{-1}K\bar{k}\log n$ for some large constant $L<\infty$.
Let $\bar{Q}_{(m)}\left\{\left(b_{1},\cdots,b_{K},\bbeta_{(m)}^{T}\right)^{T}\right\}
=\sum_{k=1}^{K}\mathrm{E}\left\{\rho_{\tau_{k}}\left(y_{i}-b_{k(m)}-\boldsymbol{x}_{i(m)}^{T}\bbeta_{(m)}\right)\right\}$. Define
\be\label{Ddelta}
D(\delta_{n})&=&\inf_{1\leq m\leq M}\inf_{\sum_{k=1}^{K}\left\|\left(b_{k(m)}-b_{k(m)}^{*},\left(\bbeta_{(m)}-\bbeta_{(m)}^{*}\right)^{T}\right)^{T}\right\|^{2}>\delta_{n}}
\left\{\bar{Q}_{(m)}\left(\left(b_{1},\cdots,b_{K},\bbeta_{(m)}^{T}\right)^{T}\right)\right.\non\\
&&\left.-\bar{Q}_{(m)}\left(\left(b_{1}^{*},\cdots,b_{K}^{*},\bbeta_{(m)^{*}}^{T}\right)^{T}\right)\right\}
\ee
and
\be
\vartheta_{m}(\delta_{n})&=&\left\{\left(b_{1},\cdots,b_{K},\bbeta_{(m)}^{T}\right)^{T}:
\sum_{k=1}^{K}\left\|\left(b_{k(m)}-b_{k(m)}^{*},\left(\bbeta_{(m)}-\bbeta_{(m)}^{*}\right)^{T}\right)^{T}\right\|^{2}>\delta_{n},\right.\non\\
&&\left.\sum_{k=1}^{K}\left\|\left(b_{k(m)}-b_{k(m)}^{*},\left(\bbeta_{(m)}-\bbeta_{(m)}^{*}\right)^{T}\right)^{T}\right\|^{2}=o(1)\right\} \non.\ee
By Knight's identity, the definition of $u_{ki(m)}\left(=\mu_{i}-\left(1,\boldsymbol{x}_{i(m)}^{T}\right)\left(b_{k(m)}^{*},\bbeta_{(m)}^{*T}\right)^{T}\right)$
, Assumption \ref{a2}(ii) for any $\left(b_{1},\cdots,b_{K},\bbeta_{(m)}^{T}\right)^{T}\in \vartheta_{m}(\delta_{n})$
and Taylor expansion,
we have
\be
&&\bar{Q}_{(m)}\left(\left(b_{1},\cdots,b_{K},\bbeta_{(m)}^{T}\right)^{T}\right)-\bar{Q}_{(m)}\left(\left(b_{1}^{*},\cdots,b_{K}^{*},\bbeta_{(m)}^{*T}\right)^{T}\right)\non\\
&=&\sum_{k=1}^{K}\mathrm{E}\left\{\rho_{\tau_{k}}\left(y_{i}-b_{k(m)}-\boldsymbol{x}_{i(m)}^{T}\bbeta_{(m)}\right)\right\}
-\sum_{k=1}^{K}\mathrm{E}\left\{\rho_{\tau_{k}}\left(y_{i}-b_{k(m)}-\boldsymbol{x}_{i(m)}^{T}\bbeta_{(m)}^{*}\right)\right\}\non\\
&=&\sum_{k=1}^{K}\mathrm{E}\left\{\rho_{\tau_{k}}\left(\varepsilon_{ki(m)}+u_{ki(m)}-\left(1,\boldsymbol{x}_{i(m)}^{T}\right)\left\{b_{k(m)}-b_{k(m)}^{*},\left(\bbeta_{(m)}-\bbeta_{(m)}^{*}\right)^{T}\right\}^{T}\right)
\right.\non\\
&&\left.-\rho_{\tau_{k}}\left(\varepsilon_{ki(m)}+u_{ki(m)}\right)\right\}\non\\
&=&\sum_{k=1}^{K}\mathrm{E}\left(\int_{0}^{\left(1,\boldsymbol{x}_{i(m)}^{T}\right)\left\{b_{k(m)}-b_{k(m)}^{*},\left(\bbeta_{(m)}-\bbeta_{(m)}^{*}\right)^{T}\right\}^{T}}
\left[\boldsymbol{1}\{\varepsilon_{ki(m)}+u_{ki(m)}\leq s\}-\boldsymbol{1}\{\varepsilon_{ki(m)}+u_{ki(m)}\leq0\}\right]\mathrm{d}s
\right)\non\\
&=&\sum_{k=1}^{K}\mathrm{E}\left[\int_{0}^{\left(1,\boldsymbol{x}_{i(m)}^{T}\right)\left\{b_{k(m)}-b_{k(m)}^{*},\left(\bbeta_{(m)}-\bbeta_{(m)}^{*}\right)^{T}\right\}^{T}}
\left\{F(-u_{ki(m)}+s|\boldsymbol{x}_{i})-F(-u_{ki(m)}|\boldsymbol{x}_{i})\right\}\mathrm{d}s
\right]\non\\
&=&\sum_{k=1}^{K}\mathrm{E}\left[\int_{0}^{\left(1,\boldsymbol{x}_{i(m)}^{T}\right)\left\{b_{k(m)}-b_{k(m)}^{*},\left(\bbeta_{(m)}-\bbeta_{(m)}^{*}\right)^{T}\right\}^{T}}
\left\{f(-u_{ki(m)}|\boldsymbol{x}_{i})s+f'(-u_{ki(m)}+\iota_{2}s|\boldsymbol{x}_{i})s^{2}\right\}\mathrm{d}s
\right],
\ee  where $\iota_{2}\in(0,1)$. By the Mean-Value Theorem and Assumption \ref{a2}, we also have
\be
&&\bar{Q}_{(m)}\left(\left(b_{1},\cdots,b_{K},\bbeta_{(m)}^{T}\right)^{T}\right)-\bar{Q}_{(m)}\left(\left(b_{1}^{*},\cdots,b_{K}^{*},\bbeta_{(m)}^{*T}\right)^{T}\right)\non\\
&=&\sum_{k=1}^{K}\mathrm{E}\left[\int_{0}^{\left(1,\boldsymbol{x}_{i(m)}^{T}\right)\left\{b_{k(m)}-b_{k(m)}^{*},\left(\bbeta_{(m)}-\bbeta_{(m)}^{*}\right)^{T}\right\}^{T}}
\left\{\int_{0}^{1}f(-u_{ki(m)}+ts|\boldsymbol{x}_{i})\mathrm{d}t\cdot s\right\}\mathrm{d}s
\right]\non\\
&\leq&\frac{\bar{c}_{f}}{2}\sum_{k=1}^{K}\mathrm{E}\left\{\left(1,\boldsymbol{x}_{i(m)}^{T}\right)\left(b_{k(m)}-b_{k(m)}^{*},\left(\bbeta_{(m)}-\bbeta_{(m)}^{*}\right)^{T}\right)^{T}\right\}^{2}\non\\
&=&\frac{\bar{c}_{f}}{2}\sum_{k=1}^{K}\left(b_{k(m)}-b_{k(m)}^{*},\left(\bbeta_{(m)}-\bbeta_{(m)}^{*}\right)^{T}\right)
\mathrm{E}\left\{\left(1,\boldsymbol{x}_{i(m)}^{T}\right)^{T}\left(1,\boldsymbol{x}_{i(m)}^{T}\right)\right\}\non\\
&&\left(b_{k(m)}-b_{k(m)}^{*},\left(\bbeta_{(m)}-\bbeta_{(m)}^{*}\right)^{T}\right)\non\\
&\leq&\frac{\bar{c}_{A}}{2}\sum_{k=1}^{K}\left\|
\left(b_{k(m)}-b_{k(m)}^{*},\left(\bbeta_{(m)}-\bbeta_{(m)}^{*}\right)^{T}\right)^{T}\right\|^{2}.
\ee
For any $\left(b_{1},\cdots,b_{K},\bbeta_{(m)}^{T}\right)^{T}\in \vartheta_{m}(\delta_{n})$, we have
\be \label{eq:maxo}
\max_{1\leq k\leq K}
\left\|\left(b_{k(m)}-b_{k(m)}^{*},\left(\bbeta_{(m)}-\bbeta_{(m)}^{*}\right)^{T}\right)^{T}\right\|^{2}=o(1).\ee
By the part (ii) of Assumption \ref{a2} and Cauchy-Schwarz inequalities,
\be
&&\mathrm{E}\left\{\left(1,\boldsymbol{x}_{i(m)}^{T}\right)\left(b_{k(m)}-b_{k(m)}^{*},\left(\bbeta_{(m)}-\bbeta_{(m)}^{*}\right)^{T}\right)^{T}\right\}^{3}\non\\
&\leq&\left[\mathrm{E}\left\{\left(1,\boldsymbol{x}_{i(m)}^{T}\right)\left(b_{k(m)}-b_{k(m)}^{*},\left(\bbeta_{(m)}-\bbeta_{(m)}^{*}\right)^{T}\right)^{T}\right\}^{4}\right]^{1/2}\non\\
&&\left[\mathrm{E}\left\{\left(1,\boldsymbol{x}_{i(m)}^{T}\right)\left(b_{k(m)}-b_{k(m)}^{*},\left(\bbeta_{(m)}-\bbeta_{(m)}^{*}\right)^{T}\right)^{T}\right\}^{2}\right]^{1/2}\non\\
&\leq&\left[\left\|\left(b_{k(m)}-b_{k(m)}^{*},\left(\bbeta_{(m)}-\bbeta_{(m)}^{*}\right)^{T}\right)^{T}\right\|^{2}\right.\non\\
&&\left.\mathrm{E}\left\{\left(1,\boldsymbol{x}_{i(m)}^{T}\right)^{T}\left(1,\boldsymbol{x}_{i(m)}^{T}\right)\left(b_{k(m)}-b_{k(m)}^{*},\left(\bbeta_{(m)}-\bbeta_{(m)}^{*}\right)^{T}\right)^{T}\right\}^{2}\right]^{1/2}\non\\
&&\left[\mathrm{E}\left\{\left(1,\boldsymbol{x}_{i(m)}^{T}\right)\left(b_{k(m)}-b_{k(m)}^{*},\left(\bbeta_{(m)}-\bbeta_{(m)}^{*}\right)^{T}\right)^{T}\right\}^{2}\right]^{1/2}\non\\
&\leq&\left[\left\|\left(b_{k(m)}-b_{k(m)}^{*},\left(\bbeta_{(m)}-\bbeta_{(m)}^{*}\right)^{T}\right)^{T}\right\|^{2}\left(b_{k(m)}-b_{k(m)}^{*},\left(\bbeta_{(m)}-\bbeta_{(m)}^{*}\right)^{T}\right)\right.\non\\
&&\left.\mathrm{E}\left\{\left(1,\boldsymbol{x}_{i(m)}^{T}\right)^{T}\left(1,\boldsymbol{x}_{i(m)}^{T}\right)\right\}^{2}\left(b_{k(m)}-b_{k(m)}^{*},\left(\bbeta_{(m)}-\bbeta_{(m)}^{*}\right)^{T}\right)^{T}\right]^{1/2}\non\\
&&\left[\mathrm{E}\left\{\left(1,\boldsymbol{x}_{i(m)}^{T}\right)\left(b_{k(m)}-b_{k(m)}^{*},\left(\bbeta_{(m)}-\bbeta_{(m)}^{*}\right)^{T}\right)^{T}\right\}^{2}\right]^{1/2}\non\\
&\leq&\lambda_{\max}\left(\mathrm{E}\left\{\left(1,\boldsymbol{x}_{i(m)}^{T}\right)^{T}\left(1,\boldsymbol{x}_{i(m)}^{T}\right)\right\}^{2}\right)
\lambda_{\max}\left(\mathrm{E}\left\{\left(1,\boldsymbol{x}_{i(m)}^{T}\right)^{T}\left(1,\boldsymbol{x}_{i(m)}^{T}\right)\right\}\right)\non\\
&&\left\|\left(b_{k(m)}-b_{k(m)}^{*},\left(\bbeta_{(m)}-\bbeta_{(m)}^{*}\right)^{T}\right)^{T}\right\|^{4}\non\\
&\leq&\bar{c}_{A}^{3}/c_{f}^{3}
\left\|\left(b_{k(m)}-b_{k(m)}^{*},\left(\bbeta_{(m)}-\bbeta_{(m)}^{*}\right)^{T}\right)^{T}\right\|^{4}
\ee
and then
\be
&&\frac{\bar{c}_{f}\sum_{k=1}^{K}\frac{1}{3}\mathrm{E}\left|\left\{\left(1,\boldsymbol{x}_{i(m)}^{T}\right)\left(b_{k(m)}-b_{k(m)}^{*},\left(\bbeta_{(m)}-\bbeta_{(m)}^{*}\right)^{T}\right)^{T}\right\}^{3}\right|}
{\sum_{k=1}^{K}\frac{1}{2}\mathrm{E}f(-u_{ki(m)}|\boldsymbol{x}_{i})\left\{\left(1,\boldsymbol{x}_{i(m)}^{T}\right)\left(b_{k(m)}-b_{k(m)}^{*},\left(\bbeta_{(m)}-\bbeta_{(m)}^{*}\right)^{T}\right)^{T}\right\}^{2}}\non\\
&\leq&\frac{\bar{c}_{A}^{3}/c_{f}^{2}\sum_{k=1}^{K}
\left\|\left(b_{k(m)}-b_{k(m)}^{*},\left(\bbeta_{(m)}-\bbeta_{(m)}^{*}\right)^{T}\right)^{T}\right\|^{4}
}
{6\sum_{k=1}^{K}\left(b_{k(m)}-b_{k(m)}^{*},\left(\bbeta_{(m)}-\bbeta_{(m)}^{*}\right)^{T}\right)A_{k(m)}
\left(b_{k(m)}-b_{k(m)}^{*},\left(\bbeta_{(m)}-\bbeta_{(m)}^{*}\right)^{T}\right)^{T}}\non\\
&\leq&\frac{\bar{c}_{A}^{3}/c_{f}^{2}\sum_{k=1}^{K}
\left\|\left(b_{k(m)}-b_{k(m)}^{*},\left(\bbeta_{(m)}-\bbeta_{(m)}^{*}\right)^{T}\right)^{T}\right\|^{4}
}
{6\underline{c}_{A_{(m)}}\sum_{k=1}^{K}\left\|\left(b_{k(m)}-b_{k(m)}^{*},\left(\bbeta_{(m)}-\bbeta_{(m)}^{*}\right)^{T}\right)^{T}\right\|^{2}}\non\\
&\leq&\frac{\bar{c}_{A}^{3}\max_{1\leq k\leq K}
\left\|\left(b_{k(m)}-b_{k(m)}^{*},\left(\bbeta_{(m)}-\bbeta_{(m)}^{*}\right)^{T}\right)^{T}\right\|^{2}}
{6\underline{c}_{A}c_{f}^{2}}\non\\
&\doteq&A_{n,\bar{k}}=o(1),
\ee
where the last upper bound is consistent for all $\left(b_{1},\cdots,b_{K},\bbeta_{(m)}^{T}\right)^{T}\in \vartheta_{m}(\delta_{n})$
and its order may be related to $n$ and $\bar{k}$, denoted as $A_{n,\bar{k}}$.
Then
\be
&&\bar{Q}_{(m)}\left(\left(b_{1i(m)},\cdots,b_{Ki(m)},\bbeta_{i(m)}^{T}\right)^{T}\right)-\bar{Q}_{(m)}\left(\left(b_{1}^{*},\cdots,b_{K}^{*},\bbeta_{(m)}^{*T}\right)^{T}\right)\non\\
&=&\sum_{k=1}^{K}\mathrm{E}\left[\int_{0}^{\left(1,\boldsymbol{x}_{i(m)}^{T}\right)\left(b_{k(m)}-b_{k(m)}^{*},\left(\bbeta_{(m)}-\bbeta_{(m)}^{*}\right)^{T}\right)^{T}}
\left\{f(-u_{ki(m)}|\boldsymbol{x}_{i})s+f'(-u_{ki(m)}+\iota_{2}s|\boldsymbol{x}_{i})s^{2}\right\}\mathrm{d}s
\right]\non\\
&\geq& \sum_{k=1}^{K}\mathrm{E}\left[\left|\int_{0}^{\left(1,\boldsymbol{x}_{i(m)}^{T}\right)\left(b_{k(m)}-b_{k(m)}^{*},\left(\bbeta_{(m)}-\bbeta_{(m)}^{*}\right)^{T}\right)^{T}}
\left\{f(-u_{ki(m)}|\boldsymbol{x}_{i})s\right\}\mathrm{d}s\right|\right.\non\\
&&-
\left.\left|\int_{0}^{\left(1,\boldsymbol{x}_{i(m)}^{T}\right)\left(b_{k(m)}-b_{k(m)}^{*},\left(\bbeta_{(m)}-\bbeta_{(m)}^{*}\right)^{T}\right)^{T}}
\left\{f'(-u_{ki(m)}+\iota_{2}s|\boldsymbol{x}_{i})s^{2}\right\}\mathrm{d}s\right|\right]\non\\
&=& \left\{\frac{1}{2}\sum_{k=1}^{K}\left(b_{ki(m)}-b_{k(m)}^{*},\left(\bbeta_{i(m)}-\bbeta_{(m)}^{*}\right)^{T}\right)
A_{k(m)}\left(b_{ki(m)}-b_{k(m)}^{*},\left(\bbeta_{i(m)}-\bbeta_{(m)}^{*}\right)^{T}\right)^{T}\right\}\non\\
&&\left(1-\frac
{\sum_{k=1}^{K}\mathrm{E}\left[\left|\int_{0}^{\left(1,\boldsymbol{x}_{i(m)}^{T}\right)\left(b_{k(m)}-b_{k(m)}^{*},\left(\bbeta_{(m)}-\bbeta_{(m)}^{*}\right)^{T}\right)^{T}}
\left\{f'(-u_{ki(m)}+\iota_{2}s|\boldsymbol{x}_{i})s^{2}\right\}\mathrm{d}s
\right|\right]}{\sum_{k=1}^{K}\mathrm{E}\left[\int_{0}^{\left(1,\boldsymbol{x}_{i(m)}^{T}\right)\left(b_{k(m)}-b_{k(m)}^{*},\left(\bbeta_{(m)}-\bbeta_{(m)}^{*}\right)^{T}\right)^{T}}
\left\{f(-u_{ki(m)}|\boldsymbol{x}_{i})s\right\}\mathrm{d}s
\right]}\right)\non\\
&\geq&\frac{\underline{c}_{A}\delta_{n}}{2}\left(1-A_{n,\bar{k}}\right).
\ee
Referring to the definition of $D(\delta_{n})$ in (\ref{Ddelta}), we know
$D(\delta_{n})\geq\frac{\underline{c}_{A}\delta_{n}}{2}\left(1-A_{n,\bar{k}}\right)$.
Then by Boole's inequality,
we obtain
\be\label{W:1}
&&\mathrm{P}\left(\max_{1\leq i\leq n}\max_{1\leq m\leq M}\sum_{k=1}^{K}\left\|\left(\hat{b}_{ki(m)}-b_{k(m)}^{*},\left(\hat{\bbeta}_{i(m)}-\bbeta_{(m)}^{*}\right)^{T}\right)^{T}\right\|\geq\delta_{n}\right)\non\\
&\leq&nM\max_{1\leq i\leq n}\max_{1\leq m\leq M}\mathrm{P}\left(\sum_{k=1}^{K}\left\|\left(\hat{b}_{ki(m)}-b_{k(m)}^{*},\left(\hat{\bbeta}_{i(m)}-\bbeta_{(m)}^{*}\right)^{T}\right)^{T}\right\|\geq\delta_{n}\right)\non\\
&\leq&nM\max_{1\leq i\leq n}\max_{1\leq m\leq M}\mathrm{P}\left(\bar{Q}_{(m)}\left(\left(\hat{b}_{1i(m)},\cdots,\hat{b}_{Ki(m)},\hat{\bbeta}_{i(m)}^{T}\right)^{T}\right)\right.\non\\
&&~~~~~~~~~~~~~~~~~~~~~~~~~~~~\left.-\bar{Q}_{(m)}\left(\left(b_{1}^{*},\cdots,b_{K}^{*},\bbeta_{(m)}^{*T}\right)^{T}\right)\geq D(\delta_{n})\right)\non\\
&\leq&nM\max_{1\leq i\leq n}\max_{1\leq m\leq M}\mathrm{P}\left(\frac{\bar{c}_{A}}{2}\mathbb{W}_{i(m)}\geq 2n\frac{\underline{c}_{A}\delta_{n}}{2}\left(1-A_{n,\bar{k}}\right)\right)\non\\
&=&nM\max_{1\leq i\leq n}\max_{1\leq m\leq M}\mathrm{P}\left(\mathbb{W}_{i(m)}\geq \frac{n\underline{c}_{A}\delta_{n}}{\bar{c}_{A}}\left(1-A_{n,\bar{k}}\right)\right),
\ee
where $$\mathbb{W}_{i(m)}=n\sum_{k=1}^{K}\left\|
\left(\hat{b}_{ki(m)}-b_{k(m)}^{*},\left(\hat{\bbeta}_{i(m)}-\bbeta_{(m)}^{*}\right)^{T}\right)\right\|^{2}
.$$
Following the proof of Lemma \ref{lem:distr} , we can also show that
$$\sqrt{n}C_ {(m)}V_{(m)}^{-1/2}\left\{\left(\hat{b}_{1i(m)},\cdots,\hat{b}_{Ki(m)},\hat{\bbeta}_{i(m)}^{T}\right)^{T}-\left(b_{1(m)}^{*},\cdots,b_{K(m)}^{*},\bbeta_{(m)}^{*T}\right)^{T}\right\}
\overd N(0,C_{0}),$$ which implies that
\be\label{W:2}
\tilde{\beta}_{i(m)}&=&\sqrt{n}\left(C_{(m)}C_{(m)}^{T}\right)^{-1/2}C_{(m)}V_{(m)}^{-1/2}\left\{\left(\hat{b}_{1i(m)},\cdots,\hat{b}_{Ki(m)},\hat{\bbeta}_{i(m)}^{T}\right)^{T}-\left(b_{1(m)}^{*},\cdots,b_{K(m)}^{*},\bbeta_{(m)}^{*T}\right)^{T}\right\}\non\\
&\overd& \mathcal{N}(0,I_{l_{m}}).\ee
Rewriting $\sqrt{n}\left\{\left(\hat{b}_{1i(m)},\cdots,\hat{b}_{Ki(m)},\hat{\bbeta}_{i(m)}^{T}\right)^{T}-\left(b_{1(m)}^{*},\cdots,b_{K(m)}^{*},\bbeta_{(m)}^{*T}\right)^{T}\right\}
$ in terms of $\tilde{\beta}_{i(m)}$ yields
\be
&&\sqrt{n}\left\{\left(\hat{b}_{1i(m)},\cdots,\hat{b}_{Ki(m)},\hat{\bbeta}_{i(m)}^{T}\right)^{T}-\left(b_{1(m)}^{*},\cdots,b_{K(m)}^{*},\bbeta_{(m)}^{*T}\right)^{T}\right\}\non\\
&=&\left\{\left(C_{(m)}C_{(m)}^{T}\right)^{-1/2}C_{(m)}V_{(m)}^{-1/2}\right\}^{-}\tilde{\beta}_{i(m)},
\ee
where $\left\{\left(C_{(m)}C_{(m)}^{T}\right)^{-1/2}C_{(m)}V_{(m)}^{-1/2}\right\}^{-}$ is the generalized inverse matrix
$\left(C_{(m)}C_{(m)}^{T}\right)^{-1/2}C_{(m)}V_{(m)}^{-1/2}$.
Let $G_{k}$ denote a $\left(\tilde{k}_{m}+1\right)\times\left(\tilde{k}_{m}+K\right)$ matrix such that
\be
&&\sqrt{n}G_{k}\left\{\left(\hat{b}_{1i(m)},\cdots,\hat{b}_{Ki(m)},\hat{\bbeta}_{i(m)}^{T}\right)^{T}-\left(b_{1(m)}^{*},\cdots,b_{K(m)}^{*},\bbeta_{(m)}^{*T}\right)^{T}\right\}\non\\
&=&\sqrt{n}\left\{\left(\hat{b}_{ki(m)},\hat{\bbeta}_{i(m)}^{T}\right)^{T}-\left(b_{k(m)}^{*},\bbeta_{(m)}^{*T}\right)^{T}\right\}.
\ee Note that $\sum_{k=1}^{K}G_{k}^{T}G_{k}=I_{\tilde{k}_{m}+K}$, where $I_{\tilde{k}_{m}+K}$ is a $\left(\tilde{k}_{m}+K\right)\times\left(\tilde{k}_{m}+K\right)$ unit matrix.
By Assumptions \ref{a2}-\ref{a3}, it follows that
\be
\mathbb{W}_{i(m)}&=&n\sum_{k=1}^{K}\left\{\hat{b}_{ki(m)}-b_{k(m)}^{*},\left(\hat{\bbeta}_{i(m)}-\bbeta_{(m)}^{*}\right)^{T}\right\}\left\{\hat{b}_{ki(m)}-b_{k(m)}^{*},\left(\hat{\bbeta}_{i(m)}-\bbeta_{(m)}^{*}\right)^{T}\right\}^{T}\non\\
&=&\sum_{k=1}^{K}\tilde{\beta}_{i(m)}^{T}\left[G_{k}\left\{\left(C_{(m)}C_{(m)}^{T}\right)^{-1/2}C_{(m)}V_{(m)}^{-1/2}\right\}^{-}\right]^{T}\non\\
&&G_{k}\left\{\left(C_{(m)}C_{(m)}^{T}\right)^{-1/2}C_{(m)}V_{(m)}^{-1/2}\right\}^{-}\tilde{\beta}_{i(m)}\non\\
&=& \sum_{k=1}^{K}G_{k}^{T}G_{k}
\tilde{\beta}_{i(m)}^{T}\left\{\left(C_{(m)}C_{(m)}^{T}\right)^{-1/2}C_{(m)}V_{(m)}^{-1}
C_{(m)}^{T}\left(C_{(m)}C_{(m)}^{T}\right)^{-1/2}\right\}^{-1}\tilde{\beta}_{i(m)}\non\\
&=&
\tilde{\beta}_{i(m)}^{T}\left\{\left(C_{(m)}C_{(m)}^{T}\right)^{1/2}\left(C_{(m)}V_{(m)}^{-1}
C_{(m)}^{T}\right)^{-1}\left(C_{(m)}C_{(m)}^{T}\right)^{1/2}\right\}\tilde{\beta}_{i(m)}\non\\
&\leq& \lambda_{\max}(V_{(m)})\tilde{\beta}_{i(m)}^{T}\tilde{\beta}_{i(m)}
\leq \left(\bar{c}_{\bar{B}}/\underline{c}_{\bar{A}}^{2}\right)\left\|\tilde{\beta}_{i(m)}\right\|^{2},
\ee
where we have used the fact that $A^{T}BA\leq \lambda_{\max}(B)A^{T}A$ for any real symmetric matrix B and conformable matrix $A$ and that $A^{-T}A=(AA^{T})^{-}$ (see, e.g., \cite{Bernstein2005Matrix} Proposition 6.1.6xvii). Let $c_{AB}=\bar{c}_{A}\bar{c}_{\bar{B}}/\underline{c}_{\bar{A}}^{2}$ and
$\bar{l}=\max_{1\leq m\leq M}l_{m}$. Then by Lemma 2.1 of \cite{Bartlett1981An}
\be\label{W:3}
&&nM\max_{1\leq i\leq n}\max_{1\leq m\leq M}\mathrm{P}\left(\mathbb{W}_{i(m)}\geq n\underline{c}_{A}\delta_{n}\left(1-A_{n,\bar{k}}\right)/\bar{c}_{A}\right)\non\\
&\leq&nM\max_{1\leq i\leq n}\max_{1\leq m\leq M}\mathrm{P}\left(\left\|\tilde{\beta}_{i(m)}\right\|^{2}\geq n\underline{c}_{A}\delta_{n}\left(1-A_{n,\bar{k}}\right)/c_{AB}\right)\non\\
&\leq&\limsup_{n\rightarrow\infty}nM\max_{1\leq m\leq M}\mathrm{P}\left(\chi^{2}(l_{m})\geq n\underline{c}_{A}\delta_{n}\left(1-A_{n,\bar{k}}\right)/c_{AB}\right)\non\\
&\leq&\limsup_{n\rightarrow\infty}nM\cdot\mathrm{P}\left(\chi^{2}(\bar{l})\geq n\underline{c}_{A}\delta_{n}\left(1-A_{n,\bar{k}}\right)/c_{AB}\right)\non\\
&\leq&\limsup_{n\rightarrow\infty}nM\cdot\mathrm{P}\left(\chi^{2}(\bar{l})\geq \bar{l}+\left\{n\underline{c}_{A}\delta_{n}\left(1-A_{n,\bar{k}}\right)/c_{AB}-\bar{l}\right\}\right)\non\\
&=&\limsup_{n\rightarrow\infty}nM\exp\left\{-\frac{\left(n\underline{c}_{A}\delta_{n}\left(1-A_{n,\bar{k}}\right)/c_{AB}-\bar{l}\right)}{2}\right.\non\\
&&\left.\cdot\left[1-\log\left\{n\underline{c}_{A}\delta_{n}\left(1-A_{n,\bar{k}}\right)/\left(\bar{l}c_{AB}\right)\right\}
/\left\{n\underline{c}_{A}\delta_{n}\left(1-A_{n,\bar{k}}\right)/\left(\bar{l}c_{AB}\right)-1\right\}\right]
\right\}\non\\
&=&0,
\ee
because $A_{n,\bar{k}}=o(1)$, $nM\exp\left\{-n\underline{c}_{A}\delta_{n}/\left(2c_{AB}\right)\right\}
=nMn^{-0.5LK\bar{k}\underline{c}_{A}/c_{AB}}=o(1)$ for
sufficiently large $L$ and $\log\left\{n\underline{c}_{A}\delta_{n}/\left(\bar{l}c_{AB}\right)\right\}
/\left\{n\underline{c}_{A}\delta_{n}/\left(\bar{l}c_{AB}\right)-1\right\}=o(1)$ under Assumption \ref{a3}(ii). Combining (\ref{W:1})-(\ref{W:3}), we have shown
$$\limsup_{n\rightarrow\infty}\mathrm{P}\left(\max_{1\leq i\leq n}\max_{1\leq m\leq M}\sum_{k=1}^{K}\left\|\hat{b}_{ki(m)}-b_{k(m)}^{*},\left(\hat{\bbeta}_{i(m)}-\bbeta_{(m)}^{*}\right)^{T}\right\|^{2}\geq\delta_{n}\right)=0$$
that and thus (i) follows.
\section{Proof of Theorem \ref{th3.1}}\label{sec:th3.1}
Following the proof of Theorem 2.1 in \cite{Li1987Asymptotic}, it suffices to show, as $n\rightarrow\infty$, that
\begin{eqnarray}
\label{eq5.23}
\sup_{\boldsymbol{w}\in\mathcal{W}}\left| \frac{CV_{n}(\boldsymbol{w})
-\mathrm{CPE}_{n}(\boldsymbol{w})}{\mathrm{CPE}_{n}(\boldsymbol{w})}\right|
=o_{p}(1)
\end{eqnarray}
in the weight space $\mathcal{W}$.

Similar to the proof of Theorem 3.3 in \cite{lu2015jackknife}, we have
\begin{eqnarray}
&&CV_{n}(\boldsymbol{w})-\mathrm{CPE}_{n}(\boldsymbol{w})\nonumber\\
&=&\frac{1}{n}\sum_{i=1}^{n}\left\{\sum_{k=1}^{K}\rho_{\tau_{k}}\left(y_{i}-\sum_{m=1}^{M}w_{m}\hat{b}_{ki(m)}-\sum_{m=1}^{M}w_{m}\boldsymbol{x}_{i(m)}^{T}\hat{\bbeta}_{i(m)}\right)\right\}\non\\
&&-\mathrm{E}\left\{\sum_{k=1}^{K}\rho_{\tau_{k}}\left(y-\sum_{m=1}^{M}w_{m}\hat{b}_{k,m}-\sum_{m=1}^{M}w_{m}\boldsymbol{x}_{m}^{T}\hat{\bbeta}_{(m)}\right)\Big|\mathcal{D}_{n}\right\}\non\\
&=&\frac{1}{n}\sum_{i=1}^{n}\left[\sum_{k=1}^{K}\left\{\rho_{\tau_{k}}\left(y_{i}-\sum_{m=1}^{M}w_{m}\hat{b}_{ki(m)}-\sum_{m=1}^{M}w_{m}\boldsymbol{x}_{i(m)}^{T}\hat{\bbeta}_{i(m)}\right)-\rho_{\tau_{k}}(\varepsilon_{i}-b_{k})\right\}\right]\non\\
&&-\left[\mathrm{CPE}_{n}(\boldsymbol{w})-\mathrm{E}\left\{\sum_{k=1}^{K}\rho_{\tau_{k}}(\varepsilon-b_{k})\right\}\right]\non\\
&&+\frac{1}{n}\sum_{i=1}^{n}\left[\sum_{k=1}^{K}\rho_{\tau_{k}}(\varepsilon_{i}-b_{k})-\mathrm{E}\left\{\sum_{k=1}^{K}\rho_{\tau_{k}}(\varepsilon-b_{k})\right\}\right]\non\\
&=&\frac{1}{n}\sum_{i=1}^{n}\left[\sum_{k=1}^{K}\left\{\rho_{\tau_{k}}\left(\varepsilon_{i}-b_{k}+b_{k}-\sum_{m=1}^{M}w_{m}\hat{b}_{ki(m)}+\mu_{i}-\sum_{m=1}^{M}w_{m}\boldsymbol{x}_{i(m)}^{T}\hat{\bbeta}_{i(m)}\right)-\rho_{\tau_{k}}(\varepsilon_{i}-b_{k})\right\}\right]\non\\
&&-\left[\mathrm{CPE}_{n}(\boldsymbol{w})-\mathrm{E}\left\{\sum_{k=1}^{K}\rho_{\tau_{k}}(\varepsilon-b_{k})\right\}\right]\non\\
&&+\frac{1}{n}\sum_{i=1}^{n}\left[\sum_{k=1}^{K}\rho_{\tau_{k}}(\varepsilon_{i}-b_{k})-\mathrm{E}\left\{\sum_{k=1}^{K}\rho_{\tau_{k}}(\varepsilon-b_{k})\right\}\right]\non\\
&=&CV_{1n}(\boldsymbol{w})+CV_{2n}(\boldsymbol{w})+CV_{3n}(\boldsymbol{w})+CV_{4n}(\boldsymbol{w})+CV_{5n},\nonumber
\end{eqnarray}
where
\be
&&CV_{1n}(\boldsymbol{w})=\frac{1}{n}\sum_{i=1}^{n}\sum_{k=1}^{K}\left\{\left(b_{k}-\sum_{m=1}^{M}w_{m}\hat{b}_{ki(m)}+\mu_{i}-\sum_{m=1}^{M}w_{m}\boldsymbol{x}_{i(m)}^{T}\hat{\bbeta}_{i(m)}\right)\psi_{\tau}(\varepsilon_{i}-b_{k})\right\},\non\\
&&CV_{2n}(\boldsymbol{w})=\frac{1}{n}\sum_{i=1}^{n}\sum_{k=1}^{K}\int_{0}^{\sum_{m=1}^{M}w_{m}\hat{b}_{ki(m)}-b_{k}+\sum_{m=1}^{M}w_{m}\boldsymbol{x}_{i(m)}^{T}\hat{\bbeta}_{i(m)}-\mu_{i}}
   \left[\boldsymbol{1}\{\varepsilon_{i}-b_{k}\leq s\}\right.\non\\&&\left.~~~~~~~~~~~~~~~~~~~~-\boldsymbol{1}\{\varepsilon_{i}-b_{k}\leq 0\}-F(s+b_{k}|\boldsymbol{x}_{i})+F(b_{k}|\boldsymbol{x}_{i})\right]\mathrm{d}s,\non\\
&&CV_{3n}(\boldsymbol{w})=\frac{1}{n}\sum_{i=1}^{n}\sum_{k=1}^{K}\left(\int_{0}^{\sum_{m=1}^{M}w_{m}\hat{b}_{ki(m)}-b_{k}+\sum_{m=1}^{M}w_{m}\boldsymbol{x}_{i(m)}^{T}\hat{\bbeta}_{i(m)}-\mu_{i}}\left\{F(s+b_{k}|\boldsymbol{x}_{i})-F(b_{k}|\boldsymbol{x}_{i})\right\}\mathrm{d}s\right.\non\\
    &&~~~~~~~~~~~~~~~~~~~~\left.-\mathrm{E}_{\boldsymbol{x}_{i}}\left[\int_{0}^{\sum_{m=1}^{M}w_{m}\hat{b}_{ki(m)}-b_{k}+\sum_{m=1}^{M}w_{m}\boldsymbol{x}_{i(m)}^{T}\hat{\bbeta}_{i(m)}-\mu_{i}}\left\{F(s+b_{k}|\boldsymbol{x}_{i})-F(b_{k}|\boldsymbol{x}_{i})\right\}\mathrm{d}s\right]\right),\non\\
&&CV_{4n}(\boldsymbol{w})=\frac{1}{n}\sum_{i=1}^{n}\sum_{k=1}^{K}\mathrm{E}_{\boldsymbol{x}_{i}}\left(\int_{0}^{\sum_{m=1}^{M}w_{m}\hat{b}_{ki(m)}-b_{k}+\sum_{m=1}^{M}w_{m}\boldsymbol{x}_{i(m)}^{T}\hat{\bbeta}_{i(m)}-\mu_{i}}\left\{F(s+b_{k}|\boldsymbol{x}_{i})-F(b_{k}|\boldsymbol{x}_{i})\right\}\mathrm{d}s\right.\non\\
    &&~~~~~~~~~~~~~~~~~~~~\left.-\mathrm{E}_{\boldsymbol{x}_{i}}\left[\int_{0}^{\sum_{m=1}^{M}w_{m}\hat{b}_{k(m)}-b_{k}+\sum_{m=1}^{M}w_{m}\boldsymbol{x}_{i(m)}^{T}\hat{\bbeta}_{(m)}-\mu_{i}}\left\{F(s+b_{k}|\boldsymbol{x}_{i})-F(b_{k}|\boldsymbol{x}_{i})\right\}\mathrm{d}s\right]\right),\non\\
\mathrm{and}\non\\
&&CV_{5n}=\frac{1}{n}\sum_{i=1}^{n}\sum_{k=1}^{K}\left\{\rho_{\tau_{k}}(\varepsilon_{i}-b_{k})-\mathrm{E}\rho_{\tau_{k}}(\varepsilon_{i}-b_{k})\right\}\non.
\ee
Hence
\allowdisplaybreaks
\be
&&\sup_{\boldsymbol{w}\in\mathcal{W}}\left| \frac{CV_{n}(\boldsymbol{w})
-\mathrm{CPE}_{n}(\boldsymbol{w})}{\mathrm{CPE}_{n}(\boldsymbol{w})}\right|\non\\
&=&\sup_{\boldsymbol{w}\in\mathcal{W}}\left| \frac{CV_{1n}(\boldsymbol{w})+CV_{2n}(\boldsymbol{w})+CV_{3n}(\boldsymbol{w})+CV_{4n}(\boldsymbol{w})+CV_{5n}}
{\mathrm{CPE}_{n}(\boldsymbol{w})}\right|\non\\
&\leq& \frac{\sup_{\boldsymbol{w}\in\mathcal{W}}\left|CV_{1n}(\boldsymbol{w})+CV_{2n}(\boldsymbol{w})+CV_{3n}(\boldsymbol{w})+CV_{4n}(\boldsymbol{w})+CV_{5n}\right|}
{\min_{\boldsymbol{w}\in\mathcal{W}}\left|\mathrm{CPE}_{n}(\boldsymbol{w})\right|}\non\\
&\leq& \left\{\sup_{\boldsymbol{w}\in\mathcal{W}}\left|CV_{1n}(\boldsymbol{w})\right|+\sup_{\boldsymbol{w}\in\mathcal{W}}\left|CV_{2n}(\boldsymbol{w})\right|
+\sup_{\boldsymbol{w}\in\mathcal{W}}\left|CV_{3n}(\boldsymbol{w})\right|\right.\non\\
&&\left.~~~+\sup_{\boldsymbol{w}\in\mathcal{W}}\left|CV_{4n}(\boldsymbol{w})\right|
+\sup_{\boldsymbol{w}\in\mathcal{W}}\left|CV_{5n}\right|\right\}\Big/
\min_{\boldsymbol{w}\in\mathcal{W}}\left|\mathrm{CPE}_{n}(\boldsymbol{w})\right|.\non
\ee
Therefore, to prove \ref{eq5.23}, we need to prove
\begin{equation*}
\begin{aligned}
&(\mathrm{i})\mathrm{min}_{\boldsymbol{w}\in\mathcal{W}}
\mathrm{CPE}_{n}(\boldsymbol{w})\geq \sum_{k=1}^{K}\mathrm{E}\left\{\rho_{\tau_{k}}(\varepsilon-b_{k})\right\}-o_{p}(1);\\
&(\mathrm{ii})\mathrm{sup}_{\boldsymbol{w}\in\mathcal{W}}|CV_{1n}(\boldsymbol{w})|=o_{p}(1);\\
&(\mathrm{iii})\mathrm{sup}_{\boldsymbol{w}\in\mathcal{W}}|CV_{2n}(\boldsymbol{w})|=o_{p}(1);\\
&(\mathrm{iv})\mathrm{sup}_{\boldsymbol{w}\in\mathcal{W}}|CV_{3n}(\boldsymbol{w})|=o_{p}(1);\\
&(\mathrm{v})\mathrm{sup}_{\boldsymbol{w}\in\mathcal{W}}|CV_{4n}(\boldsymbol{w})|=o_{p}(1);\\
&(\mathrm{vi})CV_{5n}=o_{p}(1).
\end{aligned}
\end{equation*}
We establish (i)-(v) merely, because (vi) follows by the weak law of large
numbers.

Firstly, let us establish (i). Let $u_{k}(\boldsymbol{w})=\mu-\sum_{m=1}^{M}w_{m}\boldsymbol{x}_{m}^{T}\bbeta_{(m)}^{*}
+b_{k}-\sum_{m=1}^{M}w_{m}b_{k(m)}^{*}$. Then
\be
&&\mathrm{CPE}_{n}(\boldsymbol{w})-\mathrm{E}\left\{\sum_{k=1}^{K}\rho_{\tau_{k}}\left(\varepsilon-b_{k}+u_{k}(\boldsymbol{w})\right)\right\}\non\\
&=&\mathrm{E}\left\{\sum_{k=1}^{K}\rho_{\tau_{k}}\left(y-\sum_{m=1}^{M}w_{m}\hat{b}_{k(m)}-\sum_{m=1}^{M}w_{m}\boldsymbol{x}_{m}^{T}\hat{\bbeta}_{(m)}\right)\Big|\mathcal{D}_{n}\right\}
-\mathrm{E}\left\{\sum_{k=1}^{K}\rho_{\tau_{k}}\left(\varepsilon-b_{k}+u_{k}(\boldsymbol{w})\right)\right\}\non\\
&=&\mathrm{E}\left[\sum_{k=1}^{K}\left\{\rho_{\tau_{k}}\left(y-\sum_{m=1}^{M}w_{m}\hat{b}_{k,m}-\sum_{m=1}^{M}w_{m}\boldsymbol{x}_{m}^{T}\hat{\bbeta}_{(m)}\right)
-\rho_{\tau_{k}}\left(\varepsilon-b_{k}+u_{k}(\boldsymbol{w})\right)\right\}\Big|\mathcal{D}_{n}\right].\non
\ee
By Knight's identity,
\be
&&\rho_{\tau_{k}}\left(y-\sum_{m=1}^{M}w_{m}\hat{b}_{k(m)}-\sum_{m=1}^{M}w_{m}\boldsymbol{x}_{m}^{T}\hat{\bbeta}_{(m)}\right)
-\rho_{\tau_{k}}\left(\varepsilon-b_{k}+u_{k}(\boldsymbol{w})\right)\non\\
&=&\rho_{\tau_{k}}\left(\varepsilon-b_{k}+u_{k}(\boldsymbol{w})-\sum_{m=1}^{M}w_{m}\left(\hat{b}_{k(m)}-b_{k(m)}^{*}\right)-\sum_{m=1}^{M}w_{m}\boldsymbol{x}_{m}^{T}\left(\hat{\bbeta}_{(m)}-\bbeta_{(m)}^{*}\right)\right)
-\rho_{\tau_{k}}\left(\varepsilon-b_{k}+u_{k}(\boldsymbol{w})\right)\non\\
&=&-\left\{\sum_{m=1}^{M}w_{m}\left(\hat{b}_{k(m)}-b_{k(m)}^{*}\right)+\sum_{m=1}^{M}w_{m}\boldsymbol{x}_{m}^{T}\left(\hat{b}_{k(m)}-b_{k(m)}^{*}\right)\right\}
\psi_{\tau_{k}}\left(\varepsilon-b_{k}+u_{k}(\boldsymbol{w})\right)\non\\
&&+\int_{0}^{\sum_{m=1}^{M}w_{m}\left(\hat{b}_{k(m)}-b_{k(m)}^{*}\right)+\sum_{m=1}^{M}w_{m}\boldsymbol{x}_{m}^{T}\left(\hat{b}_{k(m)}-b_{k(m)}^{*}\right)}
\left[\boldsymbol{1}\{\varepsilon-b_{k}+u_{k}(\boldsymbol{w})\leq s\}-\boldsymbol{1}\{\varepsilon-b_{k}+u_{k}(\boldsymbol{w})\leq 0\}\right]\mathrm{d}s,\non
\ee
where $\psi_{\tau_{k}}(\cdot)=\tau_{k}-\boldsymbol{1}\{\cdot <0\}$.
And the first order condition for the population minimization
problem (\ref{problem:2}) implies that
$\mathrm{E}\left[\sum_{k=1}^{K}\psi_{\tau_{k}}\left\{\varepsilon-b_{k}+u_{k}(\boldsymbol{w})\right\}(1,\boldsymbol{x}_{m}^{T})^{T}\right]=0$, therefore,
\be
&&\mathrm{CPE}_{n}(\boldsymbol{w})-\mathrm{E}\left\{\sum_{k=1}^{K}\rho_{\tau_{k}}\left(\varepsilon-b_{k}+u_{k}(\boldsymbol{w})\right)\right\}\non\\
&=&\mathrm{E}\left(\sum_{k=1}^{K}\int_{0}^{\sum_{m=1}^{M}w_{m}\left(\hat{b}_{k(m)}-b_{k(m)}^{*}\right)+\sum_{m=1}^{M}w_{m}\boldsymbol{x}_{m}^{T}\left(\hat{\bbeta}_{(m)}-\bbeta_{(m)}^{*}\right)}
\left[\boldsymbol{1}\{\varepsilon-b_{k}+u_{k}(\boldsymbol{w})\leq s\}\right.\right.\non\\
&&\left.\left.-\boldsymbol{1}\{\varepsilon-b_{k}+u_{k}(\boldsymbol{w})\leq 0\}\right]\mathrm{d}s\mid\mathcal{D}_{n}\right)\non\\
&=&\mathrm{E}_{\boldsymbol{x}}\left[\sum_{k=1}^{K}\int_{0}^{\sum_{m=1}^{M}w_{m}\left(\hat{b}_{k(m)}-b_{k(m)}^{*}\right)+\sum_{m=1}^{M}w_{m}\boldsymbol{x}_{m}^{T}\left(\hat{\bbeta}_{(m)}-\bbeta_{(m)}^{*}\right)}
\left\{F\left(s+b_{k}-u_{k}(\boldsymbol{w})\mid\boldsymbol{x}\right)\right.\right.\non\\
&&\left.\left.-F\left(b_{k}-u_{k}(\boldsymbol{w})\mid\boldsymbol{x}\right)\right\}\mathrm{d}s\right]\non.
\ee
Further, by Taylor expansion, Jensen inequality, Assumption \ref{a2}-\ref{a3} and Lemma \ref{lem:distr},
\be
&&\mathrm{CPE}_{n}(\boldsymbol{w})-\mathrm{E}\left\{\sum_{k=1}^{K}\rho_{\tau_{k}}\left(\varepsilon-b_{k}+u_{k}(\boldsymbol{w})\right)\right\}\non\\
&=&\mathrm{E}_{\boldsymbol{x}}\left[\sum_{k=1}^{K}\int_{0}^{\sum_{m=1}^{M}w_{m}\left(\hat{b}_{k(m)}-b_{k(m)}^{*}\right)+\sum_{m=1}^{M}w_{m}\boldsymbol{x}_{m}^{T}\left(\hat{\bbeta}_{(m)}-\bbeta_{(m)}^{*}\right)}
\left\{F\left(s+b_{k}-u_{k}(\boldsymbol{w})\mid\boldsymbol{x}\right)\right.\right.\non\\
&&\left.\left.-F(b_{k}-u_{k}(\boldsymbol{w})\mid\boldsymbol{x})\right\}\mathrm{d}s\right]\non\\
&=&\mathrm{E}_{\boldsymbol{x}}\left[\sum_{k=1}^{K}\int_{0}^{\sum_{m=1}^{M}w_{m}\left(\hat{b}_{k(m)}-b_{k(m)}^{*}\right)+\sum_{m=1}^{M}w_{m}\boldsymbol{x}_{m}^{T}\left(\hat{\bbeta}_{(m)}-\bbeta_{(m)}^{*}\right)}
\left\{f\left(b_{k}-u_{k}(\boldsymbol{w})\mid\boldsymbol{x}\right)s
\right.\right.\non\\
&&\left.\left.+f'\left(\kappa s+b_{k}-u_{k}(\boldsymbol{w})\mid\boldsymbol{x}\right)s^{2}\right\}\mathrm{d}s\right]\non\\
&=&\mathrm{E}_{\boldsymbol{x}}\left\{\sum_{k=1}^{K}\int_{0}^{\sum_{m=1}^{M}w_{m}\left(\hat{b}_{k(m)}-b_{k(m)}^{*}\right)+\sum_{m=1}^{M}w_{m}\boldsymbol{x}_{m}^{T}\left(\hat{\bbeta}_{(m)}-\bbeta_{(m)}^{*}\right)}
f\left(b_{k}-u_{k}(\boldsymbol{w})\mid\boldsymbol{x}\right)s\mathrm{d}s\right\}\non\\
&&+\mathrm{E}_{\boldsymbol{x}}\left\{\sum_{k=1}^{K}\int_{0}^{\sum_{m=1}^{M}w_{m}\left(\hat{b}_{k(m)}-b_{k(m)}^{*}\right)+\sum_{m=1}^{M}w_{m}\boldsymbol{x}_{m}^{T}\left(\hat{\bbeta}_{(m)}-\bbeta_{(m)}^{*}\right)}
f'\left(\iota_{3} s+b_{k}-u_{k}(\boldsymbol{w})\mid\boldsymbol{x}\right)s^{2}\mathrm{d}s\right\}\non,
\ee where $\iota_{3} \in(0,1)$.
By Cauchy-Schwarz inequalities and the part (ii) of Lemma \ref{lem:distr}, we obtain that \be
&&\max_{1\leq k\leq K}\max_{1\leq m\leq M}
\left\|\left(\hat{b}_{k(m)}-b_{k(m)}^{*},\left(\hat{\bbeta}_{(m)}-\bbeta_{(m)}^{*}\right)^{T}\right)^{T}\right\|\non\\
&\leq&\sum_{k=1}^{K}\max_{1\leq m\leq M}
\left\|\left(\hat{b}_{k(m)}-b_{k(m)}^{*},\left(\hat{\bbeta}_{(m)}-\bbeta_{(m)}^{*}\right)^{T}\right)^{T}\right\|
\non\\
&\leq&\left\{K\sum_{k=1}^{K}\max_{1\leq m\leq M}
\left\|\left(\hat{b}_{k(m)}-b_{k(m)}^{*},\left(\hat{\bbeta}_{(m)}-\bbeta_{(m)}^{*}\right)^{T}\right)^{T}\right\|^{2}
\right\}^{1/2}\non\\
&=&O_{p}\left(n^{-1/2}K(\bar{k}\log n)^{1/2} \right).
\ee
Then combing with the part (i) of Assumption \ref{a2}, Jensen inequality and assumption $MK^{2}\bar{k}^{2}\log n/n\rightarrow0$, we have
\be
&&\mathrm{E}_{\boldsymbol{x}}\left\{\sum_{k=1}^{K}\int_{0}^{\sum_{m=1}^{M}w_{m}\left(\hat{b}_{k(m)}-b_{k(m)}^{*}\right)+\sum_{m=1}^{M}w_{m}\boldsymbol{x}_{m}^{T}\left(\hat{\bbeta}_{(m)}-\bbeta_{(m)}^{*}\right)}
f'\left(\iota_{3} s+b_{k}-u_{k}(\boldsymbol{w})\mid\boldsymbol{x}\right)s^{2}\mathrm{d}s\right\}\non\\
&\leq&\frac{\bar{c}_{f}}{3}\sum_{k=1}^{K}
\mathrm{E}_{\boldsymbol{x}}\left\{\sum_{m=1}^{M}w_{m}\left(\hat{b}_{k(m)}-b_{k(m)}^{*}\right)+\sum_{m=1}^{M}w_{m}\boldsymbol{x}_{m}^{T}\left(\hat{\bbeta}_{(m)}-\bbeta_{(m)}^{*}\right)\right\}^{3}\non\\
&\leq&\frac{\bar{c}_{f}}{3}\sum_{k=1}^{K}
\mathrm{E}_{\boldsymbol{x}}\left\{\sum_{m=1}^{M}w_{m}\left\|\left(1,\boldsymbol{x}_{m}^{T}\right)^{T}\right\|
\left\|\left(\hat{b}_{k(m)}-b_{k(m)}^{*},\left(\hat{\bbeta}_{(m)}-\bbeta_{(m)}^{*}\right)^{T}\right)^{T}\right\|\right\}^{3}\non\\
&\leq&\frac{\bar{c}_{f}}{3}\sum_{k=1}^{K}
\mathrm{E}_{\boldsymbol{x}}\sum_{m=1}^{M}w_{m}\left\{\left\|\left(1,\boldsymbol{x}_{m}^{T}\right)^{T}\right\|
\left\|\left(\hat{b}_{k(m)}-b_{k(m)}^{*},\left(\hat{\bbeta}_{(m)}-\bbeta_{(m)}^{*}\right)^{T}\right)^{T}\right\|\right\}^{3}\non\\
&\leq&\frac{\bar{c}_{f}}{3}
\max_{1\leq m\leq M}
\mathrm{E}_{\boldsymbol{x}}\left\|\left(1,\boldsymbol{x}_{m}^{T}\right)^{T}\right\|^{3}
\max_{1\leq k\leq K}\max_{1\leq m\leq M}
\left\|\left(\hat{b}_{k(m)}-b_{k(m)}^{*},\left(\hat{\bbeta}_{(m)}-\bbeta_{(m)}^{*}\right)^{T}\right)^{T}\right\|\non\\
&&\sum_{k=1}^{K}\max_{1\leq m\leq M}
\left\|\left(\hat{b}_{k(m)}-b_{k(m)}^{*},\left(\hat{\bbeta}_{(m)}-\bbeta_{(m)}^{*}\right)^{T}\right)^{T}\right\|^{2}\non\\
&=&O_{p}\left(n^{-3/2}\bar{k}^{3}K^{2}\left(\log n\right)^{3/2}\right)=o_{p}(1).
\ee
Therefore,
\be
&&\mathrm{CPE}_{n}(\boldsymbol{w})-\mathrm{E}\left\{\sum_{k=1}^{K}\rho_{\tau_{k}}\left(\varepsilon-b_{k}+u_{k}(\boldsymbol{w})\right)\right\}\non\\
&=&\frac{1}{2}\mathrm{E}_{\boldsymbol{x}}\left(\sum_{k=1}^{K}f\left\{b_{k}-u_{k}(\boldsymbol{w})\mid\boldsymbol{x}\right\}
\left[\sum_{m=1}^{M}w_{m}\left\{(1,\boldsymbol{x}_{m}^{T})\left(\hat{b}_{k(m)}-b_{k(m)}^{*},\left(\hat{\bbeta}_{(m)}-\bbeta_{(m)}^{*}\right)^{T}\right)^{T}\right\}\right]^{2}\right)+o_{p}(1).\non
\ee
Noting that
\be
&&\frac{1}{2}\mathrm{E}_{\boldsymbol{x}}\left(\sum_{k=1}^{K}f\left\{b_{k}-u_{k}(\boldsymbol{w})\mid\boldsymbol{x}\right\}
\left[\sum_{m=1}^{M}w_{m}\left\{(1,\boldsymbol{x}_{m}^{T})\left(\hat{b}_{k(m)}-b_{k(m)}^{*},\left(\hat{\bbeta}_{(m)}-\bbeta_{(m)}^{*}\right)^{T}\right)^{T}\right\}\right]^{2}\right)\non\\
&\leq&\frac{1}{2}\mathrm{E}_{\boldsymbol{x}}\left(\sum_{k=1}^{K}f\left\{b_{k}-u_{k}(\boldsymbol{w})\mid\boldsymbol{x}\right\}
\sum_{m=1}^{M}w_{m}\left\{(1,\boldsymbol{x}_{m}^{T})\left(\hat{b}_{k(m)}-b_{k(m)}^{*},\left(\hat{\bbeta}_{(m)}-\bbeta_{(m)}^{*}\right)^{T}\right)^{T}\right\}^{2}\right)\non\\
&=&\frac{1}{2}\left[\sum_{k=1}^{K}\sum_{m=1}^{M}w_{m}\left(\hat{b}_{k(m)}-b_{k(m)}^{*},\left(\hat{\bbeta}_{(m)}-\bbeta_{(m)}^{*}\right)^{T}\right)
\right.\non\\
&&\left.\mathrm{E}\left\{f\left(b_{k}-u_{k}(\boldsymbol{w})\mid\boldsymbol{x}\right)\left(1,\boldsymbol{x}_{m}^{T}\right)\left(1,\boldsymbol{x}_{m}^{T}\right)^{T}\right\}\left(\hat{b}_{k(m)}-b_{k(m)}^{*},\left(\hat{\bbeta}_{(m)}-\bbeta_{(m)}^{*}\right)^{T}\right)^{T}
\right]\non\\
&\leq&\frac{\bar{c}_{A}}{2}\max_{1\leq m\leq M}\sum_{k=1}^{K}\left\|\left(\hat{b}_{k(m)}-b_{k(m)}^{*},\left(\hat{\bbeta}_{(m)}-\bbeta_{(m)}^{*}\right)^{T}\right)^{T}\right\|^{2}\non\\
&=&O_{p}\left(\bar{c}_{A}n^{-1}K\bar{k}\log n\right)=o_{p}(1),
\ee
Then we can concluded that
\be
\mathrm{CPE}_{n}(\boldsymbol{w})-\mathrm{E}\left\{\sum_{k=1}^{K}\rho_{\tau_{k}}\left(\varepsilon-b_{k}+u_{k}(\boldsymbol{w})\right)\right\}
=o_{p}(1).
\ee
Let $D(t)=\mathrm{E}\left\{\sum_{k=1}^{K}\rho_{\tau_{k}}(\varepsilon-b_{k}+t)-\sum_{k=1}^{K}\rho_{\tau_{k}}(\varepsilon-b_{k})\right\}$
where $t\in\mathbb{R}$. It is well known that $D(t)$ has a global minimum at $t=0$. This implies that
$$\min_{\boldsymbol{w}\in\mathcal{W}}\mathrm{E}\left\{\sum_{k=1}^{K}\rho_{\tau_{k}}\left(\varepsilon-b_{k}+u_{k}(\boldsymbol{w})\right)\right\}
\geq\mathrm{E}\left\{\sum_{k=1}^{K}\rho_{\tau_{k}}(\varepsilon-b_{k})\right\}.$$
 Then, we have
\be
\min_{\boldsymbol{w}\in\mathcal{W}}\mathrm{CPE}_{n}(\boldsymbol{w})
&=&\min_{\boldsymbol{w}\in\mathcal{W}}\mathrm{E}\left\{\sum_{k=1}^{K}\rho_{\tau_{k}}\left(\varepsilon-b_{k}+u_{k}(\boldsymbol{w})\right)\right\}
-o_{p}(1)\non\\
&\geq&\mathrm{E}\left\{\sum_{k=1}^{K}\rho_{\tau_{k}}(\varepsilon-b_{k})\right\}-o_{p}(1).\non\ee

Secondly, we establish (ii). Decompose $$CV_{1n}(\boldsymbol{w})=CV_{1n,1}(\boldsymbol{w})-CV_{1n,2}(\boldsymbol{w}),$$
where $$CV_{1n,1}(\boldsymbol{w})=\frac{1}{n}\sum_{i=1}^{n}\sum_{k=1}^{K}\left\{\mu_{i}-\sum_{m=1}^{M}w_{m}\left(1,\boldsymbol{x}_{i(m)}^{T}\right)\left(b_{k(m)}^{*},\bbeta_{(m)}^{*T}\right)^{T}\right\}\psi_{\tau_{k}}(\varepsilon_{i}-b_{k})$$
and $$CV_{1n,2}(\boldsymbol{w})=\frac{1}{n}\sum_{i=1}^{n}\sum_{k=1}^{K}\left\{\sum_{m=1}^{M}w_{m}
\left(1,\boldsymbol{x}_{i(m)}^{T}\right)\left(\hat{b}_{ki(m)}-b_{k(m)}^{*},\left(\hat{\bbeta}_{i(m)}-\bbeta_{(m)}^{*}\right)^{T}\right)^{T}\right\}\psi_{\tau_{k}}(\varepsilon_{i}-b_{k}).$$
To prove (ii), it suffices to show
$\mathrm{sup}_{\boldsymbol{w}\in\mathcal{W}}|CV_{1n,1}(\boldsymbol{w})|=o_{p}(1)$ and $\mathrm{sup}_{\boldsymbol{w}\in\mathcal{W}}|CV_{1n,2}(\boldsymbol{w})|=o_{p}(1).$
When both
$M$ and
$\bar{k}=\mathrm{max}_{1\leq m\leq M}\tilde{k}_{m}$ are finite, based on the proof of Theorem 3.3
in \cite{lu2015jackknife}, we can apply the Glivenko-Cantelli Theorem \citep[e.g., Theorem 2.4.1 in][]{van1996weak} to show that $\mathrm{sup}_{\boldsymbol{w}\in\mathcal{W}}|CV_{1n,1}(\boldsymbol{w})|=o_{p}(1)$.
Following \cite{lu2015jackknife}, let us define the class of functions
$$\mathcal{G}=\{g(\cdot,\cdot;\boldsymbol{w}): \boldsymbol{w}\in\mathcal{W}\},$$
where $g(\cdot,\cdot;\boldsymbol{w}):\mathbb{R}\times\mathbb{R}^{d_{\boldsymbol{x}}}\rightarrow\mathbb{R}$ is $$g(\varepsilon_{i},\boldsymbol{x}_{i};\boldsymbol{w})=\sum_{k=1}^{K}\left\{\mu_{i}-\sum_{m=1}^{M}w_{m}\left(1,\boldsymbol{x}_{i(m)}^{T}\right)\left(b_{k(m)}^{*},\bbeta_{(m)}^{*T}\right)^{T}\right\}\psi_{\tau_{k}}(\varepsilon_{i}-b_{k}).$$
Define the metric $|\cdot|_{1}$ on $\mathcal{W}$, where $|\boldsymbol{w}-\overline{\boldsymbol{w}}|_{1}=\sum_{m=1}^{M}|w_{m}-\overline{w}_{m}|,$ for any $\boldsymbol{w}=(w_{1},\ldots,w_{M})\in\mathcal{W}$ and $\overline{\boldsymbol{w}}=(\overline{w}_{1},\ldots,\overline{w}_{M})\in\mathcal{W}$.
The $\delta$-covering number
of $\mathcal{W}$ in composite quantile model averaging is $\mathcal{N}(\delta,\mathcal{W},|\cdot|_{1})=O\left(1/\delta^{M-1}\right)$.
Note that $|g(\varepsilon_{i},\boldsymbol{x}_{i};\boldsymbol{w})-g(\varepsilon_{i},\boldsymbol{x}_{i};\overline{\boldsymbol{w}})|\leq
c_{\Theta}|\boldsymbol{w}-\overline{\boldsymbol{w}}|_{1}\max_{1\leq m\leq M}\left\|\left(1,\boldsymbol{x}_{i(m)}^{T}\right)^{T}\right\|$, where $c_{\Theta}=K\max_{1\leq m\leq M}\left\|\left(b_{k(m)}^{*},\bbeta_{(m)}^{*T}\right)^{T}\right\|=O\left(K\bar{k}^{1/2}\right)$, and that $\mathrm{E}\max_{1\leq m\leq M}\left\|\left(1,\boldsymbol{x}_{i(m)}^{T}\right)^{T}\right\|<\infty$ in the case of finite $M$ and $\bar{k}$, implying that the  $\varepsilon$-bracketing number  of $\mathcal{G}$ with respect to the $L_{1}(P)$-norm is given by
$\mathcal{N}_{[~]}\left(\delta,\mathcal{G},L_{1}(P)\right)\leq C/\delta^{M-1}$ for some finite $C$. Thus, by applying Theorem 2.4.1 in \cite{van1996weak}, we can conclude that $\mathcal{G}$ is Glivenko-Cantelli.

When either $M\rightarrow\infty$ or $\bar{k}\rightarrow\infty$ as
$n\rightarrow\infty$,
Let $h_{n}=1/\left\{K\left(\overline{k}+1\right)\log n\right\}$ and create grids using regions of the form $W_{j}=\{\boldsymbol{w}:|\boldsymbol{w}-\boldsymbol{w}_{j}|_{1}\leq h_{n}\},$ where $\boldsymbol{w}_{j}=(w_{j1},\cdots,w_{jM})$. $\mathcal{W}$ can be covered with
$N=O\left(1/h_{n}^{M-1}\right)$ regions $W_{j},\ j=1,\ldots,N$.
Observe that the following result holds uniformly in $j$,
\be
&&\mathrm{sup}_{\boldsymbol{w}\in\mathcal{W}_{j}}|CV_{1n,1}(\boldsymbol{w})-CV_{1n,1}(\boldsymbol{w}_{j})|\non\\
&=&\mathrm{sup}_{\boldsymbol{w}\in\mathcal{W}_{j}}\left|\frac{1}{n}\sum_{i=1}^{n}\sum_{k=1}^{K}\left\{\sum_{m=1}^{M}(w_{m}-w_{jm})\left(1,\boldsymbol{x}_{i(m)}^{T}\right)^{T}\left(b_{k(m)}^{*},\bbeta_{(m)}^{*T}\right)^{T}\right\}\psi_{\tau_{k}}(\varepsilon_{i}-b_{k})\right|\non\\
&\leq&\mathrm{sup}_{\boldsymbol{w}\in\mathcal{W}_{j}}\left|\frac{1}{n}\sum_{i=1}^{n}\sum_{k=1}^{K}\left\{\sum_{m=1}^{M}|w_{m}-w_{jm}|\left\|\left(1,\boldsymbol{x}_{i(m)}^{T}\right)^{T}\right\|\left\|\left(b_{k(m)}^{*},\bbeta_{(m)}^{*T}\right)^{T}\right\|\right\}\right|\non\\
&\leq&Kc_{\bbeta}\max_{1\leq m\leq M}
\frac{1}{n}\sum_{i=1}^{n}\left\|\left(1,\boldsymbol{x}_{i(m)}^{T}\right)^{T}\right\|
\mathrm{sup}_{\boldsymbol{w}\in\mathcal{W}_{j}}\sum_{m=1}^{M}|w_{m}-w_{jm}|\non\\
&\leq&Kc_{\bbeta}O_{p}\left(\overline{k}^{1/2}\right)h_{n}=o_{p}(1),
\ee
where $c_{\bbeta}=\max_{1\leq m\leq M,1\leq k\leq K}\left\|\left(b_{k(m)}^{*},\bbeta_{(m)}^{*T}\right)^{T}\right\|=O\left(\overline{k}^{1/2}\right)$ and
here we have used the fact that
\be
\max_{1\leq m\leq M}\frac{1}{n}\sum_{i=1}^{n}\left\|\left(1,\boldsymbol{x}_{i(m)}^{T}\right)^{T}\right\|
&\leq&\max_{1\leq m\leq M}\frac{1}{n}\sum_{i=1}^{n}\mathrm{E}\left\|\left(1,\boldsymbol{x}_{i(m)}^{T}\right)^{T}\right\|\non\\
&&+\max_{1\leq m\leq M}\left|\frac{1}{n}\sum_{i=1}^{n}\left\{\left\|\left(1,\boldsymbol{x}_{i(m)}^{T}\right)^{T}\right\|-\mathrm{E}\left\|\left(1,\boldsymbol{x}_{i(m)}^{T}\right)^{T}\right\|\right\}\right|\non\\
&=&O\left(\bar{k}^{1/2}\right)+o_{p}(1)=O_{p}\left(\bar{k}^{1/2}\right).
\ee
Therefore,
\be
\mathrm{sup}_{\boldsymbol{w}\in\mathcal{W}}\left|CV_{1n,1}(\boldsymbol{w})\right|
&=&\max_{1\leq j\leq N}\mathrm{sup}_{\boldsymbol{w}\in W_{j}}\left|CV_{1n,1}(\boldsymbol{w})\right|\non\\
&\leq&\max_{1\leq j\leq N}\left|CV_{1n,1}(\boldsymbol{w}_{j})\right|+\max_{1\leq j\leq N}\mathrm{sup}_{\boldsymbol{w}\in W_{j}}\left|CV_{1n,1}(\boldsymbol{w})-CV_{1n,1}(\boldsymbol{w}_{j})\right|\non\\
&=&\max_{1\leq j\leq N}\left|CV_{1n,1}(\boldsymbol{w}_{j})\right|+o_{p}(1)
\ee
Let $u_{ki}(\boldsymbol{w}_{j})=\mu_{i}+b_{k}-\sum_{m=1}^{M}w_{jm}\left(1,\boldsymbol{x}_{i(m)}^{T}\right)\left(b_{k(m)}^{*},\bbeta_{(m)}^{*T}\right)^{T}$ and $e_{n}=\left(nK^{2}\bar{k}\right)^{1/2}$. Noting that $$u_{ki}(\boldsymbol{w}_{j})=\sum_{m=1}^{M}w_{jm}\left\{\mu_{i}+b_{k}-\left(1,\boldsymbol{x}_{i(m)}^{T}\right)\left(b_{k(m)}^{*},\bbeta_{(m)}^{*T}\right)^{T}\right\}
\leq\max_{1\leq m\leq M}\left|u_{ki(m)}\right|,$$
where $u_{ki(m)}=\mu_{i}+b_{k}-\left(1,\boldsymbol{x}_{i(m)}^{T}\right)\left(b_{k(m)}^{*},\bbeta_{(m)}^{*T}\right)^{T}$.
Let $U_{i}=\sum_{k=1}^{K}\max_{1\leq m\leq M}\left|u_{ki(m)}\right|$,
then for any $\epsilon>0$,
\be
 &&\mathrm{P}\left(\max_{1\leq j\leq N}\left|CV_{1n,1}(\boldsymbol{w}_{j})\right|\geq 2\epsilon\right)\non\\
 &=&\mathrm{P}\left(\max_{1\leq j\leq N}\left|\frac{1}{n}\sum_{i=1}^{n}\sum_{k=1}^{K}u_{ki}(\boldsymbol{w}_{j})\psi_{\tau_{k}}(\varepsilon_{i}-b_{k})\right|\geq 2\epsilon\right)\non\\
 &\leq&\mathrm{P}\left(\left|\frac{1}{n}\sum_{i=1}^{n}\sum_{k=1}^{K}\max_{1\leq m\leq M}\left|u_{ki(m)}\right|\psi_{\tau_{k}}(\varepsilon_{i}-b_{k})\right|\geq 2\epsilon\right)\non\\
 &\leq& P\left(\left|\frac{1}{n}\sum_{i=1}^{n}\sum_{k=1}^{K}\max_{1\leq m\leq M}\left|u_{ki(m)}\right|\psi_{\tau_{k}}(\varepsilon_{i}-b_{k})\cdot \boldsymbol{1}\left\{U_{i}\leq e_{n}\right\}\right|\geq\epsilon\right)\non\\
 &&~~~~~~+P\left(\left|\frac{1}{n}\sum_{i=1}^{n}\sum_{k=1}^{K}\max_{1\leq m\leq M}\left|u_{ki(m)}\right|\psi_{\tau_{k}}(\varepsilon_{i}-b_{k})\cdot \boldsymbol{1}\left\{U_{i}\geq e_{n}\right\}\right| \geq\epsilon\right)\non\\
 &=&T_{11}+T_{12},~\text{say}.
\ee

We prove $\max_{1\leq j\leq N}|CV_{1n,1}(\boldsymbol{w}_{j})|=o_{p}(1)$ by showing that $T_{11}=o(1)$ and $T_{12}=o(1)$.
Noting that
\be
&&\mathrm{E}\left[\sum_{k=1}^{K}\max_{1\leq m\leq M}\left|u_{ki(m)}\right|\boldsymbol{1}\left\{\left|U_{i}\right|\leq e_{n}\right\}\psi_{\tau_{k}}(\varepsilon_{i}-b_{k})\right]\non\\
&=&\mathrm{E}\left[\mathrm{E}\left[\sum_{k=1}^{K}\max_{1\leq m\leq M}\left|u_{ki(m)}\right|\boldsymbol{1}\left\{\left|U_{i}\right|\leq e_{n}\right\}\psi_{\tau_{k}}(\varepsilon_{i}-b_{k})\big|\boldsymbol{x}_{i}\right]\right]\non\\
&=&\mathrm{E}\left[\sum_{k=1}^{K}\max_{1\leq m\leq M}\left|u_{ki(m)}\right|\boldsymbol{1}\left\{\left|U_{i}\right|\leq e_{n}\right\}\mathrm{E}\left\{\psi_{\tau_{k}}(\varepsilon_{i}-b_{k})\big|\boldsymbol{x}_{i}\right\}\right]\non\\
&=&0,
\ee
where the last equality follows the fact that $\mathrm{E}\left\{\psi_{\tau_{k}}(\varepsilon_{i}-b_{k})\big|\boldsymbol{x}_{i}\right\}=0$.
Then \be
&&\mathrm{Var}\left[\sum_{k=1}^{K}\max_{1\leq m\leq M}\left|u_{ki(m)}\right|\psi_{\tau_{k}}(\varepsilon_{i}-b_{k})\cdot \boldsymbol{1}\left\{U_{i}\leq e_{n}\right\}\right]\non\\
&=&\mathrm{E}\left[\sum_{k=1}^{K}\max_{1\leq m\leq M}\left|u_{ki(m)}\right|\psi_{\tau_{k}}(\varepsilon_{i}-b_{k})\cdot \boldsymbol{1}\left\{U_{i}\leq e_{n}\right\}\right]^{2}\non\\
&\leq&\mathrm{E}\left[\sum_{k=1}^{K}\max_{1\leq m\leq M}\left|u_{ki(m)}\right|\psi_{\tau_{k}}(\varepsilon_{i}-b_{k})\right]^{2}\non\\
&\leq&\mathrm{E}\left(\sum_{k=1}^{K}\max_{1\leq m\leq M}\left|u_{ki(m)}\right|\right)^{2}\non\\
&=&\mathrm{E}\left\{\sum_{k=1}^{K}\max_{1\leq m\leq M}\left|\mu_{i}+b_{k}-\left(1,\boldsymbol{x}_{i(m)}^{T}\right)\left(b_{k(m)}^{*},\bbeta_{(m)}^{*T}\right)^{T}\right|\right\}^{2}\non\\
&\leq&K^{2}\bar{k}\bar{\sigma}^{2}\text{~for~some~}\bar{\sigma}^{2}<\infty,
\ee
where $\bar{k}=\max_{1\leq m\leq M}\tilde{k}_{m}$.
Therefore, by Bernstein's inequalities in \cite{serfling1980approximation}, we have
\be\label{T11}
T_{11}
&\leq&2\exp\left\{-\frac{n^{2}\varepsilon^{2}}{2nK^{2}\bar{k}\bar{\sigma}^{2}+2n\varepsilon e_{n}/3}\right\}\non\\
&=&2\exp\left\{-\frac{n\varepsilon^{2}}{2K^{2}\bar{k}\bar{\sigma}^{2}+2\varepsilon e_{n}/3}\right\}\non\\
&=&2\exp\left\{-\frac{n\varepsilon^{2}}{2K^{2}\bar{k}\bar{\sigma}^{2}+2\varepsilon (nK^{2}\bar{k})^{1/2}/3}\right\}=o(1).
\ee
In fact,
by noting that $M K^{2}\bar{k}^{2}\log n/n\rightarrow0$, we have
$n/(K^{2}\bar{k})\rightarrow \infty $, and thus $K^{2}\bar{k}/(nK^{2}\bar{k})^{1/2}\rightarrow0 $. Hence we have $n\epsilon^{2}/\{2K^{2}\bar{k}\bar{\sigma}^{2}+2\epsilon (nK^{2}\bar{k})^{1/2}/3\}\rightarrow \infty$ and thus $T_{11}=o(1)$.

Noting that $\mathrm{E}\left(\left|U_{i}K^{-1}\bar{k}^{-1/2}\right|^{4}\right)=O(1)$
and $\frac{nK^{2}\bar{k}}{e_{n}^{2}}=O(1)$, by Boole's, Markov's and Cauchy-Schwarz inequalities,
\be
T_{12}&=&\mathrm{P}\left(\left|\frac{1}{n}\sum_{i=1}^{n}\sum_{k=1}^{K}\max_{1\leq m\leq M}\left|u_{ki(m)}\right|\psi_{\tau_{k}}(\varepsilon_{i}-b_{k})\cdot \boldsymbol{1}\left\{U_{i}\geq e_{n}\right\}\right| \geq\varepsilon\right)\non\\
&\leq&\mathrm{P}\left(\frac{1}{n}\sum_{i=1}^{n}\sum_{k=1}^{K}\max_{1\leq m\leq M}\left|u_{ki(m)}\right|\cdot \boldsymbol{1}\left\{U_{i}\geq e_{n}\right\} \geq\varepsilon\right)\non\\
&\leq&P\left(\max_{1\leq i\leq n}\sum_{k=1}^{K}\max_{1\leq m\leq M}\left|u_{ki(m)}\right|\cdot \boldsymbol{1}\left\{U_{i}\geq e_{n}\right\}\geq \varepsilon\right)\non\\
&\leq&\sum_{i=1}^{n}\mathrm{P}\left(\sum_{k=1}^{K}\max_{1\leq m\leq M}\left|u_{ki(m)}\right|\cdot \boldsymbol{1}\left\{U_{i}\geq e_{n}\right\}\geq \varepsilon\right)\non\\
&\leq&\sum_{i=1}^{n}\mathrm{P}\left(U_{i}K^{-1}\bar{k}^{-1/2}\geq K^{-1}\bar{k}^{-1/2}e_{n}\right)\non\\
&\leq&\frac{nK^{2}\bar{k}}{e_{n}^{2}}
\mathrm{E}\left[\left|U_{i}K^{-1}\bar{k}^{-1/2}\right|^{2}\boldsymbol{1}\left\{U_{i}K^{-1}\bar{k}^{-1/2}\geq K^{-1}\bar{k}^{-1/2}e_{n}\right\}\right]\non\\
&\leq&\frac{nK^{2}\bar{k}}{e_{n}^{2}}\mathrm{E}\left(\left|U_{i}K^{-1}\bar{k}^{-1/2}\right|^{4}\right)^{1/2}\mathrm{E}\left[\boldsymbol{1}\left\{U_{i}K^{-1}\bar{k}^{-1/2}\geq K^{-1}\bar{k}^{-1/2}e_{n}\right\}\right]^{1/2}\non\\
&=&\frac{nK^{2}\bar{k}}{e_{n}^{2}}\mathrm{E}\left(\left|U_{i}K^{-1}\bar{k}^{-1/2}\right|^{4}\right)^{1/2}
\mathrm{P}\left(U_{i}K^{-1}\bar{k}^{-1/2}\geq K^{-1}\bar{k}^{-1/2}e_{n}\right)^{1/2}\non\\
&\leq&\frac{nK^{2}\bar{k}}{e_{n}^{2}}\mathrm{E}\left(\left|U_{i}K^{-1}\bar{k}^{-1/2}\right|^{4}\right)^{1/2}
\left\{\frac{\mathrm{E}\left(\left|U_{i}K^{-1}\bar{k}^{-1/2}\right|^{4}\right)}{(K^{-1}\bar{k}^{-1/2}e_{n})^{4}}
\right\}^{1/2}.\non
\ee
It follows that $\max_{1\leq j\leq N}\left|CV_{1n,1}(\boldsymbol{w}_{j})\right|=o_{p}(1)$ and thus $\mathrm{sup}_{\boldsymbol{w}\in\mathcal{W}}|CV_{1n,1}(\boldsymbol{w})|=o_{p}(1)$.

By the absolute value inequality, Lemma \ref{lem:order} and assumption $M K^{2}\bar{k}^{2}\log n/n\rightarrow0$, we have
\be
\sup_{\boldsymbol{w}\in\mathcal{W}}|CV_{1n,2}(\boldsymbol{w})|&\leq &
\sup_{\boldsymbol{w}\in\mathcal{W}}\frac{1}{n}\sum_{i=1}^{n}\sum_{k=1}^{K}
\sum_{m=1}^{M}w_{m}\left|\left(1,\boldsymbol{x}_{i(m)}^{T}\right)\left(\hat{b}_{ki(m)}-b_{k(m)}^{*},\left(\hat{\bbeta}_{i(m)}-\bbeta_{(m)}^{*}\right)^{T}\right)^{T}\psi_{\tau_{k}}(\varepsilon_{i}-b_{k})\right|\non\\
&\leq& \max_{1\leq i\leq n}\max_{1\leq m\leq M}\sum_{k=1}^{K}\left\|\left(\hat{b}_{ki(m)}-b_{k(m)}^{*},\left(\hat{\bbeta}_{i(m)}-\bbeta_{(m)}^{*}\right)^{T}\right)^{T}\right\|
\max_{1\leq m\leq M}\frac{1}{n}\sum_{i=1}^{n}\left\|\left(1,\boldsymbol{x}_{i(m)}^{T}\right)^{T}\right\|\non\\
&\leq&O_{p}\left(\sqrt{n^{-1}K\bar{k}\log n}\right)O_{p}\left(\overline{k}^{1/2}\right)=o_{p}(1).
\ee
Consequently $\sup_{\boldsymbol{w}\in\mathcal{W}}|CV_{1n,2}(\boldsymbol{w})|=o_{p}(1).$

Thirdly, we establish (iii). Decompose  $$CV_{2n}(\boldsymbol{w})=CV_{2n,1}(\boldsymbol{w})+CV_{2n,2}(\boldsymbol{w}),$$
where \be
CV_{2n,1}(\boldsymbol{w})&=&\frac{1}{n}\sum_{i=1}^{n}\sum_{k=1}^{K}\int_{0}^{\sum_{m=1}^{M}w_{m}\hat{b}_{k(m)}^{*}-b_{k}+\sum_{m=1}^{M}w_{m}\boldsymbol{x}_{i(m)}^{T}\hat{\bbeta}_{(m)}^{*}-\mu_{i}}\left[\boldsymbol{1}\{\varepsilon_{i}-b_{k}\leq s\}-\boldsymbol{1}\{\varepsilon_{i}-b_{k}\leq 0\}\right.\non\\
&&\left.-F(s+b_{k}|\boldsymbol{x}_{i})+F(b_{k}|\boldsymbol{x}_{i})\right]\mathrm{d}s \non
\ee
   and
\be
CV_{2n,2}(\boldsymbol{w})&=&\frac{1}{n}\sum_{i=1}^{n}\sum_{k=1}^{K}\int_{\sum_{m=1}^{M}w_{m}\hat{b}_{k(m)}^{*}-b_{k}+\sum_{m=1}^{M}w_{m}\boldsymbol{x}_{i(m)}^{T}\hat{\bbeta}_{(m)}^{*}-\mu_{i}}^{\sum_{m=1}^{M}w_{m}\hat{b}_{ki(m)}-b_{k}+\sum_{m=1}^{M}w_{m}\boldsymbol{x}_{i(m)}^{T}\hat{\bbeta}_{i(m)}-\mu_{i}}\left[\boldsymbol{1}\{\varepsilon_{i}-b_{k}\leq s\}-\boldsymbol{1}\{\varepsilon_{i}-b_{k}\leq 0\}\right.\non\\
&&\left.-F(s+b_{k}|\boldsymbol{x}_{i})+F(b_{k}|\boldsymbol{x}_{i})\right]\mathrm{d}s.\non\ee
We need to show that $\sup_{\boldsymbol{w}\in\mathcal{W}}|CV_{2n,1}(\boldsymbol{w})|=o_{p}(1)$ and
$\sup_{\boldsymbol{w}\in\mathcal{W}}|CV_{2n,2}(\boldsymbol{w})|=o_{p}(1).$
The proof of $\sup_{\boldsymbol{w}\in\mathcal{W}}|CV_{2n,1}(\boldsymbol{w})|=o_{p}(1)$ is analogous to that of $\sup_{\boldsymbol{w}\in\mathcal{W}}|CV_{1n,1}(\boldsymbol{w})|=o_{p}(1)$ when both $M$ and $\bar{k}=\mathrm{max}_{1\leq m\leq M}\tilde{k}_{m}$ are finite.  The proof is thus omitted for brevity. When either $M\rightarrow\infty$ or $\bar{k}\rightarrow\infty$ as $n\rightarrow\infty$, we also obtain
$$\mathrm{sup}_{\boldsymbol{w}\in\mathcal{W}}|CV_{2n,1}(\boldsymbol{w})|\leq \max_{1\leq j\leq N}|CV_{2n,1}(\boldsymbol{w}_{j})|+o_{p}(1).$$ Hence we only need to prove that $\max_{1\leq j\leq N}|CV_{2n,1}(\boldsymbol{w}_{j})|=o_{p}(1)$.

Let $e_{n}=(nM^{2}\bar{k})^{1/2}$. Then for any $\epsilon>0$,
\be
 &&\mathrm{P}\left(\max_{1\leq j\leq N}\left|CV_{2n,1}(\boldsymbol{w}_{j})\right|\geq 2\epsilon\right)\non\\
 &=&\mathrm{P}\left(\max_{1\leq j\leq N}\left|\frac{1}{n}\sum_{i=1}^{n}\sum_{k=1}^{K}\int_{0}^{\sum_{m=1}^{M}w_{m}\hat{b}_{k(m)}^{*}-b_{k}+\sum_{m=1}^{M}w_{m}\boldsymbol{x}_{i(m)}^{T}\hat{\bbeta}_{(m)}^{*}-\mu_{i}}\left[\boldsymbol{1}\{\varepsilon_{i}-b_{k}\leq s\}-\boldsymbol{1}\{\varepsilon_{i}-b_{k}\leq 0\}\right.\right.\right.\non\\
&&\left.\left.\left.-F(s+b_{k}|\boldsymbol{x}_{i})+F(b_{k}|\boldsymbol{x}_{i})\right]\mathrm{d}s\right|\geq 2\epsilon\right)\non\\
   &\leq&N\cdot\mathrm{P}\left(\left|\sum_{i=1}^{n}\sum_{k=1}^{K}\int_{0}^{\sum_{m=1}^{M}w_{m}\hat{b}_{k(m)}^{*}-b_{k}+\sum_{m=1}^{M}w_{m}\boldsymbol{x}_{i(m)}^{T}\hat{\bbeta}_{(m)}^{*}-\mu_{i}}\left[\boldsymbol{1}\{\varepsilon_{i}-b_{k}\leq s\}-\boldsymbol{1}\{\varepsilon_{i}-b_{k}\leq 0\}\right.\right.\right.\non\\
&&\left.\left.\left.-F(s+b_{k}|\boldsymbol{x}_{i})+F(b_{k}|\boldsymbol{x}_{i})\right]\mathrm{d}s\cdot\boldsymbol{1}\left\{\left|U_{i}\right|\leq e_{n}\right\}\right|\geq n\epsilon\right)\non\\
   &&+\mathrm{P}\left(\max_{1\leq j\leq N}\left|\frac{1}{n}\sum_{i=1}^{n}\sum_{k=1}^{K}\int_{0}^{\sum_{m=1}^{M}w_{m}\hat{b}_{k(m)}^{*}-b_{k}+\sum_{m=1}^{M}w_{m}\boldsymbol{x}_{i(m)}^{T}\hat{\bbeta}_{(m)}^{*}-\mu_{i}}\left[\boldsymbol{1}\{\varepsilon_{i}-b_{k}\leq s\}-\boldsymbol{1}\{\varepsilon_{i}-b_{k}\leq 0\}\right.\right.\right.\non\\
&&\left.\left.\left.-F(s+b_{k}|\boldsymbol{x}_{i})+F(b_{k}|\boldsymbol{x}_{i})\right]\mathrm{d}s\cdot\boldsymbol{1}\left\{\left|U_{i}\right|\geq e_{n}\right\}\right|\geq \epsilon\right)\non\\
 &=&T_{21}+T_{22},\mathrm{say}.
\ee
We will prove $\max_{1\leq j\leq N}|CV_{2n,1}(\boldsymbol{w}_{j})|=o_{p}(1)$ by showing that $T_{21}=o(1)$ and $T_{22}=o(1)$.

Note that
\be
&&\mathrm{E}\left(\sum_{k=1}^{K}\int_{0}^{\sum_{m=1}^{M}w_{m}\hat{b}_{k(m)}^{*}-b_{k}+\sum_{m=1}^{M}w_{m}\boldsymbol{x}_{i(m)}^{T}\hat{\bbeta}_{(m)}^{*}-\mu_{i}}\left[\boldsymbol{1}\{\varepsilon_{i}-b_{k}\leq s\}-\boldsymbol{1}\{\varepsilon_{i}-b_{k}\leq 0\}\right.\right.\non\\
&&\left.\left.-F(s+b_{k}|\boldsymbol{x}_{i})+F(b_{k}|\boldsymbol{x}_{i})\right]\mathrm{d}s\cdot\boldsymbol{1}\left\{\left|U_{i}\right|\leq e_{n}\right\}\right)\non\\
&=&\mathrm{E}\left\{\mathrm{E}\left(\sum_{k=1}^{K}\int_{0}^{\sum_{m=1}^{M}w_{m}\hat{b}_{k(m)}^{*}-b_{k}+\sum_{m=1}^{M}w_{m}\boldsymbol{x}_{i(m)}^{T}\hat{\bbeta}_{(m)}^{*}-\mu_{i}}\left[\boldsymbol{1}\{\varepsilon_{i}-b_{k}\leq s\}-\boldsymbol{1}\{\varepsilon_{i}-b_{k}\leq 0\}\right.\right.\right.\non\\
&&\left.\left.\left.-F(s+b_{k}|\boldsymbol{x}_{i})+F(b_{k}|\boldsymbol{x}_{i})\right]\mathrm{d}s\cdot\boldsymbol{1}\left\{\left|U_{i}\right|\leq e_{n}\right\}|\boldsymbol{x}_{i}\right)\right\}\non\\
&=&\mathrm{E}\left\{\left(\sum_{k=1}^{K}\int_{0}^{\sum_{m=1}^{M}w_{m}\hat{b}_{k(m)}^{*}-b_{k}+\sum_{m=1}^{M}w_{m}\boldsymbol{x}_{i(m)}^{T}\hat{\bbeta}_{(m)}^{*}-\mu_{i}}\mathrm{E}\left[\boldsymbol{1}\{\varepsilon_{i}-b_{k}\leq s\}-\boldsymbol{1}\{\varepsilon_{i}-b_{k}\leq 0\}\right.\right.\right.\non\\
&&\left.\left.\left.-F(s+b_{k}|\boldsymbol{x}_{i})+F(b_{k}|\boldsymbol{x}_{i})|\boldsymbol{x}_{i}\right]\mathrm{d}s\cdot\boldsymbol{1}\left\{\left|U_{i}\right|\leq e_{n}\right\}\right)\right\}\non\\
   &=&0,
\ee
where the last equality follows from the fact that $$\mathrm{E}\left[\boldsymbol{1}\{\varepsilon_{i}-b_{k}\leq s\}-\boldsymbol{1}\{\varepsilon_{i}-b_{k}\leq 0\}-F(s+b_{k}|\boldsymbol{x}_{i})+F(b_{k}|\boldsymbol{x}_{i})|\boldsymbol{x}_{i}\right]=0.$$
Recognizing that $\boldsymbol{1}\{\varepsilon_{i}-b_{k}\leq s\}-\boldsymbol{1}\{\varepsilon_{i}-b_{k}\leq 0\}-F(s+b_{k}|\boldsymbol{x}_{i})+F(b_{k}|\boldsymbol{x}_{i})\leq2$, \be
&&\int_{0}^{\sum_{m=1}^{M}w_{m}\hat{b}_{k(m)}^{*}-b_{k}+\sum_{m=1}^{M}w_{m}\boldsymbol{x}_{i(m)}^{T}\hat{\bbeta}_{(m)}^{*}-\mu_{i}}\left[\boldsymbol{1}\{\varepsilon_{i}-b_{k}\leq s\}-\boldsymbol{1}\{\varepsilon_{i}-b_{k}\leq 0\}\right.\non\\
&&\left.-F(s+b_{k}|\boldsymbol{x}_{i})+F(b_{k}|\boldsymbol{x}_{i})\right]\mathrm{d}s\non\\
&\leq&2\max_{1\leq m\leq M}\left|u_{ki(m)}\right|.\ee
Then \be
&&\mathrm{Var}\left(\sum_{k=1}^{K}\int_{0}^{\sum_{m=1}^{M}w_{m}\hat{b}_{k(m)}^{*}-b_{k}+\sum_{m=1}^{M}w_{m}\boldsymbol{x}_{i(m)}^{T}\hat{\bbeta}_{(m)}^{*}-\mu_{i}}\left[\boldsymbol{1}\{\varepsilon_{i}-b_{k}\leq s\}-\boldsymbol{1}\{\varepsilon_{i}-b_{k}\leq 0\}\right.\right.\non\\
&&\left.\left.-F(s+b_{k}|\boldsymbol{x}_{i})+F(b_{k}|\boldsymbol{x}_{i})\right]\mathrm{d}s\cdot\boldsymbol{1}\left\{\left|U_{i}\right|\leq e_{n}\right\}\right)\non\\
&=&\mathrm{E}\left(\sum_{k=1}^{K}\int_{0}^{\sum_{m=1}^{M}w_{m}\hat{b}_{k(m)}^{*}-b_{k}+\sum_{m=1}^{M}w_{m}\boldsymbol{x}_{i(m)}^{T}\hat{\bbeta}_{(m)}^{*}-\mu_{i}}\left[\boldsymbol{1}\{\varepsilon_{i}-b_{k}\leq s\}-\boldsymbol{1}\{\varepsilon_{i}-b_{k}\leq 0\}\right.\right.\non\\
&&\left.\left.-F(s+b_{k}|\boldsymbol{x}_{i})+F(b_{k}|\boldsymbol{x}_{i})\right]\mathrm{d}s\cdot\boldsymbol{1}\left\{\left|U_{i}\right|\leq e_{n}\right\}\right)^{2}\non\\
&\leq&4\mathrm{E}\left(U_{i}\right)^{2}\leq K^{2}\bar{k}\bar{\sigma}^{2}.
\ee
Therefore, by Bernstein's inequalities \citep[][p.95]{serfling1980approximation}, we have
\be\label{T21}
T_{21}
&\leq&2N\exp\left\{-\frac{n^{2}\epsilon^{2}}{2nK^{2}\bar{k}\bar{\sigma}^{2}+2n\epsilon e_{n}/3}\right\}\non\\
&=&2N\exp\left\{-\frac{n\epsilon^{2}}{2K^{2}\bar{k}\bar{\sigma}^{2}+2\epsilon e_{n}/3}\right\}\non\\
&=&\exp\left\{-\left\{\frac{n\epsilon^{2}}{2K^{2}\bar{k}\bar{\sigma}^{2}+2\epsilon (nK^{2}\bar{k})^{1/2}/3}-M\log(\bar{k}\log n)\right\}\right\}=o(1).
\ee
In fact, by the proof of (\ref{T11}), we have $n\epsilon^{2}/(2K^{2}\bar{k}\bar{\sigma}^{2}+2\epsilon (nK^{2}\bar{k})^{1/2}/3)\rightarrow \infty$.
$n/(K^{2}\bar{k})\gg M\log(\bar{k}\log n)$ under assumption $MK^{2}\bar{k}^{2}\log n/n\rightarrow0$ and thus $T_{21}=o(1)$.

Noting that $\mathrm{E}\left(\left|U_{i}K^{-1}\bar{k}^{-1/2}\right|^{4}\right)=O(1)$
and $\frac{nK^{2}\bar{k}}{e_{n}^{2}}=O(1)$, by Boole's, Markov's and Cauchy-Schwarz inequalities,
\be
T_{22}&=&\mathrm{P}\left(\max_{1\leq j\leq N}\left|\frac{1}{n}\sum_{i=1}^{n}\sum_{k=1}^{K}\int_{0}^{\sum_{m=1}^{M}w_{m}\hat{b}_{k(m)}^{*}-b_{k}+\sum_{m=1}^{M}w_{m}\boldsymbol{x}_{i(m)}^{T}\hat{\bbeta}_{(m)}^{*}-\mu_{i}}\left[\boldsymbol{1}\{\varepsilon_{i}-b_{k}\leq s\}\right.\right.\right.\non\\
&&\left.\left.\left.-\boldsymbol{1}\{\varepsilon_{i}-b_{k}\leq 0\}-F(s+b_{k}|\boldsymbol{x}_{i})+F(b_{k}|\boldsymbol{x}_{i})\right]\mathrm{d}s\cdot\boldsymbol{1}\left\{\left|U_{i}\right|\geq e_{n}\right\}\right|\geq \epsilon\right)\non\\
&\leq&\mathrm{P}\left(\left|\frac{1}{n}\sum_{i=1}^{n}2U_{i}\cdot \boldsymbol{1}\left\{\left|U_{i}\right|\geq e_{n}\right\}\right|\geq \epsilon\right)\non\\
&\leq&\mathrm{P}\left(\max_{1\leq i\leq n}\left|U_{i}\cdot \boldsymbol{1}\left\{\left|U_{i}\right|\geq e_{n}\right\}\right|\geq \epsilon/2\right)\non\\
&\leq&\sum_{i=1}^{n}\mathrm{P}\left(\left|U_{i}K^{-1}\bar{k}^{-1/2}\right|\geq K^{-1}\bar{k}^{-1/2}e_{n}\right)\non\\
&\leq&\frac{nK^{2}\bar{k}}{e_{n}^{2}}\mathrm{E}\left[\left|U_{i}K^{-1}\bar{k}^{-1/2}\right|^{2}
\boldsymbol{1}\left\{\left|U_{i}K^{-1}\bar{k}^{-1/2}\right|\geq K^{-1}\bar{k}^{-1/2}e_{n}\right\}\right]\non\\
&\leq&\frac{nK^{2}\bar{k}}{e_{n}^{2}}\mathrm{E}\left(\left|U_{i}K^{-1}\bar{k}^{-1/2}\right|^{4}\right)^{1/2}
\mathrm{E}\left[\boldsymbol{1}\left\{\left|U_{i}K^{-1}\bar{k}^{-1/2}\right|\geq K^{-1}\bar{k}^{-1/2}e_{n}\right\}\right]^{1/2}\non\\
&=&o(1).
\ee
It follows that $\max_{1\leq j\leq N}\left|CV_{2n,1}(\boldsymbol{w}_{j})\right|=o_{p}(1)$ and thus $\mathrm{sup}_{\boldsymbol{w}\in\mathcal{W}}|CV_{2n,1}(\boldsymbol{w})|=o_{p}(1)$.

Applying the fact that $\boldsymbol{1}\{\varepsilon_{i}-b_{k}\leq s\}-\boldsymbol{1}\{\varepsilon_{i}-b_{k}\leq 0\}-F(s+b_{k}|\boldsymbol{x}_{i})+F(b_{k}|\boldsymbol{x}_{i})\leq2,$ Lemma \ref{lem:order},
Cauchy-Schwarz inequality and assumption $MK^{2}\bar{k}^{2}\log n/n\rightarrow0$, we have
\be\label{CQRCV:2n2}
CV_{2n,2}(\boldsymbol{w})&\leq& \frac{2}{n}\sum_{i=1}^{n}\sum_{k=1}^{K}\left|\sum_{m=1}^{M}w_{m}\left(1,\boldsymbol{x}_{i(m)}^{T}\right)\left(\hat{b}_{ki(m)}-b_{k(m)}^{*},\left(\hat{\bbeta}_{i(m)}-\bbeta_{(m)}^{*}\right)^{T}\right)^{T}\right|\non\\
&\leq&2 \max_{1\leq i\leq n}\max_{1\leq m\leq M}\sum_{k=1}^{K}\left\|\left(\hat{b}_{k(m)}-b_{ki(m)}^{*},\left(\hat{\bbeta}_{i(m)}-\bbeta_{(m)}^{*}\right)^{T}\right)\right\|\non\\
&&\max_{1\leq m\leq M}\frac{1}{n}\sum_{i=1}^{n}\left\|\left(1,\boldsymbol{x}_{i(m)}^{T}\right)\right\|\non\\
&\leq&2K^{1/2}\left[\max_{1\leq i\leq n}\max_{1\leq m\leq M}\sum_{k=1}^{K}\left\|\left(\hat{b}_{k(m)}-b_{ki(m)}^{*},\left(\hat{\bbeta}_{i(m)}-\bbeta_{(m)}^{*}\right)^{T}\right)\right\|^{2}\right]^{1/2}\non\\
&&\max_{1\leq m\leq M}\frac{1}{n}\sum_{i=1}^{n}\left\|\left(1,\boldsymbol{x}_{i(m)}^{T}\right)\right\|\non\\
&=&O_{p}\left(K^{1/2}\sqrt{n^{-1}K\bar{k}\log n}\right)O_{p}(\overline{k}^{1/2})=o_{p}(1).
\ee

Thirdly, we establish (iv). Decompose $$CV_{3n}(\boldsymbol{w})=CV_{3n,1}(\boldsymbol{w})+CV_{3n,2}(\boldsymbol{w})$$
where
\be
CV_{3n,1}(\boldsymbol{w})&=&\frac{1}{n}\sum_{i=1}^{n}\sum_{k=1}^{K}\left(\int_{0}^{\sum_{m=1}^{M}w_{m}b_{k(m)}^{*}-b_{k}+\sum_{m=1}^{M}w_{m}\boldsymbol{x}_{i(m)}^{T}\bbeta_{(m)}^{*}-\mu_{i}}\left\{F(s+b_{k}|\boldsymbol{x}_{i})-F(b_{k}|\boldsymbol{x}_{i})\right\}\mathrm{d}s\right.\non\\
&&\left.-\mathrm{E}_{\boldsymbol{x}_{i}}\left[\int_{0}^{\sum_{m=1}^{M}w_{m}b_{k(m)}^{*}-b_{k}+\sum_{m=1}^{M}w_{m}\boldsymbol{x}_{i(m)}^{T}\bbeta_{(m)}^{*}-\mu_{i}}\left\{F(s+b_{k}|\boldsymbol{x}_{i})-F(b_{k}|\boldsymbol{x}_{i})\right\}\mathrm{d}s\right]\right)\non\\
CV_{3n,2}(\boldsymbol{w})&=&\frac{1}{n}\sum_{i=1}^{n}\sum_{k=1}^{K}\left(\int_{\sum_{m=1}^{M}w_{m}b_{k(m)}^{*}-b_{k}+\sum_{m=1}^{M}w_{m}\boldsymbol{x}_{i(m)}^{T}\bbeta_{(m)}^{*}-\mu_{i}}^{\sum_{m=1}^{M}w_{m}\hat{b}_{ki(m)}-b_{k}+\sum_{m=1}^{M}w_{m}\boldsymbol{x}_{i(m)}^{T}\hat{\bbeta}_{i(m)}-\mu_{i}}\left\{F(s+b_{k}|\boldsymbol{x}_{i})-F(b_{k}|\boldsymbol{x}_{i})\right\}\mathrm{d}s\right.\non\\
&&\left.-\mathrm{E}_{\boldsymbol{x}_{i}}\left[\int_{\sum_{m=1}^{M}w_{m}b_{k(m)}^{*}-b_{k}+\sum_{m=1}^{M}w_{m}\boldsymbol{x}_{i(m)}^{T}\bbeta_{(m)}^{*}-\mu_{i}}^{\sum_{m=1}^{M}w_{m}\hat{b}_{ki(m)}-b_{k}+\sum_{m=1}^{M}w_{m}\boldsymbol{x}_{i(m)}^{T}\hat{\bbeta}_{i(m)}-\mu_{i}}\left\{F(s+b_{k}|\boldsymbol{x}_{i})-F(b_{k}|\boldsymbol{x}_{i})\right\}\mathrm{d}s\right]\right).\non
\ee
In view of the fact that $|F(s+b_{k}|\boldsymbol{x}_{i})-F(b_{k}|\boldsymbol{x}_{i})|\leq 1$, we have
$$CV_{3n,2}(\boldsymbol{w})\leq CV_{3n,21}(\boldsymbol{w})+CV_{3n,22}(\boldsymbol{w}),$$
where
\be
CV_{3n,21}(\boldsymbol{w})&=&\frac{1}{n}\sum_{i=1}^{n}\sum_{k=1}^{K}\left|\sum_{m=1}^{M}w_{m}\left(1,\boldsymbol{x}_{i(m)}^{T}\right)\left(\hat{b}_{ki(m)}-b_{k(m)}^{*},\left(\hat{\bbeta}_{i(m)}-\bbeta_{(m)}^{*}\right)^{T}\right)^{T}\right|\non
\ee
and
\be
CV_{3n,22}(\boldsymbol{w})&=&\frac{1}{n}\sum_{i=1}^{n}\mathrm{E}_{\boldsymbol{x}_{i}}\sum_{k=1}^{K}
\left|\sum_{m=1}^{M}w_{m}\left(1,\boldsymbol{x}_{i(m)}^{T}\right)\left(\hat{b}_{ki(m)}-b_{k(m)}^{*},\left(\hat{\bbeta}_{(m)}-\bbeta_{i(m)}^{*}\right)^{T}\right)^{T}\right|.\non
\ee
We have to show $$\sup_{w\in \mathcal{W}}|CV_{3n,1}(\boldsymbol{w})|=o_{p}(1),$$ $$\sup_{w\in \mathcal{W}}|CV_{3n,21}(\boldsymbol{w})|=o_{p}(1),$$
 $$\sup_{w\in \mathcal{W}}|CV_{3n,22}(\boldsymbol{w})|=o_{p}(1).$$
 Analogous to the proof of $\sup_{w\in \mathcal{W}}|CV_{2n,1}(\boldsymbol{w})|=o_{p}(1),$ we can show that $\sup_{w\in \mathcal{W}}|CV_{3n,1}(\boldsymbol{w})|=o_{p}(1).$
The equation $\sup_{w\in \mathcal{W}}|CV_{3n,21}(\boldsymbol{w})|=o_{p}(1)$ can be concluded from (\ref{CV:2n2}).
By the triangle and Cauchy-Schwarz inequalities, the fact $A^{T}BA\leq \lambda_{\max}(B)A^{T}A$ for
any symmetric matrix B, Lemma \ref{lem:order} and assumption $MK^{2}\bar{k}^{2}\log n/n\rightarrow0$, we have
\be
\sup_{w\in \mathcal{W}}CV_{3n,22}(\boldsymbol{w})
&\leq& \sup_{w\in \mathcal{W}}
\frac{1}{n}\sum_{i=1}^{n}\sum_{k=1}^{K}\sum_{m=1}^{M}w_{m}\mathrm{E}_{\boldsymbol{x}_{i}}\left|\left(1,\boldsymbol{x}_{i(m)}^{T}\right)\left(\hat{b}_{ki(m)}-b_{k(m)}^{*},\left(\hat{\bbeta}_{i(m)}-\bbeta_{(m)}^{*}\right)^{T}\right)^{T}\right|\nonumber\\
&\leq& \sup_{w\in \mathcal{W}}
\frac{1}{n}\sum_{i=1}^{n}\sum_{k=1}^{K}\sum_{m=1}^{M}w_{m}\left[\left(\hat{b}_{ki(m)}-b_{k(m)}^{*},\left(\hat{\bbeta}_{i(m)}-\bbeta_{(m)}^{*}\right)^{T}\right)\right.\non\\
&&\left.\mathrm{E}\left\{\left(1,\boldsymbol{x}_{i(m)}^{T}\right)\left(1,\boldsymbol{x}_{i(m)}^{T}\right)^{T}\right\}
\left(\hat{b}_{ki(m)}-b_{k(m)}^{*},\left(\hat{\bbeta}_{i(m)}-\bbeta_{(m)}^{*}\right)^{T}\right)^{T}\right]^{\frac{1}{2}}\non\\
&\leq& \max_{1\leq m\leq M}\left\{\lambda_{\max}\left(\mathrm{E}\left\{\left(1,\boldsymbol{x}_{i(m)}^{T}\right)\left(1,\boldsymbol{x}_{i(m)}^{T}\right)^{T}\right\}\right)\right\}^{1/2}\non\\
&&\max_{1\leq i\leq n}\max_{1\leq m\leq M}\sum_{k=1}^{K}\left\|\left(\hat{b}_{ki(m)}-b_{k(m)}^{*},\left(\hat{\bbeta}_{i(m)}-\bbeta_{(m)}^{*}\right)^{T}\right)^{T}\right\|\non\\
&\leq&\max_{1\leq m\leq M}\left\{\lambda_{\max}\left(\mathrm{E}\left\{\left(1,\boldsymbol{x}_{i(m)}^{T}\right)\left(1,\boldsymbol{x}_{i(m)}^{T}\right)^{T}\right\}\right)\right\}^{1/2}\non\\
&&\left\{K\max_{1\leq i\leq n}\max_{1\leq m\leq M}\sum_{k=1}^{K}\left\|\left(\hat{b}_{k(m)}-b_{ki(m)}^{*},\left(\hat{\bbeta}_{i(m)}-\bbeta_{(m)}^{*}\right)^{T}\right)^{T}\right\|^{2}\right\}^{1/2}\non\\
&=&O_{p}\left(\overline{k}^{1/2}\right)O_{p}\left(K^{1/2}\sqrt{n^{-1}K\bar{k}\log n}\right)=o_{p}(1)
\ee

(v)For $CV_{4n}(\boldsymbol{w}),$ noting that $|F(s|\boldsymbol{x}_{i})-F(0|\boldsymbol{x}_{i})|\leq1$ and by
the study of  $\sup_{w\in \mathcal{W}}CV_{3n,22}(\boldsymbol{w})$ we have
\be
&&\sup_{w\in \mathcal{W}}CV_{4n}(\boldsymbol{w})
\leq \sup_{w\in \mathcal{W}}
\frac{1}{n}\sum_{i=1}^{n}\sum_{k=1}^{K}\sum_{m=1}^{M}w_{m}\mathrm{E}_{\boldsymbol{x}_{i}}\left|\left(1,\boldsymbol{x}_{i(m)}^{T}\right)\left(\hat{b}_{ki(m)}-\hat{b}_{k(m)},\left(\hat{\bbeta}_{i(m)}-\hat{\bbeta}_{(m)}\right)^{T}\right)^{T}\right|
=o_{p}(1)\non.\ee

Thus, we finish the proof.
\bibliographystyle{asa}
\bibliography{CQR}
\end{sloppypar}
\end{document}